\documentclass[footnotebackref,authorcolumns,numberwithinsect]{no-lipics-v2022}
\pdfoutput=1
\usepackage{amssymb}
\usepackage{mathtools}
\usepackage{bm}
\usepackage{xspace}
\usepackage[T1]{fontenc}

\SetKwComment{Comment}{$\triangleright$\ \rm}{}
\SetKwInput{KwInput}{Input}
\SetKwInput{KwOutput}{Output}

\usetikzlibrary{arrows,arrows.meta,shapes.misc, decorations.pathreplacing,decorations.pathmorphing,calligraphy, patterns.meta}
\tikzset{dotmark/.style={circle,fill,inner sep=1.5pt}}
\tikzset{emptymark/.style={circle,draw,fill=white,inner sep=1.5pt}}
\tikzset{crossmark/.style={thick,inner sep=1.5pt}}

\def\ShowAuthNotes{1}
\ifnum\ShowAuthNotes=1
\newcommand{\authnote}[3]{\textcolor{#3}{[{\bf #1:} { {#2}}]}}
\else
\newcommand{\authnote}[3]{}
\fi

\newcommand{\A}{\mathcal{A}}
\newcommand{\Ab}{\mathbf{A}}

\def\twoheadleadsto{\tikz[baseline=(a.base)]{\draw[%
    decorate,decoration={zigzag,segment length=4, amplitude=.9},%
    ] (0,0) -- (.25, 0);%
    \draw[%
    -{Classical TikZ Rightarrow}.{Classical TikZ Rightarrow},%
    ] (.25, 0) -- (.4, 0);%
    \node (a) at (.4/2,-.03) {\phantom{\(\leadsto\)}};%
}}

\newcommand{\Oh}{\mathcal{O}}
\newcommand{\Ohtilde}{\tilde{\Oh}}
\newcommand{\Ohhat}{\hat{\Oh}}

\DeclareMathOperator*{\argmin}{arg\,min}
\DeclareMathOperator*{\argmax}{arg\,max}

\def\substr{\ensuremath \preccurlyeq}
\def\fragmentco#1#2{\bm{[}\,#1\,\bm{.\,.}\,#2\,\bm{)}}
\def\fragmentoc#1#2{\bm{(}\,#1\,\bm{.\,.}\,#2\,\bm{]}}
\def\fragmentoo#1#2{\bm{(}\,#1\,\bm{.\,.}\,#2\,\bm{)}}
\def\fragment#1#2{\bm{[}\,#1\,\bm{.\,.}\,#2\,\bm{]}}
\def\position#1{[#1]}

\def\threehalves{{}^3{\mskip -3.5mu/\mskip -3mu}_2\,}
\newcommand{\ceil}[1]{\lceil #1 \rceil}
\newcommand{\floor}[1]{\lfloor #1 \rfloor}
\newcommand{\rev}[1]{\overline{#1}}
\newcommand{\per}{\operatorname{per}}
\newcommand{\rot}{\operatorname{rot}}
\newcommand{\LZ}{\mathsf{LZ}}

\newcommand{\ed}{\delta_E}
\newcommand{\hd}{\delta_H}
\newcommand{\edp}[2]{{\delta_E}(#1,#2^*)}
\newcommand{\edl}[2]{{\delta_E}(#1,{}^*\!#2^*)}
\newcommand{\eds}[2]{{\delta_E}(#1,{}^*\!#2)}
\newcommand{\edal}[1]{\delta_E^{#1}}
\newcommand\selfed{\mathsf{self}\text{-}\ed}
\newcommand{\edshifted}[1]{\delta_E^{#1}} 
\newcommand{\edpx}[2]{{\delta_E}(#1,\circ#2)}
\newcommand{\edsx}[2]{{\delta_E}(#1,#2\circ)}
\newcommand{\OccE}{\mathrm{Occ}^E}
\newcommand{\OccH}{\mathrm{Occ}^H}
\newcommand{\Occ}{\mathrm{Occ}}

\newcommand{\last}{\mathrm{last}}

\newcommand{\gaped}{{\sc Gap Edit Distance}\xspace}
\newcommand{\shifted}{{\sc Shifted}\xspace}
\newcommand{\yes}{{\sc yes}\xspace}
\newcommand{\no}{{\sc no}\xspace}

\newcommand{\mA}{\mathcal{A}}

\newcommand{\mX}{\mathcal{X}}
\newcommand{\mY}{\mathcal{Y}}

\newcommand{\onto}{\twoheadleadsto}
\newcommand{\w}{\operatorname{w}}

\newcommand{\hi}{\hat{\imath}}

\newcommand{\hx}{\hat{x}}
\newcommand{\hy}{\hat{y}}

\newcommand{\sE}{\mathsf{E}}

\newcommand{\bG}{\mathbf{G}}

\newcommand{\bc}{\operatorname{bc}}

\newcommand\Tau{\mathcal{T}}

\newcommand{\pref}{\mathsf{pref}}
\newcommand{\suf}{\mathsf{suf}}
\newcommand{\mXpref}{\mX_{\mathsf{pref}}}
\newcommand{\mXsuf}{\mX_{\mathsf{suf}}}

\newcommand{\Z}{\mathbb{Z}}
\newcommand{\Zz}{\mathbb{Z}_{\ge 0}}

\def\alphav{128}
\def\betav{8}
\def\deltavN{3}
\def\deltavD{8}

\def\problembox#1{%
    \vspace{2mm}%
    \noindent\fbox{%
    \begin{minipage}{.985\linewidth}%
        #1
    \end{minipage}%
    }%
    \vspace{2mm}%
}

\newcommand{\defproblem}[3]{%
    \problembox{%
        \textsc{#1}\\
        {\bf{Input:}} #2  \\
        {\bf{Output:}} #3
    }%
}

\newcommand{\defproblemc}[4]{%
    \problembox{%
        \textsc{#1}\\
        {\bf{Input:}} #2  \\
        {\bf{Output:}} #3
    }%
}

\makeatletter
\renewenvironment{cases}{%
  \matrix@check\cases\env@cases
}{%
  \endarray\right.%
}
\def\env@cases{%
  \let\@ifnextchar\new@ifnextchar
  \left\lbrace
  \def\arraystretch{1.1}%
  \array{@{\;}c@{\quad}l@{}}%
}
\makeatother

\def\mid{\ensuremath :}

\def\emptyset{\varnothing}

\def\epsilon{\varepsilon}

\newcommand{\zero}{\mathtt{0}}
\newcommand{\one}{\mathtt{1}}

\def\pn#1{\textsc{#1}}
\def\PMwE{\pn{Pattern Matching with Edits}\xspace}
\def\PMwM{\pn{Pattern Matching with Mismatches}\xspace}

\def\pr#1{\ensuremath{\mathsf{Pr}\!\bm{\left[}\,#1\,\bm{\right]}}}

\def\GS{Grover's Search\xspace}

\title{On the Communication Complexity of Approximate~Pattern~Matching}

\author{Tomasz Kociumaka}{Max Planck Institute for Informatics\\Saarland Informatics
Campus\\Saarbrücken, Germany}{tomasz.kociumaka@mpi-inf.mpg.de}{https://orcid.org/0000-0002-2477-1702}{}
\author{Jakob Nogler}{ETH Zürich\\Zurich, Switzerland}{jnogler@ethz.ch}{https://orcid.org/0009-0002-7028-2595}{}
\author{Philip Wellnitz}{National Institute of Informatics\\The Graduate University for Advanced Studies, SOKENDAI\\Tokyo, Japan}{wellnitz@nii.ac.jp}{https://orcid.org/0000-0002-6482-8478}{}

\authorrunning{T. Kociumaka, J. Nogler, and P. Wellnitz}
\acknowledgements{The work of Philip Wellnitz was carried out at the Max Planck Institute for
Informatics. The work of Jakob Nogler was carried out mostly during a summer internship at the Max Planck Institute for Informatics.}

\nolinenumbers
\begin{document}
\pagenumbering{roman}

\maketitle
\begin{abstract}
    The decades-old \emph{Pattern Matching with Edits} problem, given a length-$n$ string \(T\) (the text), a length-$m$ string $P$ (the pattern), and a positive integer $k$ (the threshold), asks to list all fragments of \(T\) that are at edit distance at most \(k\) from~\(P\).
    The one-way communication complexity of this problem is the minimum amount of space needed to encode the answer so that it can be retrieved without accessing the input strings $P$ and $T$.

    The closely related Pattern Matching with Mismatches problem (defined in terms of the Hamming distance instead of the edit distance) is already well understood from the communication complexity perspective:
    Clifford, Kociumaka, and Porat [SODA 2019] proved that $\Omega(n/m \cdot k \log(m/k))$ bits are necessary and $\Oh(n/m \cdot k\log (m|\Sigma|/k))$ bits are sufficient; the upper bound allows encoding not only the occurrences of $P$ in $T$ with at most $k$ mismatches but also the substitutions needed to make each $k$-mismatch occurrence exact.

    Despite recent improvements in the running time [Charalampopoulos, Kociumaka, and Wellnitz; FOCS 2020 and 2022], the communication complexity of Pattern Matching with Edits remained unexplored, with a lower bound of $\Omega(n/m \cdot k\log(m/k))$ bits and an upper bound of $\Oh(n/m \cdot k^3\log m)$ bits stemming from previous research.
    In this work, we prove an upper bound of \(\Oh(n/m \cdot k \log^2 m)\) bits, thus establishing the optimal communication complexity up to logarithmic factors.
    We also show that \(\Oh(n/m \cdot k \log m \log (m|\Sigma|))\) bits allow encoding, for each $k$-error occurrence of $P$ in $T$, the shortest sequence of edits needed to make the occurrence exact.
    Our result further emphasizes the close relationship between Pattern Matching with Mismatches and Pattern Matching with Edits.

    We leverage the techniques behind our new result on the communication complexity to obtain quantum algorithms for Pattern Matching with Edits: we demonstrate a quantum algorithm that uses \(\Oh(n^{1+o(1)}/m \cdot \sqrt{km})\) queries and \(\Oh(n^{1+o(1)}/m \cdot (\sqrt{k}m + k^{3.5}))\) quantum time.
    Moreover, when determining the existence of at least one occurrence, the algorithm uses $\Oh(\sqrt{n^{1+o(1)}/m} \cdot \sqrt{km})$ queries and $\Oh(\sqrt{n^{1+o(1)}/m} \cdot (\sqrt{k}m + k^{3.5}))$ time.
    For both cases, we establish corresponding lower bounds to demonstrate that the query complexity is optimal up to sub-polynomial factors.
\end{abstract}

\clearpage

\thispagestyle{empty}
{\tableofcontents}
\newpage
\pagenumbering{arabic}

\section{Introduction} \label{sec:intro}

While a \emph{string} is perhaps the most basic way to represent data, this fact makes \emph{algorithms} working on strings more applicable and powerful.
Arguably, the very first thing to do with any kind of data is to find \emph{patterns} in it.
The \emph{\pn{Pattern Matching}} problem for strings and its variations are thus perhaps among the most fundamental problems that Theoretical Computer Science has to offer.

In this paper, we study the practically relevant \emph{\PMwE} variation~\cite{S80}.
Given a text string \(T\) of length \(n\), a pattern string \(P\) of length \(m\), and a threshold \(k\), the aim is to calculate the set \(\OccE_k(P,T)\) consisting of (the starting positions of) all the fragments of \(T\) that are at most \(k\) edits away from the pattern \(P\).
In other words, we compute the set of \(k\)-error occurrences of \(P\) in \(T\), more formally defined as
\[
    \OccE_k (P,T)\coloneqq \{i\in \fragment{0}{n} \mid \exists_{j\in \fragment{i}{n}} \ed(P,T\fragmentco{i}{j}\leq k)\},
\]
where we utilize the classical edit distance \(\ed\) (also referred to as the Levenshtein distance)~\cite{Levenshtein66} as the distance measure.
Here, an edit is either an insertion, a deletion, or a substitution of a single character.

\defproblemc{Pattern Matching with Edits}
{a pattern $P$ of length $m$, a text $T$ of length $n$, and an integer threshold $k>0$.}
{the set $\OccE_{k}(P,T)$.}
\par

Even though the \PMwE problem is almost as classical as it can get, with key algorithmic advances (from \(\Oh(mn)\) time down to \(\Oh(kn)\) time) dating back to the early and late 1980s~\cite{S80,LV88,LandauV89}, major progress has been made very
recently, when Charalampopoulos, Kociumaka, and Wellnitz~\cite{CKW22} obtained an $\Ohtilde(n+k^{3.5} n/m)$-time%
\footnote{The $\Ohtilde(\cdot)$ and $\Ohhat(\cdot)$ notations suppress factors poly-logarithmic and sub-polynomial in the input size $n+m$, respectively.}
solution and thereby broke through the 20-years-old barrier of the $\Oh(n+k^4 n/m)$-time algorithm by Cole and Hariharan~\cite{ColeH98}.
And the journey is far from over yet: the celebrated Orthogonal-Vectors-based lower bound for edit distance~\cite{bi18} rules out only  \(\Oh(n + k^{2-\Omega(1)} n/m)\)-time algorithms (also consult \cite{CKW22} for details), leaving open a wide area of uncharted algorithmic territory.
In this paper, we provide tools and structural insights that---we believe---will aid the exploration of the said territory.

We add to the picture a powerful new finding that sheds new light on the solution
structure of the \PMwE problem---similar structural results~\cite{bkw19,CKW20} form the backbone of the aforementioned breakthrough~\cite{CKW22}.
Specifically, we investigate how much space is needed to store all \(k\)-error occurrences of \(P\) in \(T\).
We know from~\cite{CKW20} that $\Oh(n/m\cdot k^3 \log m)$ bits suffice since one may report the occurrences as $\Oh(k^3)$ arithmetic progressions if $n=\Oh(m)$.
However, such complexity is likely incompatible with algorithms running faster than \(\Ohtilde(n + k^3n/m)\).
In this paper, we show that, indeed, $\Oh(n/m\cdot k\log^2 m)$ bits suffice to represent the set $\OccE_{k}(P,T)$.

Formally, the communication complexity of \PMwE measures the space needed to encode the output so that it can be retrieved without accessing the input.
We may interpret this setting as a two-party game: Alice is given an instance of the problem and constructs a message for Bob, who must be able to produce the output of the problem given Alice's message.
Since Bob does not have any input, it suffices to consider one-way single-round communication protocols.

\begin{restatable*}{mtheorem}{ccompl}\label{thm:ccompl}
    The \PMwE problem admits a one-way deterministic communication protocol
    that sends $\Oh(n/m\cdot k\log^2 m)$ bits.
    Within the same communication complexity, one can also encode the family of all fragments of $T\fragmentco{i}{j}$ that satisfy $\ed(P,T\fragmentco{i}{j})\le k$, as well as all optimal alignments $P\onto T\fragmentco{i}{j}$ for each of these fragments.
    Further, increasing the communication complexity to $\Oh(n/m\cdot k\log m \log(m|\Sigma|))$, where $\Sigma$ denotes the input alphabet, one can also retrieve the edit information for each optimal alignment.
\end{restatable*}

Observe that our encoding scheme suffices to retrieve not only the set \(\OccE_{k}(P,T)\) (which contains only starting positions of the $k$-error occurrences) but also the fragments of \(T\) with edit distance at most \(k\) from~\(P\).
In other words, it allows retrieving all pairs $0\le i \le j\le n$ such that \(\ed(P,T\fragmentco{i}{j}) \le k\).

We complement \cref{thm:ccompl} with a simple lower bound that shows that our result is tight (essentially up to one logarithmic factor).

\begin{restatable*}{mtheorem}{cclb}\label{thm:cclb}
    Fix integers $n,m,k$ such that $n/2 \ge m > k > 0$.
    Every communication protocol for the \PMwE problem uses
    $\Omega(n/m\cdot k\log(m/k))$ bits for $P=\zero^m$ and some $T\in \{\zero,\one\}^n$.
\end{restatable*}

Observe that our lower bound holds for the very simple case that the pattern is the all-zeros string and only the text contains nonzero characters.
In this case, the edit distance of the pattern and another string depends only on the length and the number of nonzero characters in the other string, and we can thus easily compute the edit distance in linear time.

\paragraph*{From Structural Insights to Better Algorithms: A Success Story}

Let us take a step back and review how structural results aided the development of approximate-pattern-matching algorithms in the recent past.

First, let us review the key insight of~\cite{CKW20} that led to the breakthrough of~\cite{CKW22}.
Crucially, the authors use that, for any pair of strings \(P\) and \(T\) with \(|T| \le \threehalves \cdot |P|\) and threshold \(k\ge 1\), either (a) \(P\) has at most \(\Oh(k^2)\) occurrences with at most \(k\) edits in \(T\), or (b) \(P\) and the relevant part of \(T\) are at edit distance \(\Oh(k)\) to periodic strings with the same period.
This insight helps as follows:
First, one may derive that, indeed, all \(k\)-error occurrences of \(P\) in \(T\) form \(\Oh(k^3)\) arithmetic progressions.
Second, it gives a blueprint for an algorithm: one has to tackle just two important cases: an easy \emph{nonperiodic} case, where \(P\) and \(T\) are highly unstructured and $k$-error occurrences are rare, and a not-so-easy periodic case, where \(P\) and \(T\) are highly repetitive and occurrences are frequent but appear in a structured manner.

The structural insights of \cite{CKW20} have found widespread other applications.
For example, they readily yielded algorithms for differentially private approximate pattern matching~\cite{S24}, approximate circular pattern matching problems~\cite{CKP21,CKP22,CPR24}, and they even played a key role in obtaining small-space algorithms for (online) language distance problems~\cite{BKS23}, among others.

Interestingly, an insight similar to the one of \cite{CKW20} was first obtained
in~\cite{bkw19} for the much easier problem of \PMwM (where we allow neither insertions
nor deletions) before being tightened and ported to \PMwE in~\cite{CKW20}.
Similarly, in this paper, we port a known communication complexity bound from \PMwM to
\PMwE; albeit with a much more involved proof.
As proved in~\cite{CKP19}, \PMwM problem admits a one-way deterministic $\Oh(k\log(m|\Sigma|/k))$-bit communication protocol.
While we discuss later (in the Technical Overview) the result of \cite{CKP19} as well as
the challenges in porting it to \PMwE, let us highlight here that
their result was crucial for obtaining an essentially optimal \emph{streaming} algorithm
for \PMwM.

Finally, let us discuss the future potential of our new structural results.
First, as a natural generalization of \cite{CKP19}, \(\Ohhat(k)\)-space algorithms for
\PMwE should be plausible
in the semi-streaming and (more ambitiously) streaming models, because \(\Ohhat(k)\)-size edit distance sketches have been developed in parallel to this work~\cite{KS24}.
Nevertheless, such results would also require $\Ohhat(k)$-space algorithms constructing sketches and recovering the edit distance from the two sketches, and \cite{KS24} does not provide such space-efficient algorithms.
Second, our result sheds more light on the structure of the non-periodic case of \cite{CKW20}: as it turns out, when relaxing the notion of periodicity even further, we obtain a periodic structure also for patterns with just a (sufficiently large) constant number of $k$-error occurrences.
This opens up a perspective for classical \PMwE algorithms that are even faster than \(\Ohtilde(n/m + k^3)\).

\paragraph*{Application of our Main Result: Quantum Pattern Matching with Edits}

As a fundamental problem, \PMwE has been studied in a plethora of settings, including the compressed setting~\cite{GS13,Tis14,BLRSSW15,CKW20}, the dynamic setting~\cite{CKW20}, and the streaming setting~\cite{Sta17,KPS21,BK23}, among
others.
However, so far, the \emph{quantum setting} remains vastly unexplored.
While quantum algorithms have been developed for Exact Pattern Matching~\cite{HV03},
\PMwM~\cite{JN23}, Longest Common Factor (Substring)~\cite{GS22,AJ22,JN23}, Lempel--Ziv
factorization~\cite{GJKT24}, as well as other fundamental string
problems~\cite{AGS19,WY20,ABIKKPSS20,BEGHS21,CKKSW22}, no quantum algorithm for \PMwE has been known so far.
The challenge posed by \PMwE, in comparison to \PMwM, arises already from the fact that, while the computation of Hamming distance between two strings can be easily accelerated in the quantum setting, the same is not straightforward for the edit distance case.
Only very recently, Gibney, Jin, Kociumaka, and Thankachan~\cite{GJKT24} demonstrated a quantum edit-distance algorithm with the optimal query complexity of $\Ohtilde(\sqrt{kn})$ and the time complexity of $\Ohtilde(\sqrt{kn}+k^2)$.

We follow the long line of research on quantum algorithms on strings and employ our new
structural results (combined with the structural results from \cite{CKW20}) to obtain the
following quantum algorithms for the \PMwE problem.

\begin{restatable*}{mtheorem}{qpmwe}\label{thm:qpmwe}
    Let $P$ denote a pattern of length $m$, let $T$ denote a text of length $n$, and let $k > 0$ denote an integer threshold.
    \begin{enumerate}
        \item There is a quantum algorithm that solves the \PMwE problem using $\hat{\mathcal{O}}(n/m \cdot \sqrt{km})$ queries and $\hat{\mathcal{O}}(n/m \cdot (\sqrt{k}m + k^{3.5}))$ time.
        \item There is a quantum algorithm deciding whether $\OccE_k(P,T)\ne \emptyset$ using $\hat{\mathcal{O}}(\sqrt{n/m} \cdot \sqrt{km})$ queries and $\hat{\mathcal{O}}(\sqrt{n/m} \cdot (\sqrt{k}m + k^{3.5}))$ time.
        \qedhere
    \end{enumerate}
\end{restatable*}
\medskip

Surprisingly, for $n = \Oh(m)$, we achieve the same query complexity as quantum algorithms
for computing the (bounded) edit distance~\cite{GJKT24} and even the bounded Hamming
distance of strings (a simple application of \GS yields an $\Ohtilde(\sqrt{kn})$ upper bound; a matching \(\Omega(\sqrt{kn})\) lower bound is also known~\cite{BBC01}).
While we did not optimize the time complexity of our algorithms (reasonably, one could expect a time complexity of \(\Ohtilde( n/m \cdot (\sqrt{km} + k^{3.5}) )\) based on our structural insights and~\cite{CKW22}), we show that our query complexity is essentially optimal by proving a matching lower bound.

\begin{restatable*}{mtheorem}{qpmwelb}\label{thm:qpmwelb}
    Let us fix integers $n \ge m > k > 0$.
    \begin{enumerate}
        \item Every quantum algorithm that solves the \PMwE problem uses $\Omega(n/m\cdot \sqrt{k(m-k)})$ queries for $P=\zero^m$ and some $T\in {\{\zero,\one\}}^n$.
        \item Every quantum algorithm that decides whether $\OccE_k(P,T)\ne \emptyset$ uses $\Omega(\sqrt{n/m}\cdot \sqrt{k(m-k)})$ queries for $P=\zero^m$ and some $T\in {\{\zero,\one\}}^n$.
        \qedhere
    \end{enumerate}
\end{restatable*}
\medskip

Again, our lower bounds hold already for the case when the pattern is the all-zeroes string and just the text contains nonzero entries.

\section{Technical Overview}\label{sec:overview}

In this section, we describe the technical contributions behind our positive results: \cref{thm:ccompl,thm:qpmwe}.
We assume that $n \le \threehalves m$ (if the text is longer, one may split the text into $\Oh(n/m)$ overlapping pieces of length \(\Oh(m)\) each) and that $k = o(m)$ (for $k=\Theta(m)$, our results trivialize).

\subsection{Communication Complexity of Pattern Matching with Mismatches}

Before we tackle \cref{thm:ccompl}, it is instructive to learn how to prove an analogous
result for \PMwM.
Compared to the original approach of Clifford, Kociumaka, and Porat~\cite{CKP19}, we neither optimize logarithmic factors nor provide an efficient decoding algorithm; this enables significant simplifications.
Recall that our goal is to encode the set $\OccH_k(P,T)$, which is the Hamming-distance analog of the set $\OccE_k(P,T)$.
Formally, we set
\[\OccH_k(P,T) \coloneqq \{i\in \fragment{0}{n-m} : \hd(P, T\fragmentco{i}{i+m})\le k\}.\]
Without loss of generality, we assume that $\{0,n-m\}\subseteq \OccH_k(P,T)$, that is, $P$ has $k$-mismatch occurrences both as a prefix and as a suffix of $T$.
Otherwise, either we have $\OccH_k(P,T)=\emptyset$ (which can be encoded trivially), or we can crop $T$ by removing the characters to the left of the leftmost $k$-mismatch occurrence and to the right of the rightmost $k$-mismatch occurrence.

\begin{figure*}[t!]
    \renewcommand\tabularxcolumn[1]{m{#1}}
    \centering
    \begin{tabularx}{\linewidth}{*{3}{>{\centering\arraybackslash}X}}
    \includegraphics[scale=1.5]{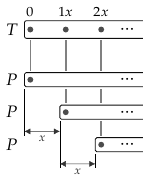}
    &
    \includegraphics[scale=1.3]{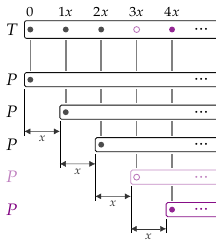}
    &
    \includegraphics[scale=1.5]{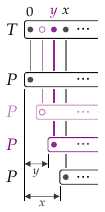}
    \\[-2ex]
    \begin{subfigure}[t]{\linewidth}
        \caption{The pattern \(P\) occurs in \(T\) starting at the positions \(0\), \(x\), and
            \(2x\); these starting positions form the arithmetic progression \((ix)_{0 \le
            i \le 2}\).}\label{fig1a}
    \end{subfigure}
    &
    \begin{subfigure}[t]{\linewidth}
        \caption{%
            Suppose that we were to identify an additional occurrence of \(P\) in
            \(T\) starting at position \(4x\).
            Now, since occurrences start at \(0, 2x\), and \(4x\) (which
            in particular implies that \(T\fragmentco{0}{2x + |P|} = T\fragmentco{2x}{4x +
            |P|}\)), as well as at position
            \(x\), we directly obtain that
            there is also an occurrence that starts at position \(3x\) in \(T\);
            which means that the arithmetic progression from \cref{fig1a} is extended to
            \((ix)_{0 \le i \le 4}\).
            More generally, one may prove that any additional occurrence at a position
            \(ix\) extends the existing arithmetic progression in a similar fashion.
        }\label{fig1b}
    \end{subfigure}
    &
    \begin{subfigure}[t]{\linewidth}
        \caption{
            Suppose that we were to identify an additional occurrence of \(P\) in
            \(T\) starting at position \(0 < y < x\).
            Now, similarly to \cref{fig1b}, we can argue that there is also an occurrence
            that starts at every position of the form \(i \text{gcd}(x, y)\) (this is a
            consequence of the famous Periodicity Lemma due to~\cite{FW65}; see
            \cref{lem:perlemma})---again an arithmetic progression.\\
            Crucially, the difference of the arithmetic progression obtained in this
            fashion decreased by a factor of at least two compared to the initial
            arithmetic progression.
        }\label{fig1c}
    \end{subfigure}
    \end{tabularx}
    \caption{
        The structure of occurrences of exact pattern matching is easy: either all
        exact occurrences of \(P\) in \(T\) form an arithmetic progression or there is
        just one such occurrence (which we may also view as a degenerate arithmetic
        progression).\\
        Depicted is a text \(T\) and exact occurrences starting at the positions denoted
        above the text; we may assume that there is an occurrence that starts at position
         \(0\) and that there is an occurrence that ends at position \(|T|-1\).
    }\label{fig1}
\end{figure*}

\subparagraph*{Encoding All \(k\)-Mismatch Occurrences.}
First, if $k=0$, as a famous consequence of the Periodicity Lemma~\cite{FW65}, the set
$\OccH_0(P,T)=\Occ(P,T)$ is guaranteed to form a single arithmetic progression (recall
that $n \le \threehalves m$ and see \cref{fct:periodicity}), and thus it can be encoded
using $\Oh(\log m)$ bits.
Consult \cref{fig1} for a visualization of an example.

If $k>0$, the set $\OccH_k(P,T)$ does not necessarily form an arithmetic progression.
Still, we may consider the smallest arithmetic progression that contains $\OccH_k(P,T)$ as a subset.
Since $0 \in \OccH_k(P,T)$, the difference of this progression can be expressed as $g\coloneqq\gcd(\OccH_k(P,T))$.

A crucial property of the $\gcd(\cdot)$ function is that, as we add elements to a set maintaining its greatest common divisor $g$, each insertion either does not change $g$ (if the inserted element is already a multiple of $g$) or results in the value $g$ decreasing by a factor of at least $2$ (otherwise).
Consequently, there is a set $\{0,n-m\}\subseteq S\subseteq \OccH_k(P,T)$ of size $|S| = \Oh(\log m)$ such that $\gcd(S) = \gcd(\OccH_k(P,T)) = g$.

The encoding that Alice produces consists of the set $S$ with each $k$-mismatch occurrences $i\in S$ augmented with the \emph{mismatch information} for $P$ and $T\fragmentco{i}{i+m}$, that is, a set
\[\{(j, P\position{j},T\position{i+j}) : j\in \fragmentco{0}{m} \text{ such that }P\position{j}\ne T\position{i+j}\}.\]
For a single $k$-mismatch occurrence, the mismatch information can be encoded in $\Oh(k\log(m|\Sigma|))$ bits, where $\Sigma$ is the alphabet of $P$ and $T$.
Due to $|S|=\Oh(\log m)$, the overall encoding size is $\Oh(k\log m \log(m|\Sigma|))$.

\begin{figure*}[tp]
    \renewcommand\tabularxcolumn[1]{m{#1}}
    \centering
    \begin{tabularx}{\linewidth}{*{2}{>{\centering\arraybackslash}X}}
        \includegraphics[scale=1.16]{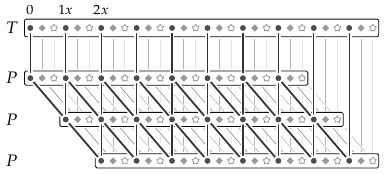}
        &
        \includegraphics[scale=1.17]{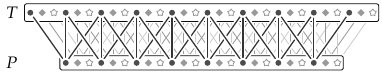}
        \\[-3ex]
        \begin{subfigure}[t]{\linewidth}
            \caption{Compare \cref{fig1a}.
                So far, we identified three occurrences of \(P\) in \(T\); each
                occurrence is an exact occurrence.
                Correspondingly,
                we have \(S = \left\{ (0, \emptyset),  (x, \emptyset),  (2x, \emptyset) \right\} \).\\
                With this set \(S\), we obtain three different black components, which we
                depict with a circle, a diamond, or a star.
            }\label{fig2a}
        \end{subfigure}
        &
        \begin{subfigure}[t]{\linewidth}
            \caption{The graph \(\bG_{S}\) that corresponds to \cref{fig2a}: observe how
            we collapsed the different patterns from \cref{fig2a} into a single pattern
        \(P\).\\
                In the example, we have three black components, that is, \(\bc(\bG_S) = 3\).
    }\label{fig2b}
        \end{subfigure}
        \\[1.5ex]
        \includegraphics[scale=1.17]{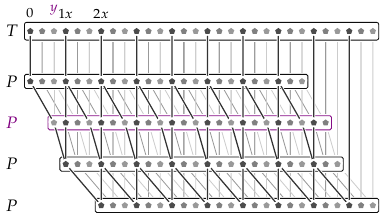}
        &
        \includegraphics[scale=1.17]{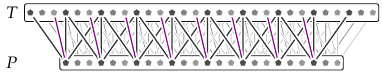}
        \\[-3ex]
        \begin{subfigure}[t]{\linewidth}
            \caption{
            Suppose that we were to identify an additional occurrence of \(P\) in
            \(T\) starting at position \(0 < y < x\) (highlighted in purple).
            From \cref{fig1c},
            we know how the set of all occurrences changes, but---and this is the crucial
            point--- we do not add all of these implicitly found occurrences to \(S\), but
            just \(y\).\\
            In our example, we observe that the black components collapse into a single
            black component, which we depict with a cloud.
        }\label{fig2c}
        \end{subfigure}
        &
        \begin{subfigure}[t]{\linewidth}
            \caption{The graph \(\bG_{S}\) that corresponds to \cref{fig2c}: observe how
            we collapsed the different patterns from \cref{fig2c} into a single pattern \(P\).
            Highlighted in purple are some of the edges that we added due to the new occurrence
        that we added to \(S\).\\
                In the example, we have one black components, that is, \(\bc(\bG_S) = 1\).
    }
            \label{fig2d}
        \end{subfigure}
        \\[1.5ex]
        \multicolumn{2}{>{\hsize=\dimexpr2\hsize+2\tabcolsep+\arrayrulewidth\relax}X}{%
            \vspace*{1.5ex}%
            \centering
            \includegraphics[scale=1.2]{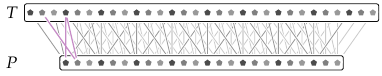}%
        }
        \\[-3ex]
        \multicolumn{2}{>{\hsize=\dimexpr2\hsize+2\tabcolsep+\arrayrulewidth\relax}X}{%
            \centering
            \begin{subfigure}[t]{2\linewidth}
                \caption{Recovering an occurrence in \(\bG_S\) from \cref{fig2d} that starts at position
                \(\text{gcd}(x, y)\), illustrated for the first character of the pattern.
            }
            \end{subfigure}
        }
    \end{tabularx}
    \caption{Compare \cref{fig1}: we fully understand the easy structure of exact pattern
        matching. In this figure, we reinterpret our knowledge in terms of the encoding
        scheme of Alice for \PMwM (in particular we show just
        the occurrences included in the set \(S\)) and
        showcase how the corresponding graph \(\bG_{S}\) and its black components
        evolve.\\
        We connect the
        same positions in \(P\), as well as pairs of positions that are aligned by an
        occurrence of \(P\) in \(T\).
        As there are no mismatches, every such line implies that the connected characters
        are equal.\\
        For each connected component of the resulting graph (a black component),
        we know that all involved positions in \(P\) and \(T\) must have the same symbol.
        For illustrative purposes,
        we assume that \(x = 3\) and
        we replace each character of a black component with a sentinel character (unique
        to that component),
        that is, we depict the strings \(P^\#\) and \(T^\#\).
    }\label{fig2}
\end{figure*}
\begin{figure*}[tp]
    \renewcommand\tabularxcolumn[1]{m{#1}}
    \centering
    \begin{tabularx}{\linewidth}{*{2}{>{\centering\arraybackslash}X}}
        \includegraphics[scale=1.16]{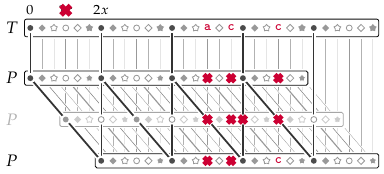}
        &
        \includegraphics[scale=1.17]{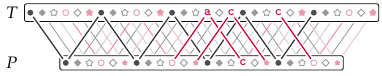}
        \\[-3ex]
        \begin{subfigure}[t]{\linewidth}
            \caption{Compare \cref{fig2a}.
                We depict  mismatched characters in an alignment of \(P\) to \(T\)
                by placing a cross over the corresponding
                character in \(P\).\\
                If we allow at most  \(3\) mismatches, we now do not have an occurrence
                starting at position \(x\) anymore; hence we obtain six black components.
        }\label{fig3a}
        \end{subfigure}
        &
        \begin{subfigure}[t]{\linewidth}
            \caption{The graph \(\bG_S\) that correspond to \cref{fig3a}.
                We make explicit characters that are different from the ``default''
                character of a component; the corresponding red edges (that are
                highlighted) are exactly the mismatch information that is stored in \(S\).
                For the remaining edges, the color depicts the color of the connected
                component that they belong to.\\
                In the example, we have four black components, that is, \(\bc(\bG_S) =
                4\).\\
                (Observe that contrary to what the image might make you believe, not
                every ``non-default'' character needs to end in a highlighted red edge.)
            }
        \end{subfigure}
        \\[1.5ex]
        \includegraphics[scale=1.17]{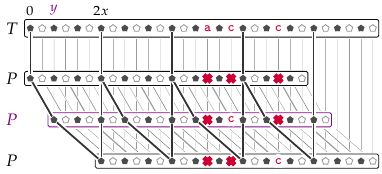}
        &
        \includegraphics[scale=1.17]{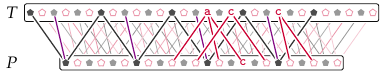}
        \\[-3ex]
        \begin{subfigure}[t]{\linewidth}
            \caption{
                Compare \cref{fig2c}.
            We are still able to identify an additional occurrence of \(P\) in
            \(T\) starting at position \(0 < y < x\) (highlighted in purple).
            Now, as before, connected components of \(\bG_S\) merge; this time, this also
            means that some characters that were previously part of a black component now
            become part of a red component (but crucially never vice-versa).\\
            In the example, this means that we now have just a single black component,
            that is, \(\bc(\bG_S) = 1\).
        }\label{fig3c}
        \end{subfigure}
        &
        \begin{subfigure}[t]{\linewidth}
            \caption{The graph \(\bG_S\) for the situation in \cref{fig3c}.
Again,                we make explicit characters that are different from the ``default''
                character of a component; the corresponding red edges (that are
                highlighted) are exactly the mismatch information that is stored in \(S\).
                For the remaining edges, the color depicts the color of the connected
                component that they belong to (where purple highlights some of the black
                edges added due to the new occurrence).
            }
        \end{subfigure}
        \\[1.5ex]
        \multicolumn{2}{>{\hsize=\dimexpr2\hsize+2\tabcolsep+\arrayrulewidth\relax}X}{%
            \vspace*{1.5ex}%
            \centering
            \includegraphics[scale=1.2]{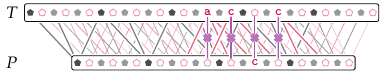}%
        }
        \\[-3ex]
        \multicolumn{2}{>{\hsize=\dimexpr2\hsize+2\tabcolsep+\arrayrulewidth\relax}X}{%
            \centering
            \begin{subfigure}[t]{2\linewidth}
                \caption{
                    Checking for an occurrence at position \(2 \text{gcd}(x, y)\) (which
                    would be an occurrence were it not for mismatched characters).
                    We check two things, first that the
                    black component
                    aligns; and second, for the red component where we know all
                    characters, we compute
                    exactly the Hamming distance (which is \(4\) in the example, meaning that there is
                    no occurrence at the position in question).
                }
            \end{subfigure}
        }
    \end{tabularx}
    \caption{
        Compared to \cref{fig2}, we now have characters in \(P\) and \(T\) that mismatch.
        Again, we showcase how the corresponding graph \(\bG_{S}\) and its black components
        evolve; in the example, we allow for up to \(k = 3\) mismatches.\\
        Again, for illustrative purposes,
        we assume that \(x = 3\) and
        we replace each character of a black component with a sentinel character (unique
        to that component),
        that is, we depict the strings \(P^\#\) and \(T^\#\).
    }\label{fig3}
\end{figure*}

\subparagraph*{Recovering the \(k\)-Mismatch Occurrences.}
It remains to argue that the encoding is sufficient for Bob to recover $\OccH_k(P,T)$.
To that end, consider a graph $\bG_S$ whose vertices correspond to characters in $P$ and $T$.
For every $i\in S$ and $j\in \fragmentco{0}{m}$, the graph $\bG_S$ contains an edge between $P\position{j}$ and $T\position{i+j}$.
If $P\position{j}=T\position{i+j}$, then the edge is \emph{black}; otherwise, the edge is \emph{red} and annotated with the values $P\position{j}\ne T\position{i+j}$.
Observe that Bob can reconstruct $\bG_S$ using the set $S$ and the mismatch information for the $k$-mismatch occurrences at positions $i\in S$.

Next, we focus on the connected components of the graph $\bG_S$.
We say that a component is black if all of its edges are black and red if it contains at least one red edge.
Observe that Bob can reconstruct the values of all characters in red components: the annotations already provide this information for vertices incident to red edges, and since black edges connect matching characters, the values can be propagated along black edges, ultimately covering all vertices in red components.
The values of characters in black components remain unknown, but each black component is guaranteed to be \emph{uniform}, meaning that every two characters in a single black component match.

The last crucial observation is that the connected components of $\bG_S$ are very structured: for every remainder $c\in \fragmentco{0}{g}$ modulo $g$, there is a connected component consisting of all vertices $P\position{i}$ and $T\position{i}$ with $i\equiv_g c$.
This can be seen as a consequence of the Periodicity Lemma~\cite{FW65} applied to strings obtained from $P$ and $T$ by replacing each character with a unique identifier of its connected component.
Consult \cref{fig2} for an illustration of an example for the special case
if there are no mismatches and consult \cref{fig3} for a visualization of an example with
mismatches.

\subparagraph*{Testing if an Occurrences Starts at a Given Position.}
With these ingredients, we are now ready to explain how Bob tests whether a given position $i\in \fragment{0}{n-m}$ belongs to $\OccH_k(P,T)$.
If $i$ is not divisible by $g$, then for sure $i\notin \OccH_k(P,T)$.
Otherwise, for every $j\in \fragmentco{0}{m}$, the characters $P\position{j}$ and $T\position{i+j}$ belong to the same connected component.
If this component is red, then Bob knows the values of $P\position{j}$ and $T\position{i+j}$, so he can simply check if the characters match.
Otherwise, the component is black, meaning that $P\position{j}$ and $T\position{i+j}$ are guaranteed to match.
As a result, Bob can compute the Hamming distance $\hd(P, T\fragmentco{i}{i+m})$ and check if it does not exceed $k$.
In either case (as long as $i$ is divisible by $g$), he can even retrieve the underlying mismatch information.

A convenient way of capturing Bob's knowledge about $P$ and $T$ is to construct auxiliary strings $P^\#$ and $T^\#$ obtained from $P$ and $T$, respectively, by replacing all characters in each black component with a sentinel character (unique for the component).
Then, $\OccH_k(P,T)=\OccH_k(P^{\#},T^{\#})$ and the mismatch information is preserved for the $k$-mismatch occurrences.

\subsection{The Communication Complexity of Pattern Matching with Edits}\label{sec:ccov}

On a very high level, our encoding for \PMwE builds upon the
approach for \PMwM presented above:
\begin{itemize}
    \item Alice still constructs an appropriate size-$\Oh(\log m)$ set $S$ of $k$-error occurrences of $P$ in $T$, including a prefix and a suffix of $T$.
    \item Bob uses the edit information for the occurrences in $S$ to construct a graph $\bG_S$ and strings $P^{\#}$ and~$T^{\#}$, obtained from $P$ and $T$ by replacing characters in some components with sentinel characters so that $\OccE_k(P,T)=\OccE_k(P^{\#},T^{\#})$.
\end{itemize}
At the same time, the edit distance brings new challenges, so we also deviate from the original strategy:
\begin{itemize}
    \item Connected components of $\bG_S$ do not have a simple periodic structure, so $g=\gcd(S)$ loses its meaning. Nevertheless, we prove that black components still behave in a structured way, and thus the number of black components, denoted $\bc(\bG_S)$, can be used instead.
    \item The value $\bc(\bG_S)$ is not as easy to compute as $\gcd(S)$, so we grow the set $S\subseteq \OccE_k(P,T)$ iteratively.
    In each step, either we add a single $k$-error occurrence so that $\bc(\bG_S)$ decreases by a factor of at least $2$, or we realize that the information related to the alignments already included in $S$ suffices to retrieve all $k$-error occurrences of $P$ in $T$.
    \item Once this process terminates, there may unfortunately remain $k$-error occurrences whose addition to $S$ would decrease $\bc(\bG_S)$---yet, only very slightly.
    In other words, such $k$-error occurrences generally obey the structure of black components, but may occasionally violate it.
    We need to understand where the latter may happen and learn the characters behind the black components involved so that they are not masked out in $P^\#$ and $T^\#$.
    This is the most involved part of our construction, where we use recent insights relating edit distance to compressibility~\cite{CKW23,GJKT24} and store compressed representations of certain fragments of $T$.
\end{itemize}

\subsubsection{General Setup}

Technically, the set $S$ that Alice constructs contains, instead of $k$-error occurrences $T\fragmentco{t}{t'}$, specific alignments $P\onto T\fragmentco{t}{t'}$ of cost at most $k$.
Every such alignment describes a sequence of (at most~$k$) edits that transform $P$ onto $T\fragmentco{t}{t'}$; see \cref{def:alignment}.
In the message that Alice constructs, each alignment is augmented with \emph{edit information}, which specifies the positions and values of the edited characters; see \cref{def:edinfo}.
For a single alignment of cost $k$, this information takes $\Oh(k\log(m|\Sigma|))$ bits, where $\Sigma$ is the alphabet of $P$ and $T$.

Just like for \PMwM, we can assume without loss of generality that $P$ has $k$-error occurrences both as a prefix or as a suffix of $T$.
Consequently, we always assume that $S$ contains an alignment $\mXpref$ that aligns $P$ with a prefix of $T$ and an alignment $\mXsuf$ that aligns $P$ with a suffix of~$T$.

The graph $\bG_S$ is constructed similarly as for mismatches: the vertices are characters of $P$ and $T$, whereas the edges correspond to pairs of characters aligned by any alignment in $S$.
Matched pairs of characters correspond to black edges, whereas substitutions correspond to red edges, annotated with the values of the mismatching characters.
Insertions and deletions are also captured by red edges; see \cref{def:bg} for details.

Again, we classify connected components of $\bG_S$ into black (with black edges only) and red (with at least one red edge).
Observe that Bob can reconstruct the graph $\bG_S$ and the values of all characters in red components and that black components remain \emph{uniform}, that is, every two characters in a single black component match.
Consult \cref{fig4} for a visualization of an example.

Finally, we define $\bc(\bG_S)$ to be the number of black components in $\bG_S$.
If $\bc(\bG_S)=0$, then Bob can reconstruct the whole strings $P$ and $T$, so we henceforth assume $\bc(\bG_S)>0$.

\subparagraph*{First Insights into \(\bG_S\).}
Our first notable insight is that black components exhibit periodic structure.
To that end,
write $P_{|S}$ for the subsequence of \(P\) that contains all characters of
\(P\) that are contained in a black component in \(\bG_S\)
and
write $T_{|S}$ for the subsequence of \(T\) that contains all characters of
\(T\) that are contained in a black component in \(\bG_S\).
Then, for every $c\in \fragmentco{0}{\bc(\bG_S)}$, there is
a component consisting of all characters $P_{|S}\position{i}$ and $T_{|S}\position{i}$
such that $i\equiv_{\bc(\bG_S)} c$; for a formal statement and proof, consult \cref{lem:periodicity}.
Also consult \cref{fig4c} for an illustration of an example.

Next, we denote the positions in $P$ and $T$ of the subsequent characters of $P_{|S}$ and
$T_{|S}$ belonging to a specific component $c\in \fragmentco{0}{\bc(\bG_S)}$ as
$\pi_0^c,\pi_1^c,\ldots$ and $\tau_0^c,\tau_1^c,\ldots$, respectively; see
\cref{def:pitau}.
The characterization of the black components presented above implies that $\pi_{j}^c < \pi_{j'}^{c'}$ if and only if either $j<j'$ or $j=j'$ and $c<c'$ (analogously for $\tau_j^c < \tau_{j'}^{c'}$).
We assume that the $c$th black component contains $m_c$ characters of $P$ and $n_c$ characters in $T$; note that $m_c\in \{m_0,m_0-1\}$ and $n_c\in \{n_0,n_0-1\}$.

\begin{figure*}[t!]
    \renewcommand\tabularxcolumn[1]{m{#1}}
    \centering
    \begin{tabularx}{\linewidth}{*{2}{>{\centering\arraybackslash}X}}
        \includegraphics[scale=1.16]{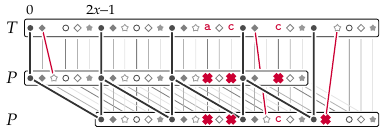}
        &
        \includegraphics[scale=1.17]{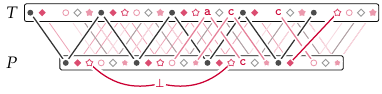}
        \\[-3ex]
        \begin{subfigure}[t]{\linewidth}
            \caption{Compare \cref{fig3a}. In addition to mismatched characters, we now
                also have missing characters in \(P\) and \(T\) (depicted by a white
                space). Further, as alignments for occurrences are no longer unique, we
                have to choose an alignment for each occurrence in the set \(S\) (which
                can fortunately be stored efficiently).
        }\label{fig4a}
        \end{subfigure}
        &
        \begin{subfigure}[t]{\linewidth}
            \caption{The graph \(\bG_S\) that corresponds to the situation in
            \cref{fig3a}. Observe that now, we also have a sentinel vertex \(\bot\) to
            represent that an insertion or deletion happened.
            Observe further that due to insertions and deletions,
            the last empty star character of \(T\) now belongs to the
            component of filled diamonds.\\
            In the example, we have two black components, that is, \(\bc(\bG_S) = 2\).
        }
        \end{subfigure}
        \\[3ex]
        \multicolumn{2}{>{\hsize=\dimexpr2\hsize+2\tabcolsep+\arrayrulewidth\relax}X}{%
            \vspace*{3ex}%
            \centering
            \includegraphics[scale=1.16]{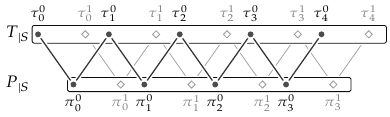}
        }
        \\[-3ex]
        \multicolumn{2}{>{\hsize=\dimexpr2\hsize+2\tabcolsep+\arrayrulewidth\relax}X}{%
            \centering
            \begin{subfigure}[t]{2\linewidth}
                \caption{An illustration of the additional notation that we use to analyze
                    \(\bG_S\).
                    Removing every character involved in
                    a red component, we obtain the strings \(T_{|S}\) and \(P_{|S}\).
                    For each black component, we
                    number the corresponding characters in \(P\) and \(T\) from left to
                    right.
            }\label{fig4c}
            \end{subfigure}
        }
    \end{tabularx}
    \caption{Compare \cref{fig2,fig3}. In addition to mismatches, we now also allow
    character insertions or deletions. In the example, we depict occurrence with at most
\(k = 4\) edits.}\label{fig4}
\end{figure*}

\subsubsection{Extra Information to Capture Close Alignments}
By definition of the graph $\bG_S$, the alignments in $S$ obey the structure of the black components.
Specifically, for every $\mX \in S$, there is a shift $i\in \fragment{0}{n_0-m_0}$ such that $\mX$ matches $P\position{\pi_j^c}$ with $T\position{\tau_{i+j}^c}$ for every $c\in \fragmentco{0}{\bc(\bG_S)}$ and $j\in \fragmentco{0}{m_c}$.
The quasi-periodic structure of $P$ and $T$ suggests that we should expect further shifts $i\in \fragment{0}{n_0-m_0}$ with
low-cost alignments matching $P\position{\pi_j^c}$ with $T\position{\tau_{i+j}^c}$ for every $c\in \fragmentco{0}{\bc(\bG_S)}$ and $j\in \fragmentco{0}{m_c}$.
Unfortunately, even if an optimum alignment $\mX:P\onto T\fragmentco{t}{t'}$ matches $P\position{\pi_j^c}$ with $T\position{\tau_{i+j}^c}$, there is no guarantee that it also matches $P\position{\pi_{j'}^{c'}}$ with $T\position{\tau_{i+j'}^{c'}}$ for other values $c'\in \fragmentco{0}{\bc(\bG_S)}$ and $j'\in \fragmentco{0}{m_{c'}}$.
Even worse, it is possible that no optimal alignment $P\fragmentco{\pi_j^{c-1}}{\pi_j^{c+1}}\onto T\fragmentco{\tau_{i+j}^{c-1}}{\tau_{i+j}^{c+1}}$ matches $P\position{\pi_j^c}$ with $T\position{\tau_{i+j}^c}$.
The reason behind this phenomenon is that the composition of optimal edit-distance alignments is not necessarily optimal (more generally, the edit information of optimal alignments $X\onto Y$ and $Y\onto Z$ is insufficient to recover $\ed(X,Z)$).

In these circumstances, our workaround is to identify a set $C_S\subseteq \fragmentco{0}{\bc(\bG_S)}$ such that the underlying characters can be encoded in $\Ohtilde(k|S|)$ space and every alignment $\mX:P \onto T\fragmentco{t}{t'}$ that we need to capture matches $P\position{\pi_j^c}$ with $T\position{\tau_{i+j}^c}$ for every $c\in \fragmentco{0}{\bc(\bG_S)}\setminus C_S$ and $j\in \fragmentco{0}{m_c}$.
For this, we investigate how an optimal alignment $\mX:P\onto T\fragmentco{t}{t'}$ may differ from a canonical alignment $\mA : P\onto T\fragmentco{t}{t'}$ that matches $P\position{\pi_j^c}$ with $T\position{\tau_{i+j}^c}$ for all $c\in \fragmentco{0}{\bc(\bG_S)}$ and $j\in \fragmentco{0}{m_c}$.
Following recent insights from~\cite{CKW23,GJKT24} (see~\cref{prp:edimpliesselfed}), we observe that the fragments of $P$ on which $\mA$ and $\mX$ are disjoint can be compressed into $\Oh(\ed^{\A}(P,T\fragmentco{t}{t'}))$ space (using Lempel--Ziv factorization~\cite{DBLP:journals/tit/ZivL77}, for example).
Moreover, the compressed size of each of these fragments is at most proportional to the cost of $\mA$ on the fragment.
Consequently, our goal is to understand where $\mA$ makes edits and learn all the fragments of $P$ (and $T$) with a sufficiently high density of edits compared to the compressed size.
Due to the quasi-periodic nature of $P$ and $T$, for each  $c\in \fragmentco{0}{\bc(\bG_S)}$, all characters in the $c$th black component are equal to $T\position{\tau_0^c}$, so we can focus on learning fragments of $T\fragment{\tau_0^0}{\tau_0^{\bc(\bG_S)-1}}$.

The bulk of the alignment $\mA$ can be decomposed into pieces that align $P\fragmentco{\pi_j^c}{\pi_j^{c+1}}$ onto $T\fragmentco{\tau_{i+j}^c}{\tau_{i+j}^{c+1}}$.
In \cref{lem:exfuncover}, we prove that $\ed(P\fragmentco{\pi_j^c}{\pi_j^{c+1}},T\fragmentco{\tau_{i+j}^c}{\tau_{i+j}^{c+1}})\le \w_S(c)$, where $\w_S(c)$ is the total cost incurred by alignments in $S$ on all fragments $P\fragmentco{\pi_{j'}^c}{\pi_{j'}^{c+1}}$ for $j'\in \fragmentco{0}{m_c}$.
Intuitively, this is because the path from $P\position{\pi_j^c}$ to $T\position{\tau_{i+j}^c}$ in $\bG_S$ allows obtaining an alignment $P\fragmentco{\pi_j^c}{\pi_j^{c+1}}\onto T\fragmentco{\tau_{i+j}^c}{\tau_{i+j}^{c+1}}$ as a composition of pieces of alignments in $S$ and their inverses.
Every component $c\in \fragmentco{0}{\bc(\bG_S)}$ uses distinct pieces, so the total weight $w\coloneqq \sum_c \w_S(c)$ does not exceed $k\cdot |S|$.

The weight function $\w_S(c)$ governs which characters of $T\fragment{\tau_0^0}{\tau_0^{\bc(\bG_S)-1}}$ need to be learned.
We formalize this with a notion of a \emph{period cover} $C_S\subseteq \fragmentco{0}{\bc(\bG_S)}$; see \cref{def:periodcover_alt}.
Most importantly, we require that $\fragment{a}{b}\subseteq C_S$ holds whenever the compressed size of $T\fragment{\tau_0^a}{\tau_0^{b}}$ is smaller than the total weight $\sum_{c=a-1}^b \w_S(c)$ (scaled up by an appropriate constant factor).
Additionally, to handle corner cases, we also learn the longest prefix and the longest suffix of $T\fragment{\tau_0^0}{\tau_0^{\bc(\bG_S)-1}}$ of compressed size $\Oh(w + k)$.
As proved in \cref{prp:encode_simple_funcover}, the set $\{(c, T\position{\tau_0^c}):c\in C_S\}$ can be encoded in $\Oh((w + k)\log(m|\Sigma|))=\Oh(k|S|\log(m|\Sigma|))$ bits on top of the graph $\bG_S$ (which can be recovered from the edit information for alignments in $S$).

Following the aforementioned strategy of comparing the regions where $\mX: P \onto T\fragmentco{t}{t'}$ is disjoint with the canonical alignment $\A : P \onto T\fragmentco{t}{t'}$, we prove the following result.
Due to corner cases arising at the endpoints of $T\fragmentco{t}{t'}$ and between subsequent fragments $T\fragment{\tau_{i+j}^0}{\tau_{i+j}^{\bc(\bG_S)-1}}$ and $T\fragment{\tau_{i+j+1}^0}{\tau_{i+j+1}^{\bc(\bG_S)-1}}$, the proof is rather complicated.

\begin{restatable*}{proposition}{prpclose}\label{prp:close}
    Let $\mX : P \onto T\fragmentco{t}{t'}$ be an optimal alignment of $P$ onto a fragment $T\fragmentco{t}{t'}$ such that $\ed(P, T\fragmentco{t}{t'})\le k$.
    If there exists $i \in \fragment{0}{n_0-m_0}$ such that $|\tau_{i}^0 - t - \pi_0^0| \leq w + 3k$, then the following holds for every $c\in \fragmentco{0}{\bc(\bG_S)}\setminus C_S$:
    \begin{enumerate}
        \item\label{it:close:in}  $\mX$ aligns $P\position{\pi^c_j}$ to $T\position{\tau^c_{i+j}}$ for every $j \in \fragmentco{0}{m_c}$, and
        \item\label{it:close:out} $\tau^c_{i'}\notin \fragmentco{t}{t'}$ for every $i'\in \fragmentco{0}{n_c}\setminus \fragmentco{i}{i+m_c}$.\qedhere
    \end{enumerate}
\end{restatable*}

\subsubsection{Extending $S$ with Uncaptured Alignments}
\Cref{prp:close} indicates that $S$ \emph{captures} all $k$-error occurrences $T\fragmentco{t}{t'}$ such that $|\tau_{i}^0 - t - \pi_0^0| \leq w + 3k$ holds for some $i\in \fragment{0}{n_0-m_0}$.
As long as $S$ does not capture some $k$-error occurrence $T\fragmentco{t}{t'}$, we add an underlying optimal alignment $\mX : P \onto T\fragmentco{t}{t'}$ to the set $S$.
In \cref{lem:periodhalves}, we prove that $\bc(\bG_{S\cup\{\mX\}})\le \bc(\bG_S)/2$ holds for such an alignment $\mX$.
For this, we first eliminate the possibility of $t+ \pi_0^0 \gg \tau_{n_0-m_0}^0$ (using $\mXsuf\in S$, which matches $P\position{\pi_0^0}$ with $T\position{\tau_{n_0-m_0}^0}$).
If $|\tau_{i}^0 - t - \pi_0^0| > w + 3k$ holds for every $i\in \fragmentco{0}{n_0}$, on the other hand, then there is no $c\in \fragmentco{0}{\bc(\bG_S)}$ such that $P\position{\pi_0^c}$ can be matched with any character in the $c$th connected component.
Consequently, each black component becomes red or gets merged with another black component, resulting in the claimed inequality $\bc(\bG_{S\cup\{\mX\}})\le \bc(\bG_S)/2$.

By \cref{lem:periodhalves} and since $\bc(\bG_S)\le m$ holds when we begin with $|S|=2$, the total size $|S|$ does not exceed $\Oh(\log m)$ before we either arrive at $\bc(\bG_S)=0$, in which case the whole input can be encoded in $\Oh(k|S|\log (m|\Sigma|))$ bits, or $S$ captures all $k$-error occurrences.
In the latter case, the encoding consists of the edit information for all alignments in $S$, as well as the set $\{(c, T\position{\tau_0^c}):c\in C_S\}$ encoded using \cref{prp:encode_simple_funcover}.
Based on this encoding, we can construct strings $P^\#$ and $T^\#$ obtained from $P$ and $T$, respectively, by replacing with $\#_c$ every character in the $c$th connected component for every $c\in \fragmentco{0}{\bc(\bG_S)}\setminus C_S$. As a relatively straightforward consequence of \cref{prp:close}, in \cref{prp:subhash} we prove that $\OccE_k(P,T)=\OccE_k(P^\#,T^\#)$ and that the edit information is preserved for every optimal alignment $P\onto T\fragmentco{t}{t'}$ of cost at most $k$.

\subsection{Quantum Query Complexity of Pattern Matching with Edits}

As an illustration of the applicability of the combinatorial insights behind our
communication complexity result (\cref{thm:ccompl}), we study quantum algorithms for \PMwE.
As indicated in \cref{thm:qpmwe,thm:qpmwelb}, the query complexity we achieve is only a sub-polynomial factor away from the unconditional lower bounds, both for the decision version of the problem (where we only need to decide whether $\OccE_k(P, T)$ is empty or not) and for the standard version asking to report $\OccE_k(P, T)$.

Our lower bounds (in \cref{thm:qpmwelb}) are relatively direct applications of the adversary method of Ambainins~\cite{Ambainis02}, so this overview is solely dedicated to the much more challenging upper bounds.
Just like for the communication complexity above, we assume that $n\le \threehalves m$ and $k=o(m)$.
In this case, our goal is to achieve the query complexity of $\Ohhat(\sqrt{km})$.

Our solution incorporates four main tools:
\begin{itemize}
    \item the universal approximate pattern matching algorithm of~\cite{CKW20},
    \item the recent quantum algorithm for computing (bounded) edit distance~\cite{GJKT24},
    \item the novel combinatorial insights behind \cref{thm:ccompl},
    \item a new quantum $n^{o(1)}$-factor approximation algorithm for edit distance that uses $\Ohhat(\sqrt{n})$ queries and is an adaptation of a classic sublinear-time algorithm of~\cite{gapED}.
\end{itemize}

\subsubsection{A Baseline Algorithm}

We set the stage by describing a relatively simple algorithm that relies only on the first two of the aforementioned four tools.
This algorithm makes $\Ohtilde(\sqrt{k^3m})$ quantum queries to decide whether $\OccE_k(P,T)=\emptyset$.

The findings of \cite{CKW20} outline two distinct scenarios: either there are \emph{few \(k\)-error occurrences of \(P\) in \(T\)} or the pattern is \emph{approximately periodic}.
In the former case, the set \(\OccE_{k}(P,T)\) is of size $\Oh(k^2)$, and it is contained in a union of $\Oh(k)$ intervals of length $\Oh(k)$ each.
In the latter case, a primitive \emph{approximate period} \(Q\) of small length \(|Q| = \Oh(m/k)\) exists such that \(P\) and the relevant portion of $T$ (excluding the characters to the left of the leftmost $k$-error occurrence and to the right of the rightmost $k$-error occurrence) are at edit distance \(\Oh(k)\) to substrings of \(Q^{\infty}\).
It is solely the pattern that determines which of these two cases holds: the initial two options in the following lemma correspond to the \emph{non-periodic} case, where there are few \(k\)-error occurrences of \(P\) in \(T\), whereas the third option indicates the (approximately) \emph{periodic} case, where the pattern admits a short approximate period $Q$.
Here, $\edl{S}{Q}$ denotes the minimum edit distance between $S$ and any substring of $Q^\infty$.

\begin{restatable*}[{\cite[Lemma 5.19]{CKW20}}]{lemmaq}{EI}\label{prp:EI}
    Let $P$ denote a string of~length $m$ and let $k \le m$ denote a positive integer.
    Then, at least one of the following holds:
    \begin{enumerate}[(a)]
        \item\label{item:a:prp:EI} The string $P$ contains $2k$ disjoint fragments $B_1,\ldots, B_{2k}$ (called \emph{breaks}) each having period $\per(B_i)> m/\alphav k$ and length $|B_i| = \lfloor m/\betav k\rfloor$.
        \item\label{item:b:prp:EI} The string $P$ contains disjoint \emph{repetitive regions} $R_1,\ldots, R_{r}$ of total length $\sum_{i=1}^r |R_i| \ge \deltavN/\deltavD \cdot m$ such that each region $R_i$ satisfies $|R_i| \ge m/\betav k$ and has a primitive \emph{approximate period} $Q_i$ with $|Q_i| \le m/\alphav k$ and $\edl{R_i}{Q_i} = \ceil{\betav k/m\cdot |R_i|}$.
        \item\label{item:c:prp:EI} The string $P$ has a primitive \emph{approximate period} $Q$ with $|Q|\le m/\alphav k$ and $\edl{P}{Q} < \betav k$. \qedhere
    \end{enumerate}
\end{restatable*}

\indent

The proof of \cref{prp:EI} is constructive, providing a classical algorithm that performs the necessary decomposition and identifies the specific case.
The analogous procedure for \PMwM also admits an efficient quantum implementation~\cite{JN23} using $\Ohtilde(\sqrt{km})$ queries and time.
As our first technical contribution (\cref{lem:analyzeP}), we adapt the decomposition algorithm for the edit case to the quantum setting so that it uses $\tilde{\mathcal{O}}(\sqrt{km})$ queries and $\tilde{\mathcal{O}}(\sqrt{km} + k^2)$ time.

Compared to the classic implementation in~\cite{CKW20} and the mismatch version in~\cite{JN23}, it is not so easy to efficiently construct repetitive regions.
In this context, we are given a length-$\lfloor m/\betav k\rfloor$ fragment with exact period $Q_i$ and the task is to extend it to $R_i$ so that $k_i\coloneqq \edl{R_i}{Q_i}$ reaches $\ceil{\betav k/m\cdot |R_i|}$.
Previous algorithms use Longest Common Extension queries and gradually grow $R_i$, increasing $k_i$ by one unit each time; this can be seen as an online implementation of the Landau--Vishkin algorithm for the bounded edit distance problem~\cite{LV88}.
Unfortunately, the near-optimal quantum algorithm for bounded edit distance~\cite{GJKT24} is much more involved and does not seem amenable to an online implementation.
To circumvent this issue, we apply exponential search (just like in Newton's root-finding method, this is possible even though the sign of $\ceil{\betav k/m\cdot |R_i|}-\edl{R_i}{Q_i}$ may change many times).
At each step, we apply a slightly extended version of the algorithm of~\cite{GJKT24} that allows simultaneously computing the edit distance between $R_i$ and multiple substrings of $Q_i^\infty$;
see \cref{lem:find_min_approx_per}.

Once the decomposition has been computed, the next step is to apply the structure of the pattern in various cases to find the $k$-error occurrences.
The fundamental building block needed here is a subroutine that \emph{verifies} an interval $I$ of $\Oh(k)$ positive integers, that is, computes $\OccE_k(P, T) \cap I$.
The aforementioned extension of the bounded edit distance algorithm of~\cite{GJKT24} (\cref{lem:find_min_approx_per}) allows implementing this operation using $\Ohtilde(\sqrt{km})$ quantum queries and $\Ohtilde(\sqrt{km}+k^2)$ time.

By directly following the approach of \cite{CKW20}, computing $\OccE_k(P, T)$ can be reduced to verification of $\Oh(k^2)$ intervals (the periodic case constitutes the bottleneck for the number of intervals), which yields total a query complexity of $\Ohtilde(\sqrt{k^5 m})$.
If we only aim to decide whether $\OccE_{k}(P, T)$, we can apply \GS on top of the verification algorithm, reducing the query complexity to $\Ohtilde(\sqrt{k^3 m})$.
One can also hope for further speed-ups based on the more recent results of~\cite{CKW22}, where the number of intervals is effectively reduced to $\Ohtilde(k^{1.5})$.
Nevertheless, already in the non-periodic case, where the number of intervals is $\Oh(k)$,
this approach does not provide any hope of reaching query complexity beyond
$\Ohtilde(\sqrt{k^2 m})$ for the decision version and $\Ohtilde(\sqrt{k^3 m})$ for the
reporting version of \PMwE.

\subsubsection{How to Efficiently Verify \boldmath$\Oh(k)$ Candidate Intervals}

As indicated above, the main bottleneck that we need to overcome to achieve the near-optimal query complexity is to verify $\Oh(k)$ intervals using $\Ohhat(\sqrt{km})$ queries. Notably, an unconditional lower bound for bounded edit distance indicates that $\Omega(\sqrt{km})$ queries are already needed to verify a length-$1$ interval.

A ray of hope stemming from our insights behind \cref{thm:ccompl} is that, as described in \cref{sec:ccov},
already a careful selection of just $\Oh(\log m)$ among the $k$-error occurrences reveals a lot of structure that can be ultimately used to recover the whole set $\OccE_k(P,T)$.
To illustrate how to use this observation, let us initially make an unrealistic assumption that every candidate interval $I$ contains a $K$-error occurrence for some $K=\Ohhat(k)$.
Such occurrences can be detected using the existing verification procedure using $\Ohtilde(\sqrt{Km})=\Ohhat(\sqrt{km})$ queries.

First, we verify the leftmost and the rightmost intervals. This allows finding the leftmost and the rightmost $K$-error occurrences of $P$ in $T$.
We henceforth assume that text $T$ is cropped so that these two $K$-error occurrences constitute a prefix and a suffix of $T$, respectively.
The underlying alignments are the initial elements of the set $S$ that we maintain using the insights of \cref{sec:ccov}. Even though these two alignments have cost at most $K$, for technical reasons, we subsequently allow adding to $S$ alignments of cost up to $K'=K+\Oh(k)$.
Using the edit information for alignments $\mX\in S$, we build the graph $\bG_S$, calculate its connected components, and classify them as red and black components.

If there are no black components, that is, $\bc(\bG_S)=0$, then the edit information for the alignments $\mX \in S$ allows recovering the whole input strings $P$ and $T$.
Thus, no further quantum queries are needed, and we complete the computation using a classical verification algorithm in $\Oh(m+k^{3})$ time.

If there are black components, we retrieve the positions $\pi_0^0,\ldots,\pi_{m_0-1}^0$ and $\tau_0^0,\ldots,\tau_{n_0-1}^0$ contained in the $0$-th black component.
Based on these positions, we can classify $K'$-error occurrences $T\fragmentco{t}{t'}$ into those that are \emph{captured} by $S$ (for which $|\tau_i^0-\pi_0^0-t|$ is small for some $i\in \fragment{0}{n_0-m_0}$) and those which are not captured by $S$.
Although we do not know $K'$-error occurrences other than those contained in $S$, the test of comparing $|\tau_i^0-\pi_0^0-t|$ against a given threshold (which is $\Oh(K'|S|)$) can be performed for any position $t$, and thus we can classify arbitrary positions $t\in \fragment{0}{|T|}$ into those that are captured by $S$ and those that are not.

If any of the candidate intervals $I$ contains a position $t\in I$ that is not captured by $S$, we verify that interval and, based on our assumption, obtain a $K$-error occurrence of $P$ in $T$ that starts somewhere within $I$.
Furthermore, we can derive an optimal alignment $\mX:P\onto T\fragmentco{t}{t'}$ whose cost does not exceed $K+|I|\le K'$ because $|I|=\Oh(k)$.
This $K'$-error occurrence is not captured by $S$, so we can add $\mX$ to $S$ and, as a result, the number of black components decreases at least twofold by \cref{lem:periodhalves}.

The remaining possibility is that $S$ captures all positions $t$ contained in the candidate intervals~$I$.
In this case, our goal is to construct strings $P^\#$ and $T^\#$ of \cref{prp:subhash}, which are guaranteed to satisfy $\OccE_{k}(P,T)\cap I = \OccE_{k}(P^\#,T^\#)\cap I$ for each candidate interval $I$ because $k\le K'$.
For this, we need to build a period cover $C_S$ satisfying \cref{def:periodcover_alt},
which requires retrieving certain compressible substrings of $T$.
The minimum period cover $C_S$ utilized in our encoding (\cref{prp:encode_simple_funcover}) does not seem to admit an efficient quantum construction procedure, so we build a slightly larger period cover whose encoding incurs a logarithmic-factor overhead; see \cref{def:periodcover}.
The key subroutine that we repeatedly use while constructing this period cover asks to compute the longest fragment of $T$ (or of the reverse text $\rev{T}$) that starts at a given position and admits a Lempel--Ziv factorization~\cite{DBLP:journals/tit/ZivL77} of size bounded by a given threshold.
For this, we use exponential search combined with the recent quantum LZ factorization algorithm~\cite{GJKT24}; see \cref{prp:quantumlz,prp:quantum_blackcover}.
Based on the computed period cover, we can construct the strings $P^\#$ and $T^\#$
and resort to a classic verification algorithm (that performs no quantum queries) to process all $\Oh(k)$ intervals $I$ in time $\Oh(m+k^3)$.

The next step is to drop the unrealistic assumption that every candidate interval $I$ contains a $K$-error occurrence of $P$.
The natural approach is to test each of the candidate intervals using an approximation algorithm that either reports that $\Occ_{k}(P,T)\cap I= \emptyset$ (in which case we can drop the interval since we are ultimately looking for $k$-error occurrences) or that  $\Occ_{K}(P,T)\cap I \ne \emptyset$ (in which case the interval satisfies our assumption).
Given that $|I|$ is much smaller than $K$, it is enough to approximate $\ed(P, T\fragmentco{t}{t+m})$ for an arbitrary single position $t\in I$ (distinguishing between distances at most $\Oh(k)$ and at least $K-\Oh(k)$).
Although the quantum complexity of approximating edit distance has not been studied yet, we observe that the recent sublinear-time algorithm of Goldenberg, Kociumaka, Krauthgamer, and Saha~\cite{gapED} is easy to adapt to the quantum setting, resulting in a query complexity of $\Ohhat(\sqrt{n})$ and an approximation ratio of $n^{o(1)}=\Ohhat(1)$; see \cref{sec:quantumgaped} for details.

Unfortunately, we cannot afford to run this approximation algorithm for every candidate interval: that would require $\Ohhat(k\sqrt{m})$ queries.
Our final trick is to use \GS on top: given a subset of the $\Oh(k)$ candidate intervals, using just $\Ohhat(\sqrt{km})$ queries, we can either learn that none of them contains any $k$-error occurrence (in this case, we can discard all of them) or identify one that contains a $K$-error occurrence.
Combined with binary search, this approach allows discarding some candidate intervals so that the leftmost and the rightmost among the remaining ones contain $K$-error occurrences.
The underlying alignments (constructed using the exact quantum bounded edit distance algorithm of~\cite{GJKT24}) are used to initialize the set $S$.
At each step of growing $S$, on the other hand, we apply our approximation algorithm to the set of all candidate intervals that are not yet (fully) captured by $S$.
Either none of these intervals contain $k$-error occurrences (and the construction of $S$ may stop), or we get one that is guaranteed to contain a $K$-error occurrence.
In this case, we construct an appropriate low-cost alignment $\mX$ using the exact algorithm and extend the set $S$ with $\mX$.
Thus, the unrealistic assumption is not needed to construct the set $S$ and the strings $P^\#$ and $T^\#$ using $\Ohhat(\sqrt{km})$ queries.
\subsubsection{Handling the Approximately Periodic Case}

Verifying $\Oh(k)$ candidate intervals was the only bottleneck of the non-periodic case of
\PMwE.
In the approximately periodic case, on the other hand, we may have $\Oh(k^2)$ candidate intervals, so a direct application of the approach presented above only yields an $\Ohhat(\sqrt{k^2m})$-query algorithm.

Fortunately, a closer inspection of the candidate intervals constructed in~\cite{CKW20} reveals that they satisfy the unrealistic assumption that we made above: each of them contains an $\Oh(k)$-error occurrence of~$P$.
This is because both $P$ and the relevant part of $T$ are at edit distance $\Oh(k)$ from substrings of $Q^\infty$ and each of the intervals contains a position that allows aligning $P$ into $T$ via the substrings of $Q^\infty$ (so that perfect copies of $Q$ are matched with no edits).
Consequently, the set $S$ of $\Oh(\log m)$ alignments covering all candidate intervals can be constructed using $\Ohtilde(\sqrt{km})$ queries.
Moreover, once we construct the strings $P^\#$ and $T^\#$, instead of verifying all $\Oh(k^2)$ candidate intervals, which takes $\Oh(m+k^4)$ time, we can use the classic  $\Ohtilde(m+k^{3.5})$-time algorithm of~\cite{CKW22} to construct the entire set $\OccE_k(P^\#,T^\#)=\OccE_k(P,T)$.

\section{Preliminaries}\label{sec:prelims}

\subparagraph*{Sets.}
For integers $i,j \in \mathbb{Z}$, we write $\fragment{i}{j}$ to denote the set $\{i, \dots, j\}$ and $\fragmentco{i}{j}$ to denote the set $\{i ,\dots, j - 1\}$; we define the sets $\fragmentoc{i}{j}$ and $\fragmentoo{i}{j}$ similarly.

For a set $S$, we write $kS$ to denote the set obtained from \(S\) by multiplying every element with \(k\), that is, $kS \coloneqq \{ k\cdot s \mid s \in S\}$.
Similarly, we define $\floor{S/k} \coloneqq \{ \floor{s/k} \mid s \in S \}$ and $k\floor{S/k} \coloneqq \{k \cdot \floor{s/k} \mid s \in S \}$.

\subparagraph*{Strings.}

An \emph{alphabet} \(\Sigma\) is a set of characters.
We write $X=X\position{0}\, X\position{1}\cdots X\position{n-1} \in \Sigma^{n}$ to denote a \textit{string} of length $|X|=n$ over $\Sigma$.
For a \emph{position} $i \in \fragmentco{0}{n}$, we say that $X\position{i}$ is the $i$-th character of $X$.
For integer indices $0 \leq i \le j \leq |X|$, we say that $X\fragmentco{i}{j} \coloneqq X\position{i} \cdots X\position{j-1}$ is a \emph{fragment} of $X$.
We may also write $X\fragment{i}{j-1}$, $X\fragmentoc{i-1}{j-1}$, or $X\fragmentoo{i-1}{j}$ for the fragment $X\fragmentco{i}{j}$.
A \emph{prefix} of~a string $X$ is a fragment that starts at position $0$, and a \emph{suffix} of~a string $X$ is a fragment that ends at position $|X|-1$.

A string $Y$ of~length $m\in \fragment{0}{n}$ is a \emph{substring} of another string $X$ if there is a fragment $X\fragmentco{i}{i + m}$ that is equal to $Y$.
In this case, we say that there is an \emph{exact occurrence} of~$Y$ at position $i$ in~$X$.
Further, we write $\Occ(Y,X) \coloneqq \{i\in \fragment{0}{n-m} : Y = X\fragmentco{i}{i+m}\}$ for the set of starting positions of the (exact) occurrences of $Y$ in $X$.

For two strings $A$ and $B$, we write $AB$ for their concatenation.
We write $A^k$ for the concatenation of $k$ copies of the string $A$.
We write $A^\infty$ for an infinite string (indexed with non-negative integers) formed as the concatenation of an infinite number of copies of the string \(A\).
A \textit{primitive} string is a string that cannot be expressed as $A^k$ for any string $A$ and any integer $k > 1$.

An integer $p\in \fragment{1}{n}$ is \emph{a period} of a string $X\in\Sigma^n$ if we have $X\position{i} = X\position{i + p}$ for all $i \in \fragmentco{0}{n-p}$.
\emph{The period} of a string \(X\), denoted $\per(X)$, is the smallest period of \(X\).
A string \(X\) is \textit{periodic} if \(\per(X) \le |X| / 2\).

An important tool when dealing with periodicity is Fine and Wilf's Periodicity Lemma~\cite{FW65}.

\begin{lemmaq}[Periodicity Lemma~\cite{FW65}] \label{lem:perlemma}
    If \(p, q\) are periods of a string \(X\) of length \(|X| \geq p + q - \gcd(p, q)\), then \(\gcd(p, q)\) is a period of \(X\).
\end{lemmaq}

This allows us to derive the following relationship between exact occurrences and periodicity.

\begin{lemma}\label{fct:periodicity}
    Consider a non-empty pattern $P$ and a text $T$ with $|T|\le 2|P|+1$.
    If $\{0,|T|-|P|\}\subseteq \Occ(P,T)$, that is, $P$ occurs both as a prefix and as a suffix of $T$, then $\gcd(\Occ(P,T))$ is a period of $T$.
\end{lemma}
\begin{proof}
    Write $\Occ(P,T)=\{p_0,\ldots,p_{\ell-1}\}$, where $0=p_0 < \cdots < p_{\ell-1}=|T|-|P|$, and set $g\coloneqq\gcd(\Occ(P,T))$.
    From $T\fragmentco{p_i}{p_i+|P|}=P=T\fragmentco{p_{i-1}}{p_{i-1}+|P|}$, we obtain that, for each $i\in \fragmentoo{0}{\ell}$, the pattern $P$ has period $p_{i}-p_{i-1}$.
    Since
    \[\sum_{i\in \fragmentoo{0}{\ell}} \left(p_{i}-p_{i-1}\right)= p_{\ell-1}=|T|-|P| \le |P|+1 \le |P|+g,\]
    \cref{lem:perlemma} implies that $\gcd\{p_{i}-p_{i-1}\mid i\in \fragmentoo{0}{\ell}\}$ is a period of $P$.
    By repeatedly applying the property \(\gcd(a + b, b) = \gcd(a, b)\), we obtain $\gcd\{p_{i}-p_{i-1} \mid i\in \fragmentoo{0}{\ell}\} = g$.

    It remains to prove that $g$ is also a period of $T$, that is, we need to show that $T\position{j}=T\position{j\bmod g}$ holds for each $j\in \fragmentco{0}{|T|}$.
    To that end, fix a position \(j \in \fragmentco{0}{|T|}\).
    Unless \(j = |P|\), \(|T| = 2|P| + 1\), and \(\Occ(P,T)=\{0,|T|-|P|\}\), we have $j \in \fragmentco{p_i}{p_i+|P|}$ for some $i\in \fragmentco{0}{\ell}$.
    Now, as $g$ is a period of $P$ and a divisor of $p_i$, and because $\{0,p_i\}\subseteq
    \Occ(P,T)$, we obtain
    \[
        T\position{j} = P\position{j-p_i} = P\position{(j-p_i)\bmod g}= P\position{j\bmod g}=T\position{j\bmod g}.
    \]
    Finally, if $j=|P|$, $|T|=2|P|+1$, and $\Occ(P,T)=\{0,|T|-|P|\}$, then we observe that $g = |P|+1$, so $T\position{j}=T\position{j\bmod g}$ is trivially satisfied.
\end{proof}

\subparagraph*{Edit Distance and Alignments.}
The \emph{edit distance} (the \emph{Levenshtein distance}~\cite{Levenshtein66}) between two strings $X$ and $Y$, denoted by $\ed(X,Y)$, is the minimum number of character insertions, deletions, and substitutions required to transform $X$ into~$Y$.
Formally, we first define an \emph{alignment} between string fragments.
\begin{definition}[{\cite[Definition 2.1]{CKW22}}]\label{def:alignment}
    A sequence $\mA=(x_i,y_i)_{i=0}^{m}$ is an \emph{alignment} of $X\fragmentco{x}{x'}$ onto $Y\fragmentco{y}{y'}$, denoted by \(\mA: X\fragmentco{x}{x'} \onto Y\fragmentco{y}{y'}\), if it satisfies $(x_0,y_0)=(x,y)$, $(x_{i+1},y_{i+1})\in \{(x_{i}+1,\allowbreak y_{i}+1),(x_{i}+1,\allowbreak y_{i}),(x_{i},y_{i}+1)\}$ for $i\in \fragmentco{0}{m}$, and $(x_m,y_m) =(x',y')$. Moreover, for $i\in \fragmentco{0}{m}$:
    \begin{itemize}
        \item If $(x_{i+1},y_{i+1})=(x_{i}+1,y_{i})$, we say that $\mA$ \emph{deletes} $X\position{x_i}$.
        \item If $(x_{i+1},y_{i+1})=(x_{i},y_{i}+1)$, we say that $\mA$ \emph{inserts} $Y\position{y_i}$.
        \item If $(x_{i+1},y_{i+1})=(x_{i}+1,y_{i}+1)$, we say that $\mA$ \emph{aligns} $X\position{x_i}$ to $Y\position{y_i}$.
        If~additionally $X\position{x_i}=Y\position{y_i}$, we say that $\mA$ \emph{matches} $X\position{x_i}$ and $Y\position{y_i}$; otherwise, $\mA$ \emph{substitutes} $X\position{x_i}$ with $Y\position{y_i}$. \qedhere
    \end{itemize}
\end{definition}
Recall \cref{fig4} for a visualization of an example for an alignment.

The \emph{cost} of an alignment $\mA$ of $X\fragmentco{x}{x'}$ onto $Y\fragmentco{y}{y'}$, denoted by $\edal{\mA}(X\fragmentco{x}{x'},Y\fragmentco{y}{y'})$, is the total number of characters that $\mA$ inserts, deletes, or substitutes.
The edit distance $\ed(X,Y)$ is the minimum cost of an alignment of $X\fragmentco{0}{|X|}$ onto~$Y\fragmentco{0}{|Y|}$.
An alignment of $X$ onto $Y$ is \emph{optimal} if its cost is equal to $\ed(X, Y)$.

An alignment $\mA'':X\fragmentco{x}{x'}\onto Z\fragmentco{z}{z'}$ is a \emph{product} of alignments $\mA:X\fragmentco{x}{x'}\onto Y\fragmentco{y}{y'}$ and $\mA':Y\fragmentco{y}{y'}\onto Z\fragmentco{z}{z'}$ if, for every $(\bar{x},\bar{z})\in \mA''$, there is $\bar{y}\in \fragment{y}{y'}$ such that $(\bar{x},\bar{y})\in \mA$ and $(\bar{y},\bar{z})\in \mA'$.
A product alignment always exists, and every product alignment satisfies
\[
\edal{\mA''}(X\fragmentco{x}{x'},Z\fragmentco{z}{z'})\le
\edal{\mA}(X\fragmentco{x}{x'},Y\fragmentco{y}{y'})+\edal{\mA'}(Y\fragmentco{y}{y'},\allowbreak
Z\fragmentco{z}{z'}).
\]

For an alignment $\mA:X\fragmentco{x}{x'}\onto Y\fragmentco{y}{y'}$ with \(\mA = (x_i,y_i)_{i=0}^m\), we define the \emph{inverse alignment} $\mA^{-1} : Y\fragmentco{y}{y'}\onto X\fragmentco{x}{x'}$ as $\mA^{-1} \coloneqq (y_i,x_i)_{i=0}^m$.
The inverse alignment satisfies
\[
    \edal{\mA^{-1}}(Y\fragmentco{y}{y'},X\fragmentco{x}{x'})=\edal{\mA}(X\fragmentco{x}{x'},Y\fragmentco{y}{y'}).
\]

Given an alignment $\A:X\fragmentco{x}{x'}\onto Y\fragmentco{y}{y'}$ and a fragment $X\fragmentco{\bar{x}}{\bar{x}'}$ of $X\fragmentco{x}{x'}$, we write $\A(X\fragmentco{\bar{x}}{\bar{x}'})$ for the fragment $Y\fragmentco{\bar{y}}{\bar{y}'}$ of $Y\fragmentco{y}{y'}$ that $\A$ aligns against $X\fragmentco{\bar{x}}{\bar{x}'}$.
As insertions and deletions may render this definition ambiguous, we formally set
\[\bar{y} \coloneqq \min\{\hat{y} : (\bar{x},\hat{y})\in \A\}\quad\text{and}\quad
    \bar{y}' \coloneqq \left\{\begin{array}{c l}
            y' & \text{if }\bar{x}' = x',\\
            \min\{\hat{y}' : (\bar{x}',\hat{y}')\in \A\} & \text{otherwise}.
    \end{array}\right.
\]
This particular choice satisfies the following decomposition property.
\begin{fact}[{\cite[Fact 2.2]{CKW22}}]\label{fct:ali}
    For any alignment $\A$ of $X$ onto $Y$ and a decomposition $X=X_1\cdots X_t$ into $t$ fragments, $Y=\A(X_1)\cdots \A(X_t)$ is a decomposition into $t$ fragments with $\ed^\A(X,Y)= \sum_{i=1}^t \ed^{\A}(X_i,\A(X_i))$.
    Further, if \(\A\) is an optimal alignment, then $\ed(X,Y) = \sum_{i=1}^t \ed(X_i,\A(X_i))$.
\end{fact}

We use the following \emph{edit information} notion to encode alignments in space proportional to their costs.

\begin{definition}[Edit information]\label{def:edinfo}
    For an alignment $\mA=(x_i,y_i)_{i=0}^m$ of $X\fragmentco{x}{x'}$ onto $Y\fragmentco{y}{y'}$,
    the \emph{edit information} is defined as the set of 4-tuples $\sE_{X,
        Y}(\mA)=\{(x_i,\mathsf{cx}_i \mid y_i,\mathsf{cy}_i) : i\in
    \fragmentco{0}{m}\text{ and }\mathsf{cx}_i\ne \mathsf{cy}_i\}$, where
    \[\mathsf{cx}_i = \begin{cases}
        X\position{x_i} & \text{if }x_{i+1}=x_i+1,\\
        \varepsilon & \text{otherwise};
    \end{cases}\qquad\text{and}\qquad
    \mathsf{cy}_i = \begin{cases}
        Y\position{y_i} & \text{if }y_{i+1}=y_i+1,\\
        \varepsilon & \text{otherwise}.
    \end{cases}
    \]
\end{definition}

Observe that given two strings \(X\) and \(Y\), along with the edit information \(\sE_{X, Y}(\mA)\)
for an alignment \(\mA : X \onto Y\fragmentco{y}{y'}\) of non-zero cost,
we are able to fully reconstruct \(\mA\).
This is because elements in \(\mA\) without corresponding entries in \(\sE_{X, Y}(\mA)\) represent matches.
Therefore, we can deduce the missing pairs of \(\mA\) between two consecutive elements in \(\sE_{X, Y}(\mA)\),
as well as before (or after) the first (or last) element of \(\sE_{X, Y}(\mA)\).
This inference requires at least one 4-tuple to be contained in \(\sE_{X, Y}(\mA)\).

\subparagraph*{Pattern Matching with Edits.}

We denote the minimum edit distance between a string $S$ and any prefix of a string $T^\infty$ by $\edp{S}{T} \coloneqq \min\{\ed(S,T^\infty\fragmentco{0}{j}) \mid j \in \Zz\}$.
We denote the minimum edit distance between a string $S$ and any substring of $T^\infty$ by $\edl{S}{T} \coloneqq \min\{\ed(S,T^\infty\fragmentco{i}{j}) \mid i, j \in \Zz\text{ and }i \le j\}$.
Lastly, we set $\edp{S}{T} \coloneqq \min \{\ed(S,T^{\infty}\fragmentco{i}{j|T|}) \mid i,j \in \Zz, i\leq j|T|\}$.

In the context of two strings $P$ (referred to as the pattern) and $T$ (referred to as the text), along with a positive integer $k$ (referred to as the threshold), we say that there is a $k$-error or $k$-edits occurrence of $P$ in $T$ at position $i\in \fragment{0}{|T|}$ if $\ed(P, T\fragmentco{i}{j})\leq k$ holds for some position $j\in \fragment{i}{|T|}$.
The set of all starting positions of $k$-error occurrences of $P$ in $T$ is denoted by $\OccE_k(P,T)$; formally, we set
\[
    \OccE_k (P,T)\coloneqq \{i\in \fragment{0}{|T|} \mid \exists_{j\in \fragment{i}{|T|}}, \ed(P,T\fragmentco{i}{j}\leq k)\}.
\]

We often need to compute \(\OccE_k(P, T)\).
To that end, we use the recent algorithm of Charalampopoulos, Kociumaka, and Wellnitz~\cite{CKW22}.

\begin{lemma}[{\cite[Main Theorem 1]{CKW22}}]\label{prp:classical_pmwe}
    Let $P$ denote a pattern of length $m$, let $T$ denote a text of length~$n$, and let $k\in \Zz$ denote a threshold.
    Then, there is a (classical) algorithm that computes $\OccE_{k}(P,T)$ in time $\Ohtilde(n+n/m \cdot k^{3.5})$. \lipicsEnd
\end{lemma}

\subparagraph*{Compression and Lempel--Ziv Factorizations.}

We say that a fragment $X\fragmentco{i}{i+\ell}$ is a \emph{previous factor} if it has an earlier occurrence in $X$, that is, $X\fragmentco{i}{i+\ell}=X\fragmentco{i'}{i'+\ell}$ holds for some $i'\in \fragmentco{0}{i}$.
An \emph{LZ77-like factorization} of $X$ is a factorization $X = F_1\cdots F_f$ into non-empty \emph{phrases} such that each phrase $F_j$ with $|F_j|>1$ is a previous factor.
In the underlying \emph{LZ77-like representation}, every phrase is encoded as follows (consult \cref{fig:lz77_example} for a visualization of an example).
\begin{itemize}
    \item A previous factor phrase $F_j=X\fragmentco{i}{i+\ell}$ is encoded as $(i',\ell)$, where $i'\in \fragmentco{0}{i}$ satisfies $X\fragmentco{i}{i+\ell}=X\fragmentco{i'}{i'+\ell}$.
    The position \(i'\) is chosen arbitrarily in case of ambiguities.
    \item Any other phrase $F_j=X\position{i}$ is encoded as $(X\position{i},0)$.
\end{itemize}

\begin{figure}[t]
    \resizebox{\textwidth}{!}{%
        \centering
        \begin{tabular}{m{12mm} m{7mm}| m{7mm} | m{7mm} | m{7mm} | m{7mm} m{7mm} | m{7mm}
            m{7mm} m{7mm} m{7mm} m{7mm} | m{7mm} m{7mm} m{7mm} |m{7mm}}
        Index &0 & 1 & 2 & 3 & 4 & 5 & 6 & 7 & 8 & 9 & 10 & 11 & 12 & 13 & 14\\
        $X$ & $\mathtt{a}$  & $\mathtt{b}$ & $\mathtt{a}$  & $\mathtt{c}$ & $\mathtt{a}$ & $\mathtt{b}$ & $\mathtt{c}$ & $\mathtt{a}$ & $\mathtt{b}$ & $\mathtt{c}$ & $\mathtt{a}$ & $\mathtt{a}$ & $\mathtt{a}$ & $\mathtt{a}$ & $\mathtt{b}$ \\
        \end{tabular}
    }
    \caption{The LZ77 factorization of a string $X = \mathtt{abacabcabcaaaab}$ of length $n = 15$.
    The resulting encoding has $z=8$ elements: $(\mathtt{a}, 0)$, $(\mathtt{b}, 0)$, $(0,1)$, $(\mathtt{c}, 0)$, $(0, 2)$, $(3, 5)$, $(10, 3)$, $(8,1)$.}\label{fig:lz77_example}
\end{figure}

The LZ77 factorization~\cite{DBLP:journals/tit/ZivL77} (or the LZ77 parsing) of a string $X$, denoted by $\LZ(X)$ is an LZ77-like factorization that is constructed by greedily parsing $X$ from left to right into the longest possible phrases.
More precisely, the $j$-th phrase $F_j=X\fragmentco{i}{i+\ell}$ is the longest previous factor starting at position $i$; if no previous factor starts at position \(i\), then $F_j$ is the single character \(X\position{i}\).
It is known that the aforementioned greedy approach produces the shortest possible LZ77-like factorization.

\subparagraph*{Self-Edit Distance of a String.}
Another measure of compressibility that is instrumental in this paper is the \emph{self-edit distance} of a string, recently introduced by Cassis, Kociumaka, and Wellnitz~\cite{CKW23}.
An alignment $\mA : X \onto X$ is a \emph{self-alignment} if $\mA$ does not align any character $X\position{x}$ to itself.
The \emph{self-edit distance} of $X$, denoted by $\selfed(X)$, is the minimum cost of a self-alignment; formally, we set $\selfed(X) \coloneqq \min_{\mA} \edal{\mA} (X, X)$, where the minimization ranges over all self-alignments $\mA : X \onto X$.
A small self-edit distance implies that we can efficiently encode the string.

\begin{lemma} \label{prp:lz_selfed}
    For any string $X$, we have $|\LZ(X)| \leq 2\cdot\selfed(X)$.
\end{lemma}
\begin{proof}
    Consider an optimal self-alignment $\mA : X \onto X$.
    Without loss of generality, we may assume that each point $(x,y)\in \mA$ satisfies $x\ge y$.%
    \footnote{If this is not the case, we write $\mA = (x_t,y_t)_{t=0}^m$ and observe that $\mA' = (\max(x_t,y_t), \min(x_t,y_t))_{t=0}^m$ is a self-alignment of the same cost.
    Geometrically, this means that we replace parts of $\mA$ below the main diagonal with their mirror image.}
    Let us partition $X$ into individual characters that $\mA$ deletes or substitutes, as well as maximal fragments that $\mA$ matches perfectly.

    Observe that each fragment $X\fragmentco{x}{x'}$ that $\mA$ matches perfectly is a previous factor:
    indeed, $X\fragmentco{x}{x'}$ matches $X\fragmentco{y}{y'}$, and $x>y$ holds because $(x,y)\in \mA$, $x\ge y$, and $\mA$ does not match $X\position{x}$ with itself.
    Consequently, $X\fragmentco{x}{x'}$ is a previous factor and the partition defined above forms a valid LZ77-like representation.

    Since $\mA$ makes at most $\selfed(X)$ edits and (if $|X|>0$) this includes a deletion of $X\position{0}$, the total number of phrases in the partition does not exceed $2\cdot \selfed(X)$.
\end{proof}

Note that, since $\selfed(X)$ is insensitive to string reversal, we also have $|\LZ(\rev{X})| \leq 2 \cdot \selfed(X)$, where $\rev{X}=X[|X|-1]\cdots X[1]X[0]$ is the reversal of $X$.
We use the following known properties of $\selfed$.

\begin{lemmaq}[Properties of $\selfed$, {\cite[Lemma~4.2]{CKW23}}] \label{prp:prop_selfed}
    For any string  $X$, all of the following hold:
    \begin{description}
        \item[Monotonicity.] For any $\ell' \leq \ell \leq r \leq r' \in \fragment{0}{|X|}$, we have $\selfed(X\fragmentco{\ell}{r}) \leq \selfed(X\fragmentco{\ell'}{r'})$.
        \item[Sub-additivity.] For any $m \in \fragment{0}{|X|}$, we have $\selfed(X) \leq \selfed(X\fragmentco{0}{m}) + \selfed(X\fragmentco{m}{|X|})$.
        \item[Triangle inequality.] For any string $Y$, we have $\selfed(Y) \leq \selfed(X) + 2\ed(X,Y)$. \qedhere
    \end{description}
\end{lemmaq}

Cassis, Kociumaka, and Wellnitz~\cite{CKW23} also proved the following lemma that bounds the self-edit distance of $Y$ in the presence of two disjoint alignments mapping $Y$ to nearby fragments of another string~$X$.

\begin{lemmaq}[{\cite[Lemma~4.5]{CKW23}}]\label{prp:edimpliesselfed}
    Consider strings $X, Y$ and alignments $\A : Y \onto X\fragmentco{i}{j}$ and $\A' : Y \onto X\fragmentco{i'}{j'}$.
    If there is no $(y, x) \in \A \cap \A'$  such that both $\A$ and $\A'$ match $Y\position{y}$ with $X\position{x}$, then
    \[
        \selfed(Y) \leq |i - i'| + \edal{\mA}(Y, X\fragmentco{i}{j}) + \edal{\mA'}(Y, X\fragmentco{i'}{j'}) + |j-j'|.
        \qedhere
    \]
\end{lemmaq}

\section{New Combinatorial Insights for Pattern Matching with Edits}
\label{sec:combres}

In this section, we fix a threshold $k$ and two strings $P$ and $T$ over a common input alphabet $\Sigma$.
Furthermore, we let $S$ be a set of alignments of $P$ onto fragments of $T$ such that all alignments in $S$ have cost at most $k$.

\subsection{The Periodic Structure induced by \texorpdfstring{$S$}{S}}

We now formally define the concepts introduced in the Technical Overview.
Again, revisit \cref{fig4} for a visualization of an example.

\begin{definition}\label{def:bg}
    We define the undirected graph $\bG_{S} = (V, E)$ as follows.
    \begin{description}
        \item[The vertex set $V$]contains:
        \begin{itemize}
            \item $|P|$ vertices representing characters of $P$;
            \item $|T|$ vertices representing characters of $T$; and
            \item one special vertex $\bot$.
        \end{itemize}
        \item[The edge set $E$]contains the following edges for each alignment $\mX\in S$:
        \begin{enumerate}[(i)]
            \item $\{P\position{x},\bot\}$ for every character $P\position{x}$ that $\mX$ deletes;\label{it:bg:i}
            \item $\{\bot, T\position{y}\}$ for every character $T\position{y}$ that $\mX$ inserts;\label{it:bg:ii}
            \item $\{P\position{x},T\position{y}\}$ for every pair of characters $P\position{x}$ and $T\position{y}$ that $\mX$ aligns.\label{it:bg:iii}
        \end{enumerate}
        We say that an edge $\{P\position{x},T\position{y}\}$ is \emph{black} if $\mX$ matches $P\position{x}$ and $T\position{y}$;
        all the remaining edges are \emph{red}.
    \end{description}

    We say that a connected component of $\bG_{S}$ is \emph{red} if it contains at least one red edge; otherwise, we say that the connected component is \emph{black}.
    We denote with $\bc(\bG_{S})$ the number of black components in $\bG_{S}$.
\end{definition}

Note that all vertices contained in black components correspond to characters of $P$ and $T$.
Moreover, all characters of a single black component are the same---a remarkable property.
This is because the presence of a black edge indicates that some alignment in $S$ matches the two corresponding characters.
Hence, the two characters must be equal.
Since a black component contains only black edges, all characters contained in one such component are equal.

Suppose we know the edit information $\sE_{P, T}(\mX)$ for each alignment $\mX: P\onto T\fragmentco{t}{t'}$ in $S$.
Then, we can reconstruct the complete edge set of $\bG_{S}$, without needing any other information.
Moreover, we can also distinguish between red and black edges, and consequently between red and black components.

Next, we introduce the property that we assume about $S$.
This property induces the periodic structure in $P$ and $T$ at the core of our combinatorial results.

\begin{definition}
    We say $S$ \emph{encloses $T$} if $|T| \le 2 \cdot |P| - 2k$ and there exist two distinct alignments $\mXpref, \mXsuf \in S$ such that $\mXpref$ aligns $P$ with a prefix of $T$ and $\mXsuf$ aligns $P$ with a suffix of $T$, or equivalently $(0,0) \in \mXpref$ and $(|P|,|T|) \in \mXsuf$.
\end{definition}

Suppose that we have $\bc(\bG_{S}) = 0$.
Then, we want to argue that the information $\{\sE_{P,T}(\mX) : S \ni \mX : P \onto T\fragmentco{t}{t'}\}$ mentioned earlier is sufficient to fully retrieve $P$ and $T$.
As already argued before, the information is sufficient to fully reconstruct $\bG_{S}$.
From $\bc(\bG_{S}) = 0$ follows that every character of $P$ and $T$ lies in a red component.
As a consequence, for every character of $P$ and $T$, there exists a path (possibly of length zero) starting at that character, only containing black edges, and ending in a character incident to a red edge.
Since the path only contains black edges, all characters contained in the path must be equal.
Since we store all characters incident to red edges, we can retrieve the character at the starting point of the path.
We conclude that we can fully retrieve $P$ and $T$.
Note that encoding this information occupies $\Oh(k|S|)$ space, since, for each alignment in $S$, we can store the corresponding edit information using $\Oh(k)$ space.

We dedicate the remaining part of \cref{sec:combres} until \cref{sec:recocc} to the case $\bc(\bG_{S}) > 0$.
We  prove which information, other than the edit information, we must store to encode all $k$-edit occurrences.
By storing the edit information for each $\mX \in S$, we  be always able to fully reconstruct $\bG_{S}$ and to infer which characters are contained in red components.
Hence, the only additional information left to store in this case, are the characters
contained in black components.

Throughout the remaining part of \cref{sec:combres} until \cref{sec:recocc}, we  always assume that $S$ encloses $T$ and that $\bc(\bG_{S}) > 0$.
This is where we  need to exploit the periodic structure induced by $S$.
Before describing more formally the periodic structure induced by $S$, we need to introduce two additional strings $T_{|S}$ and $P_{|S}$, constructed by retaining characters contained in black connected components.

\begin{definition}
    We let $T_{|S}$ denote the subsequence of $T$ consisting of characters contained in black components of~$\bG_{S}$.
    Similarly, $P_{|S}$ denotes the subsequence of $P$ consisting of characters contained in black components of~$\bG_{S}$.
\end{definition}

Now, we prove the lemma that characterizes the periodic structure induced by $S$.

\begin{lemma}\label{lem:periodicity}
    For every $c \in \fragmentco{0}{\bc(\bG_S)}$, there exists a black connected component with node set
    \[
        \{P_{|S}[i] : i \equiv_{\bc(\bG_S)} c\} \cup \{T_{|S}[i] : i \equiv_{\bc(\bG_S)} c\},
    \]
    that is, there exists a black connected component containing all characters of $P_{|S}$ and $T_{|S}$ appearing at positions congruent to $c$ modulo $\bc(\bG_S)$.
    Moreover, the last characters of $P_{|S}$ and $T_{|S}$ are contained in the same black connected component, that is, $|T_{|S}| \equiv_{\bc(\bG_{S})} |P_{|S}|$.
\end{lemma}
\begin{proof}
    The following claim characterizes the edges of $\bG_S$ induced by a single alignment $\mX\in S$.
    The crucial observation is that these edges correspond to an \emph{exact occurrence of $P_{|S}$ in $T_{|S}$}.

    \begin{claim}\label{clm:occ}
        For every $\mX \in S$, there exists $\Delta_{\mX}\in \fragment{0}{|T_{|S}|-|P_{|S}|}$ such that $\mX$ induces edges between $P_{|S}\position{p}$ and $T_{|S}\position{p+\Delta_{\mX}}$ for all $p\in \fragmentco{0}{|P_{|S}|}$ and no other edges incident to characters of $P_{|S}$ or $T_{|S}$.
    \end{claim}
    \begin{claimproof}
        By \cref{def:bg}, black edges induced by $\mX$ form a matching between $P_{|S}$ and the characters of $T_{|S}$ contained in the image $\mX(P)$.
        The characters contained in $\mX(P)$ appear consecutively in $T_{|S}$, so $\mX$ induces a matching between $P_{|S}$ and a length-$|P_{|S}|$ fragment of $T_{|S}$.
        Every such fragment is of the form $T_{|S}\fragmentco{\Delta_{\mX}}{\Delta_{\mX}+|P_{|S}|}$ for some $\Delta_{\mX}\in \fragment{0}{|T_{|S}|-|P_{|S}|}$.
        Moreover, $\mX$ is non-crossing and every edge induced by $\mX$ corresponds to an element of $\mX$,so the matching induced by $\mX$ consists of edges between $P_{|S}\position{p}$ and $T_{|S}\position{p+\Delta_{\mX}}$ for every $p\in \fragmentco{0}{|P_{|S}|}$.
    \end{claimproof}

    Next, we explore the consequences of the assumption that $S$ encloses $T$.

    \begin{claim}\label{clm:nosingle}
        There exist alignments $\mXpref,\mXsuf\in S$ such that $\Delta_{\mXpref}=0$ and $\Delta_{\mXsuf}=|T_{|S}|-|P_{|S}|$.
        Moreover, $|T_{|S}|\le 2|P_{|S}|$.
    \end{claim}
    \begin{claimproof}
        Since $S$ encloses $T$, there are alignments $\mXpref, \mXsuf \in S$ such that $(0,0) \in \mXpref$ and $(|P|,|T|) \in \mXsuf$, that is, $\mXpref(P)$ is a prefix of $T$ whereas $\mXsuf(P)$ is a suffix of $T$.
        Every alignment $\mX\in S$ induces a matching between $P_{|S}$ and the characters of $T_{|S}$ contained in the image $\mX(P)$, so we must have $\Delta_{\mXpref}=0$ and $\Delta_{\mXsuf}=|T_{|S}|-|P_{|S}|$.
        The costs of $\mXpref(P)$ and $\mXsuf(P)$ are at most $k$, so $|\mXpref(P)|\ge |P|-k$ and $|\mXsuf(P)|\ge |P|-k$.
        Due to the assumption $|T|\le 2|P|-2k$, this means that every character of $T$ is contained in $\mXpref(P)$ or $\mXsuf(P)$.
        Consequently, the two matchings induced by these alignments jointly cover $T_{|S}$, and thus $|T_{|S}|\le 2|P_{|S}|$.
    \end{claimproof}

    Let us assign a unique label $\$_C$ to each black component $C$ of $\bG_S$ and define strings $P_{\rm{bc}}$ and $T_{\rm{bc}}$ of
    length $|P_{|S}|$ and $|T_{|S}|$, respectively, as follows:
    For $i\in \fragmentco{0}{|P_{|S}|}$, set $P_{\rm{bc}}[i] = \$_C$, where $C$ is black component containing $P_{|S}[i]$.
    Similarly, for $i\in \fragmentco{0}{|T_{|S}|}$, set $T_{\rm{bc}}[i] = \$_C$, where $C$ is the black connect component containing $T_{|S}[i]$.

    By \cref{clm:occ,clm:nosingle}, $\{0,|T_{\rm{bc}}|-|P_{\rm{bc}}|\}\subseteq \{\Delta_{\mX} : \mX\in S\} \subseteq \Occ(P_{\rm{bc}},T_{\rm{bc}})$.
    Due to $|T_{\rm{bc}}|\le 2|P_{\rm{bc}}|$, \cref{fct:periodicity} implies that
    $\gcd(\Occ(P_{\rm{bc}},T_{\rm{bc}}))$ is a period of $T_{\rm{bc}}$, and thus also $g\coloneqq
    \gcd\{\Delta_{\mX} : \mX\in S\}$ is a period of both $T_{\rm{bc}}$ and of its prefix
    $P_{\rm{bc}}$.
    Consequently, for each $c\in \fragmentco{0}{g}$, the set $C_c\coloneqq \{P_{|S}[i] : i \equiv_{g} c\} \cup \{T_{|S}[i] : i \equiv_{g} c\}$ belongs to a single connected component of $\bG_S$.
    It remains to prove $g = \bc(\bG_S)$, and we do this by demonstrating that no edge of $\bG_S$ leaves $C_c$.
    Note that \cref{clm:occ} further implies that every edge incident to $P_{|S}$ or $T_{|S}$ connects $P_{|S}\position{p}$ with $T_{|S}\position{t}$ such that $t=p+\Delta_{\mX}$ for some $\mX\in S$.
    In particular, $t \equiv_g p$, so $P_{|S}\position{p}\in C_c$ holds if and only if $T_{|S}\position{t}\in C_c$.

    Lastly, since $P_{\rm{bc}}$ is a suffix of $T_{\rm{bc}}$, the last characters of $P_{|S}$ and $T_{|S}$ are in the same black component.
\end{proof}

For the sake of convenience, we index black connected components with integers in $\fragmentco{0}{\bc(\bG_S)}$ and introduce some notation:

\begin{definition}\label{def:pitau}
    For $c \in \fragmentco{0}{\bc(\bG_S)}$, we define the \emph{$c$-th black connected
    component} as the black connected component containing $P_{|S}[c]$ and set
    \[
        m_c \coloneqq  \left\lceil\frac{|P_{|S}| - c}{\bc(\bG_{S})} \right\rceil \quad\text{and}\quad n_c \coloneqq  \left\lceil \frac{|T_{|S}| - c}{\bc(\bG_{S})} \right\rceil,
    \]
    as the number of characters in $P$ and $T$, respectively, belonging to the $c$-th black connected component.
    Furthermore:
    \begin{itemize}
        \item We define $c_{\last} \in \fragmentco{0}{\bc(\bG_{S})}$ as the black component containing the last characters of $P_{|S}$ and $T_{|S}$.
        \item For $c \in \fragmentco{0}{\bc(\bG_{S})}$ and $j \in \fragmentco{0}{m_c}$ we define $\pi_j^c \in \fragmentco{0}{|P|}$ as the position of $P_{|S}\position{c + j \cdot \bc(\bG_{S})}$ in $P$.
        \item For $c \in \fragmentco{0}{\bc(\bG_{S})}$ and $i \in \fragmentco{0}{n_c}$, we define $\tau_i^c \in \fragmentco{0}{|T|}$ as the position of $T_{|S}\position{c + i \cdot \bc(\bG_{S})}$ in $T$.
        \item For sake of convenience, we denote $m_{\bc(\bG_S)}=m_0-1$ and $\pi_j^{\bc(\bG_S)}=\pi_{j+1}^0$  for $j\in \fragmentco{0}{m_{\bc(\bG_S)}}$.
        Similarly, we denote $n_{\bc(\bG_S)}=n_0-1$ and $\tau_i^{\bc(\bG_S)}=\tau_{i+1}^0$ for $i\in \fragmentco{0}{n_{\bc(\bG_S)}}$.\qedhere
    \end{itemize}
\end{definition}

\begin{remark} \label{rmk:perstr}
    We conclude this (sub)section with some observations.
    \begin{enumerate}
        \item Consider $c \in \fragmentco{0}{\bc(\bG_{S})}$.
        Notice that the characters $\{T\position{\tau_i^c}\}_{i=0}^{n_{c}-1}$ and $\{P\position{\pi_j^c}\}_{j=0}^{m_{c}-1}$ are precisely those within the $c$-th black connected component.
        Consequently, they are all identical.
        \item For all $c \in \fragment{0}{\bc(\bG_{S})}$, we have $n_c \in \{n_0, n_0 - 1\}$ and $m_c \in \{m_0, m_0 - 1\}$.
        \item The indices ${\{\tau_i^c\}}_{i, c}$ induce a partition of $T$ into (possibly empty) substrings.
        This partition consists of: the initial fragment $T\fragmentco{0}{\tau_0^0}$, all the intermediate fragments $T\fragmentco{\tau_{i}^{c}}{\tau_{i}^{c+1}}$ for $c \in \fragmentco{0}{\bc(\bG_S)}$ and $i \in \fragmentco{0}{n_{c+1}}$, and the final fragment $T\fragmentco{\tau_{n_0-1}^{c_{\last}}}{|T|}$.
        Similarly, the indices ${\{\pi_j^c\}}_{j, c}$ induce a partition of $P$.
        \item Recall that $\bG_S$ contains an edge between $P\position{\pi_j^c}$ and $T\position{\tau_i^c}$ only if there is an alignment $\mX\in S$ such that $(\pi_j^c,\tau_i^c)\in \mX$.
        If $j< m_{c+1}$ and $i<n_{c+1}$, we also have $(\pi_j^{c+1},\tau_i^{c+1})\in \mX$.
        For such $(\pi_j^c,\tau_i^c)$, we can associate with it at least one value $\edal{\mX}(P\fragmentco{\pi_j^c}{\pi_j^{c+1}}, T\fragmentco{\tau_i^c}{\tau_i^{c+1}})$, noting that there might be multiple values associated with $(\pi_j^c,\tau_i^c)$.
        \label{rmk:perstr:4}
    \end{enumerate}
\end{remark}

\subsection{Covering Weight Functions}

\begin{definition}
    \label{def:funcover}
    We say that a \emph{weight function} $\w_S : \fragmentco{0}{\bc(\bG_S)}\to
    \mathbb{Z}_{\ge 0}$ \emph{covers $S$} if all of the following conditions hold:
    \begin{enumerate}
        \item for all $c\in \fragmentco{0}{\bc(\bG_S)}$, $j\in \fragmentco{0}{m_{c+1}}$, and $i\in \fragmentco{0}{n_{c+1}}$,
        we have
        \begin{equation}
            \label{eq:funcover}
            \w_S(c) \ge \ed(P\fragmentco{\pi_j^c}{\pi_j^{c+1}},T\fragmentco{\tau_i^c}{\tau_i^{c+1}});
        \end{equation}
        \label{it:funcover:1}
        \item $\w_S(\bc(\bG_S)-1) \geq \ed(P\fragmentco{0}{\pi_0^0}, T\fragmentco{0}{\tau_0^0})$;
        \label{it:funcover:2}
        \item $\w_S(\bc(\bG_S)-1)\ge \ed(P\fragmentco{0}{\pi_0^0},T\fragmentco{t}{\tau_i^0})$ holds for every $i\in \fragmentco{1}{n_{0}}$ and some $t\in \fragment{\tau_{i-1}^{\bc(\bG_S)-1}}{\tau_i^0}$;
        \label{it:funcover:3}
        \item $\w_S(c_{\last}) \geq \ed(P\fragmentco{\pi_{m_{0}-1}^{c_{\last}}}{|P|}, T\fragmentco{\tau_{n_{0}-1}^{c_{\last}}}{|T|})$; and
        \label{it:funcover:4}
        \item $\w_S(c_{\last}) \ge \ed(P\fragmentco{\pi_{m_{0}-1}^{c_{\last}}}{|P|}, T\fragmentco{\tau_{i}^{c_{\last}}}{t'})$ holds for every $i\in \fragmentco{0}{n_{0}-1}$ and some $t'\in \fragment{\tau_{i}^{c_{\last}}}{\tau_{i}^{c_{\last}+1}}$.
        \label{it:funcover:5}
    \end{enumerate}
    We define $\sum_{c=0}^{\bc(\bG_S)-1} \w_S(c)$ to be the \emph{total weight} of $\w_S$.
\end{definition}

To lay the groundwork for the construction of such a function, we first introduce some helpful notation.

\begin{definition}
    Given a position $x \in \fragmentco{0}{|P|}$, we define a function $\Pi_P(x)$ that maps $x$ to the largest position in $P$ that is less than or equal to $x$ and belongs to a black connected component. If no such character exists in $P$ preceding $x$, then $\Pi_P(x) = -1$. We define a symmetric function $\Pi_T(y)$ for all positions $y \in \fragmentco{0}{|T|}$.
\end{definition}

\begin{lemma} \label{lem:exfuncover}
    The following construction yields a weight function $\w_S : \fragmentco{0}{\bc(\bG_S)}\to
    \mathbb{Z}_{\ge 0}$ that covers $S$ of total weight at most $\mathcal{O}(k|S|)$.
    \begin{enumerate}[(1)]
        \item Initialize the weight function $\w_S(c) = 0$ for all $c\in \fragmentco{0}{\bc(\bG_S)}$.
        \item Iterate through each alignment $\mX \in S$ one by one. For each $\mX$, process all edits performed by $\mX$:
        \begin{enumerate}[(i)]
            \item If $\mX$ deletes $P\position{x}$, or substitutes $P\position{x}$ with $T\position{y}$, consider $\Pi_P(x)$.
    If $\Pi_P(x) = -1$, then increase $\w_S(\bc(\bG_S)-1)$ by one. Otherwise, increase $\w_S(c)$ by one, where $c$ is such that $\Pi_P(x) = \pi_j^c$.
            \label{lem:exfuncover:2a}
            \item  If $\mX$ inserts $T\position{y}$, consider $\Pi_T(y)$.
    If $\Pi_T(y) = -1$, then increase $\w_S(\bc(\bG_S)-1)$ by one. Otherwise, increase $\w_S(c)$ by one, where $c$ is such that $\Pi_T(y) = \tau_i^c$.
            \label{lem:exfuncover:2b}
        \end{enumerate}
    \end{enumerate}
\end{lemma}

\begin{proof}
    \newcommand{\bbG}{\bar{\bG}}
    \newcommand{\bbc}{\bar{\rm{bc}}}
    The total weight of $\w_S$ is clearly $\Oh(k|S|)$ because each edit in an alignment of $\mX$ increases the total weight of $\w_S$ by at most $\Oh(1)$. To demonstrate that $\w_S$ covers $S$, we first construct another function $\w_S' : \fragmentco{0}{\bc(\bG_S)} \to \mathbb{Z}_{\ge 0}$. We  show that $\w_S'$ covers $S$ and satisfies $\w_S'(c) \leq \w_S(c)$ for all $c \in \fragmentco{0}{\bc(\bG_S)}$. Since the properties defined in \cref{def:funcover} are all componentwise monotone, it follows that $\w_S$ also covers $S$.

    First, we introduce some useful notation. Let $\bbG_S$ be the subgraph of $\bG_S$ induced by all vertices $P\position{\pi_j^c}$ and $T\position{\tau_i^c}$ for $c\in \fragmentco{0}{\bc(\bG_S)}$, $j\in \fragmentco{0}{m_{c+1}}$, and $i\in \fragmentco{0}{n_{c+1}}$.
    In other words, we remove all red components as well as vertices $P\position{\pi_{m_0-1}^{c_{\last}}}$ and $T\position{\tau_{n_0-1}^{c_{\last}}}$ from the $c_{\last}$-th black component.
    By doing so, the edges of $\bbG_S$  be exactly those for which the association observed in \cref{rmk:perstr}\eqref{rmk:perstr:4} is well defined.
    We prove that $\bbG_S$ preserves the same connectedness structure of black components as $\bG_S$.

    \begin{claim}\label{claim:funcover:0}
        For each $c\in \fragmentco{0}{\bc(\bG_S)}$, the graph $\bbG_S$ contains a connected component $\bbG_S^c$ induced by vertices $P\position{\pi_j^c}$ and $T\position{\tau_i^c}$ with $j\in \fragmentco{0}{m_{c+1}}$ and $i\in \fragmentco{0}{n_{c+1}}$.
    \end{claim}
    \begin{proof}
        Since $\bbG_S$ is a subgraph of $\bG_S$, the characterization of connected components of $\bG_S$ in \cref{lem:periodicity} implies that $\bbG_S$ does not contain any edge leaving $\bbG_S^c$.
        Thus, it suffices to prove that $\bbG_S^c$ is connected.

        For this, let us define strings $P_{\bbc}$ and $T_{\bbc}$ analogously to $P_{\rm{bc}}$ and $T_{\rm{bc}}$ in the proof of \cref{lem:periodicity}, but with respect to $\bbG_S$ rather than $\bG_S$.
        Specifically, $|P_{\bbc}|=|P_{|S}|-1$ and $|T_{\bbc}|=|T_{|S}|-1$, and the characters $P_{\bbc}[j]$ and $T_{\bbc}[i]$ are equal if and only if the corresponding characters $P_{|S}[j]$ and $T_{|S}[i]$ belong to the same connected component of $\bbG_S$.

        By \cref{clm:occ,clm:nosingle}, we have $\{0,|T_{\bbc}|-|P_{\bbc}|\}\subseteq \{\Delta_{\mX} : \mX\in S\} \subseteq \Occ(P_{\bbc},T_{\bbc})$ and $|T_{\bbc}|=|T_{|S}|-1 \le 2|P_{|S}|-1=2|P_{\bbc}|+1$.
        If $|P_{\bbc}|=0$, then $|T_{\bbc}|\le 1$, and $\bbG_S$ is either empty or consists of a single vertex constituting the connected component $\bbG_S^0$.
        Otherwise, \cref{fct:periodicity} implies that $\gcd(\Occ(P_{\bbc},T_{\bbc}))$ is a period of $T_{\rm{bc}}$, and thus also $\bc(\bG_S)=\gcd\{\Delta_{\mX} : \mX\in S\}$ is a period of both $T_{\bbc}$ and of its prefix $P_{\bbc}$.
        Consequently, $\bbG_S^c$ is a connected subgraph of $\bbG_S$.
    \end{proof}

    Let us also assign weights to edges of $\bbG_S$.
    If $\bbG_S$ contains an edge between $P\position{\pi_j^c}$ and $T\position{\tau_i^c}$,
    we define its weight to be the smallest value $\edal{\mX}(P\fragmentco{\pi_j^c}{\pi_j^{c+1}}, T\fragmentco{\tau_i^c}{\tau_i^{c+1}})$
    such that $\mX \in S$ and $(\pi_j^c, \tau_i^c) \in \mX$.
    By $\w(\bbG_S^c)$, we denote the total weight of edges in $\bbG_S^c$.

    \begin{claim}
        \label{claim:funcover:1}
        If $\w_S'(c) \geq \w(\bbG_S^c)$ for all $c \in \fragmentco{0}{\bc(\bG_S)}$, then property \eqref{it:funcover:1} holds.
    \end{claim}
    \begin{claimproof}
        By \cref{claim:funcover:0}, $\bbG_S^c$ is connected for all $c \in \fragmentco{0}{\bc(\bG_S)}$.
        Thus, for any
        $c \in \fragmentco{0}{\bc(\bG_S)}$, $j\in \fragmentco{0}{m_{c+1}}$, $i \in \fragmentco{0}{n_{c+1}}$,
        we can construct an alignment of cost at most $\w(\bbG_S^c)$
        aligning $P\fragmentco{\pi_j^c}{\pi_j^{c+1}}$ onto $T\fragmentco{\tau_i^{c}}{\tau_i^{c+1}}$.
        To obtain such alignment, we compose
        the alignments inducing the weights on the edges
        contained in the path in $\bbG_S^c$ from $P\position{\pi_j^c}$ to $T\position{\tau_i^{c}}$.
    \end{claimproof}

    Now, we formally define $\w_S'$.
    We set $\w_S'(c) = \w(\bbG_S^c)$ for all $c \in \fragmentco{0}{\bc(\bG_S)} \setminus \{c_{\last}, \bc(\bG_S)-1\}$.
    Furthermore, let $\alpha, \alpha' \in \mathbb{Z}_{\geq 0}$ be determined later.
    If $\bc(\bG_S)-1 \neq c_{\last}$, then we set
    $\w_S'(\bc(\bG_S)-1) = \w(\bbG_S^{\bc(\bG_S)-1}) + \alpha$ and $\w_S'(c_{\last}) = \w(\bbG_S^{c_{\last}}) + \alpha'$.
    Otherwise, if $\bc(\bG_S)-1 = c_{\last}$,
    set $\w_S'(\bc(\bG_S)-1) = \w(\bbG_S^{\bc(\bG_S)-1}) + \alpha + \alpha'$.
    From $\alpha, \alpha' \geq 0$ and \cref{claim:funcover:1},
    follows that property~\eqref{it:funcover:1} holds.

    We  choose $\alpha, \alpha'$ large enough
    to ensure that also properties~\eqref{it:funcover:2},~\eqref{it:funcover:3},~\eqref{it:funcover:4}, and~\eqref{it:funcover:5} hold.
    For that purpose, consider alignments $\mXpref, \mXsuf \in S$ such that $(0,0) \in \mXpref$ and $(|P|,|T|) \in \mXsuf$.
    First, let us define $\alpha$ as $\alpha \coloneqq  \edal{\mXpref}(P\fragmentco{0}{\pi_0^0}, T\fragmentco{0}{\tau_0^0}) + \edal{\mXsuf}(P\fragmentco{0}{\pi_0^0}, T\fragmentco{y}{\tau_{\hi}^{0}})$,
    assuming $\mXsuf$ aligns $P\fragmentco{0}{\pi_0^0}$ onto $T\fragmentco{y}{\tau_{\hi}^{0}}$ for some $\hi \in \fragmentco{0}{n_0}$ and $y \in \fragment{\tau_{\hi-1}^{\bc(\bG_S)-1}}{\tau_{\hi}^0}$.

    \begin{claim}
        \label{claim:funcover:2}
        Properties~\eqref{it:funcover:2},~\eqref{it:funcover:3} hold.
    \end{claim}
    \begin{claimproof}
        Clearly, property~\eqref{it:funcover:2} holds because
        \[
            \w_S'(\bc(\bG_S)-1) \geq \alpha \geq \edal{\mXpref}(P\fragmentco{0}{\pi_0^0}, T\fragmentco{0}{\tau_0^0}) \geq \ed(P\fragmentco{0}{\pi_0^0}, T\fragmentco{0}{\tau_0^0}).
        \]
        Regarding property~\eqref{it:funcover:3},
        consider the case where $(\pi_0^0, \tau_0^0) \in \mXsuf$, that is, $\hi = 0$.
        This implies $n_0=1$, as all other alignments in $S$ also align $P[\pi_0^0]$ with $T[\tau_0^0]$.
        In this case, there is nothing to prove, as the quantifier $i$ in property~\eqref{it:funcover:3} ranges over an empty set.
        It remains to cover the case when $\hi \in \fragmentco{1}{n_0}$. Choose an arbitrary $i \in \fragmentco{1}{n_0}$.
        Since $\tau_{\hi-1}^{\bc(\bG_S)-1}$ and $\tau_{i-1}^{\bc(\bG_S)-1}$ belong to the same connected component $\bbG_S^{\bc(\bG_S)-1}$, we can employ a similar argument as in \cref{claim:funcover:1}.
        Thus, there exists an alignment $\mY$ with a cost of at most $\w(\bbG_S^{\bc(\bG_S)-1})$, aligning $T\fragmentco{\tau_{\hi-1}^{\bc(\bG_S)-1}}{\tau_{\hi}^{0}}$ to $T\fragmentco{\tau_{i-1}^{\bc(\bG_S)-1}}{\tau_{i}^{0}}$.
        Assume $\mY(T\fragmentco{y}{\tau_{\hi}^{0}}) = T\fragmentco{t}{\tau_{i}^{0}}$ for some $t\in \fragment{\tau_{i-1}^{\bc(\bG_S)-1}}{\tau_i^0}$.
        By composing $\mXsuf : P\fragmentco{0}{\pi_0^0} \onto T\fragmentco{y}{\tau_{\hi}^{0}}$ and $\mY : T\fragmentco{y}{\tau_{\hi}^{0}} \onto T\fragmentco{t}{\tau_{i}^{0}}$, we obtain an alignment $P\fragmentco{0}{\pi_0^0} \onto T\fragmentco{t}{\tau_{i}^{0}}$ of a cost of at most
        \[
            \edal{\mXsuf}(P\fragmentco{0}{\pi_0^0}, T\fragmentco{t}{\tau_{\hi}^{0}}) + \w(\bbG_S^{\bc(\bG_S)-1}) \leq \alpha + \w(\bbG_S^{\bc(\bG_S)-1}) = \w_S'(\bc(\bG_S)-1).\claimqedhere
        \]
    \end{claimproof}

    Symmetrically, we set
    $\alpha' \coloneqq  \edal{\mXsuf}(P\fragmentco{\pi_{m_{0}-1}^{c_{\last}}}{|P|}, T\fragmentco{\tau_{n_{0}-1}^{c_{\last}}}{|T|}) + \edal{\mXpref}(P\fragmentco{\pi_{m_{0}-1}^{c_{\last}}}{|P|}, T\fragmentco{\tau_{\hi'}^{c_{\last}}}{y'})$,
    assuming $\mXpref$ aligns $P\fragmentco{\pi_{m_{0}-1}^{c_{\last}}}{|P|}$ onto $T\fragmentco{\tau_{\hi'}^{c_{\last}}}{y'}$
    for some $\hi'\in \fragmentco{0}{n_{0}}$ and $y'\in \fragment{\tau_{\hi'}^{c_{\last}}}{\tau_{\hi'+1}^{c_{\last}}}$.
    Using a symmetric argument to \cref{claim:funcover:2}, we obtain that also
    properties~\eqref{it:funcover:4} and~\eqref{it:funcover:5} hold.

    \begin{claim}\label{claim:funcover:3}
        $\w_S'(c) \leq \w_S(c)$ for all $c \in \fragmentco{0}{\bc(\bG_S)}$.
    \end{claim}
    \begin{claimproof}
        The total cost of $\w_S'$ is the sum of the cost of partial alignments from $S$.
         More specifically, the total cost of $\w_S'$ is the sum of the following terms:
    \begin{enumerate}[1.]
        \item $\sum_{c=0}^{\bc(\bG_S)-1} \w(\bbG_S^c)$, where each of the $\w(\bbG_S^c)$
        is the sum of some $\edal{\mX}(P\fragmentco{\pi_j^c}{\pi_j^{c+1}}, T\fragmentco{\tau_i^c}{\tau_i^{c+1}})$
        for some distinct triplets $(\mX, i, j)$ such that $\mX \in S$, $j\in \fragmentco{0}{m_{c+1}}$ and $i\in \fragmentco{0}{n_{c+1}}$;
        \label{claim:funcover:3:1}
        \item $\edal{\mXpref}(P\fragmentco{0}{\pi_0^0}, T\fragmentco{0}{\tau_0^0})$;
        \label{claim:funcover:3:2}
        \item $\edal{\mXsuf}(P\fragmentco{0}{\pi_0^0}, T\fragmentco{y}{\tau_{\hi}^{0}})$;
        \label{claim:funcover:3:3}
        \item $\edal{\mXsuf}(P\fragmentco{\pi_{m_{0}-1}^{c_{\last}}}{|P|}, T\fragmentco{\tau_{n_{0}-1}^{c_{\last}}}{|T|})$; and
        \label{claim:funcover:3:4}
        \item $\edal{\mXpref}(P\fragmentco{\pi_{m_{0}-1}^{c_{\last}}}{|P|}, T\fragmentco{\tau_{\hi'}^{c_{\last}}}{y'})$.
        \label{claim:funcover:3:5}
    \end{enumerate}
    Note that all alignments that are (possibly) involved in the sum are disjoint and are partial alignments that belong to $S$.
    To prove the claim, it suffices to show that
    each edit contained in one of these partial alignments corresponds to a unique increase of $\w_S$ by one in step \eqref{lem:exfuncover:2a} or \eqref{lem:exfuncover:2b}.
    \begin{description}
        \item[\eqref{claim:funcover:3:1}:] Each edit involved in $\edal{\mX}(P\fragmentco{\pi_j^c}{\pi_j^{c+1}}, T\fragmentco{\tau_i^c}{\tau_i^{c+1}})$ for some distinct triplet $(\mX, i, j)$ corresponds to the increase by one of $\w_S(c)$ when that edit in $\mX$ is processed.
        \item[\eqref{claim:funcover:3:2}, \eqref{claim:funcover:3:3}:] Each edit involved in $\edal{\mXpref}(P\fragmentco{0}{\pi_0^0}, T\fragmentco{0}{\tau_0^0})$ or  $\edal{\mXsuf}(P\fragmentco{0}{\pi_0^0}, T\fragmentco{y}{\tau_{\hi}^{0}})$ corresponds
        to the increase by one of $\w_S(\bc(\bG_S)-1)$ when that edit in $\mXsuf,\mXpref$ is processed, respectively.
        \item[\eqref{claim:funcover:3:4}, \eqref{claim:funcover:3:5}:] Each edit involved in $\edal{\mXsuf}(P\fragmentco{\pi_{m_{0}-1}^{c_{\last}}}{|P|}, T\fragmentco{\tau_{n_{0}-1}^{c_{\last}}}{|T|})$ or $\edal{\mXpref}(P\fragmentco{\pi_{m_{0}-1}^{c_{\last}}}{|P|}, T\fragmentco{\tau_{\hi'}^{c_{\last}}}{y'})$ corresponds
        to the increase of $\w_S(c_{\last})$ when that edit in $\mXsuf,\mXpref$ is processed, respectively. \claimqedhere
    \end{description}
    \end{claimproof}
    This concludes the proof of \cref{lem:exfuncover}.
\end{proof}

In the remaining part of \cref{sec:combres}, we let $\w_S$ denote a \emph{weight
function} that covers $S$ of total weight at most $w$.
We do not restrict ourselves to the case $w=\mathcal{O}(k|S|)$ to make the following
results more general.

\subsection{Period Covers}

In the following subsection, we formalize the characters contained in the black components
that we need to learn.
We specify this through a subset of indices $C_S \subseteq \fragmentco{0}{\bc(\bG_S)}$ of black components,
which we  refer to as \emph{period cover}.

\begin{definition}\label{def:periodcover_alt}
    A set $C_S \subseteq \fragmentco{0}{\bc(\bG_S)}$ is a \emph{period cover with respect
    to $\w_S$} if $\fragment{a}{b}\subseteq C_S$ holds for every interval
    $\fragment{a}{b}\subseteq \fragmentco{0}{\bc(\bG_S)}$
    such that at least one of the following holds:
    \begin{enumerate}
        \item $\selfed(T\fragment{\tau_0^a}{\tau_0^{b}}) \le 6w+11k$ and $a=0$;
            \label{it:periodcover_alt:1}
        \item $\selfed(T\fragment{\tau_0^a}{\tau_0^{b}}) \le 6w+11k$ and $b=\bc(\bG_S)-1$;
            \label{it:periodcover_alt:2}
        \item $\selfed(T\fragment{\tau_0^a}{\tau_0^{b}}) \le 6w+11k$ and $b=c_{\last}$;
            \label{it:periodcover_alt:3}
        \item $\selfed(T\fragment{\tau_0^a}{\tau_0^{b}}) \le 6w+11k$ and $a=c_{\last}+1$; or
            \label{it:periodcover_alt:4}
        \item $\selfed(T\fragment{\tau_0^a}{\tau_0^{b}}) \le 6 \sum_{c=a-1}^{b} \w_S(c)$.
            \label{it:periodcover_alt:5}  \qedhere
    \end{enumerate}
\end{definition}

In the remainder of this (sub)section,
we illustrate that by constructing a period cover $C_S$ in two different ways,
we can effectively (w.r.t. parameters $w$, $k$ and $|S|$) encode the information $\{(c, T[\tau_0^c]):c\in C_S\}$.

First, we establish that this is possible if we construct $C_S$
by straightforwardly verifying whether all intervals $\fragment{a}{b}$
satisfy any of the conditions \eqref{it:periodcover_alt:1},\eqref{it:periodcover_alt:2},\eqref{it:periodcover_alt:3},\eqref{it:periodcover_alt:4},\eqref{it:periodcover_alt:5}
of \cref{def:periodcover_alt}.

\begin{lemma} \label{prp:encode_simple_funcover}
    Let ${\{\fragment{a_i}{b_i}\}}_{i=0}^{\ell-1}$ denote the set that contains all intervals that satisfy any
    of the conditions \eqref{it:periodcover_alt:1},\eqref{it:periodcover_alt:2},\eqref{it:periodcover_alt:3},\eqref{it:periodcover_alt:4},\eqref{it:periodcover_alt:5}
    of \cref{def:periodcover_alt}.
    Consider the period cover $C_S = \bigcup_{i=0}^{\ell-1} \fragment{a_i}{b_i}$.
    Then, we can encode $\{(c, T[\tau_0^c]):c\in C_S\}$ using $\mathcal{O}(w + k|S|)$ space.
\end{lemma}

\begin{proof}
    First, we briefly argue that it is possible to encode $\{(c, \tau_0^c):c\in C_S\}$
    using $\mathcal{O}(k|S|)$ space.
    As already argued before, by storing for every $\mX \in S$
    the corresponding edit information,
    we can totally reconstruct $\bG_{S}$.
    This allows us to check whether a character of $T$ is contained in a black connected component.
    Moreover, given it is not contained in a red component,
    we can infer in which black component it is contained.
    Since the sets of edit information can be stored using $\mathcal{O}(k)$ space
    (through the compressed edit information)
    for every $\mX \in S$,
    it follows that we can encode the corresponding information using $\mathcal{O}(k|S|)$ space.

    Therefore, to prove the lemma, it suffices to show that it is possible to select a subset of intervals
    indexed by $I \subseteq \fragmentco{0}{\ell}$ such that $\bigcup_{i\in I} \fragment{a_i}{b_i} = C_S$
    and that ${\{T\fragment{\tau_0^{a_i}}{\tau_0^{b_i}}\}}_{i\in I}$ can be efficiently encoded.

    Among all intervals that satisfy condition~\eqref{it:periodcover_alt:1} of \cref{def:periodcover_alt},
    that is, among all intervals $\fragment{a_i}{b_i}$ such that $a_i=0$, we add to $I$ the
    index $j$ that maximizes $b_j$.
    Clearly, for all other intervals $\fragment{a_i}{b_i}$ such that $a_i=0$,
    we have $\fragment{a_i}{b_i} \subseteq \fragment{a_j}{b_j}$, and we can safely leave them out from $I$.
    Moreover, from $\selfed(T\fragment{\tau_0^{a_j}}{\tau_0^{b_j}}) \le 6w+11k$
    and from \cref{prp:lz_selfed} follows that
    $|\LZ(T\fragment{\tau_0^{a_j}}{\tau_0^{b_j}})| \leq \mathcal{O}(k + w)$.
    By using a similar approach for the intervals that satisfy any of the
    conditions~\eqref{it:periodcover_alt:2},~\eqref{it:periodcover_alt:3},~\eqref{it:periodcover_alt:4},
    henceforth, we may assume that we have taken care of all the intervals
    that satisfy any of the
    conditions~\eqref{it:periodcover_alt:1},~\eqref{it:periodcover_alt:2},~\eqref{it:periodcover_alt:3},~\eqref{it:periodcover_alt:4},
    and that ${\{\fragment{a_i}{b_i}\}}_{i=0}^{\ell-1}$ only contains
    intervals that satisfy~\eqref{it:periodcover_alt:5}.

    Now, we proceed by iteratively adding indices to $I$.
    The first index we add to $I$ is $\argmin_i \{a_i\}$.
    Next, we continue to add indices to $I$ based on the last index we added, $j$, as follows:
    \begin{itemize}
        \item if $\{i : a_j < a_i \leq b_j < b_i\} \neq \emptyset$, then we add the index $\argmax_i \{b_i : a_j < a_i \leq b_j < b_i\}$ to $I$;
        \item otherwise, we add the index $\argmin_i \{a_i : b_j < a_i\}$ to $I$.
    \end{itemize}
    We terminate if we can not add any index to $I$ anymore, that is, if $b_j = \max_i \{b_i\}$.
    Note, this iterative selection process is guaranteed to end, because
    in every iteration for the $i$ that we add to $I$,
    we have $a_j < a_i$.

    We want to argue ${\{T\fragment{\tau_0^{a_i}}{\tau_0^{b_i}}\}}_{i\in I} = C_S$.
    To do this, consider an interval $\fragment{x}{y}$ in the interval representation of $C_S$, that is,
    the representation of $C_S$ using the smallest number of non-overlapping intervals.
    We observe that an index $i$ with $a_i = x$ is added either initially or through the first rule.
    Subsequently, the selection process applies a combination of the two rules until we include an index $i$ with $b_i = y$ into $I$.

    Finally, we want to show that
    $\sum_{i \in I} \selfed(T\fragment{\tau_0^{a_i}}{\tau_0^{b_i}}) \leq 24w$.
    This is sufficient to prove that ${\{T\fragment{\tau_0^{a_i}}{\tau_0^{b_i}}\}}_{i\in I}$ can be encoded efficiently,
    because from \cref{prp:lz_selfed} follows
    \[
        \sum_{i \in I} |\LZ(T\fragment{\tau_0^{a_i}}{\tau_0^{b_i}})| \leq 2 \sum_{i \in I} \selfed(T\fragment{\tau_0^{a_i}}{\tau_0^{b_i}}) = \mathcal{O}(w).
    \]
    Consider three indices $i,i',i''$ that are added one after the other to $I$.
    If for the selection of $i'$ or $i''$ we need to apply the second rule, then clearly $b_i < a_{i''}$.
    Otherwise, if always apply the first rule, then from $b_{i'} < b_{i''}$ follows that $b_i < a_{i''}$,
    because otherwise, we would have selected $i''$ instead of $i'$ in the $\argmin$.

    Hence, for every $c \in \fragmentco{0}{\bc(\bG_S)}$
    it holds $|\{i \in I : c \in \fragment{a_i}{b_i}\}| \leq 2$,
    from which follows $|\{i \in I : c \in \fragment{a_i-1}{b_i}\}| = |\{i \in I : c \in \fragment{a_i}{b_i}\}| + |\{i \in I : c + 1 = a_i\}| \leq 4$.
    Note, the previous property also holds when we add less than three indices to $I$.
    Using a double counting argument, we obtain
    \begin{align*}
        \sum_{i \in I} \selfed(T\fragment{\tau_0^{a_i}}{\tau_0^{b_i}})
        \leq 6 \sum_{i \in I} \sum_{c=a_i-1}^{b_i} \w_S(c)
        = 6\sum_{c=0}^{\bc(\bG_S)-1} \sum_{\substack{i \in I \\ c \in \fragment{a_i-1}{b_i}}} \w_S(c)
        \leq 24w,
    \end{align*}
    which concludes the proof.
\end{proof}

One drawback of this approach is the necessity to read all characters within the black components prior to encoding.
To circumvent this, we offer a second alternative construction of a period cover $C_S$ designed
to minimize the number of queries to the string.
Such a strategy proves particularly advantageous in the quantum setting.
In the remaining part of this paper, we abuse notation and write $\w_S(-1) = \w_S(\bc(\bG_S)-1)$.

\begin{lemma}
    \label{def:periodcover}
    The following (constructive) definition delivers a period cover $C_S$.
    Set
    \[
        C_S \coloneqq   \fragment{0}{c_{\mathrm{pref}}} \cup \fragmentco{c_{\mathrm{suff}}}{\bc(\bG_S)} \cup \fragment{c_{\mathrm{lsuff}}}{c_{\mathrm{lpref}}} \cup C^{0,\bc(\bG_S)-1}_S,
    \]
    where $c_{\mathrm{pref}}$, $c_{\mathrm{suff}}$, $c_{\mathrm{lsuff}}$, $c_{\mathrm{lpref}}$, and $C^{0,\bc(\bG_S)-1}_S$ are defined as follows.
    \begin{enumerate}[(a)]

        \item Let  $c_{\mathrm{pref}} \in \fragmentco{0}{\bc(\bG_S)}$ denote the largest index such that $|\LZ(T\fragment{\tau_0^0}{\tau^{c_{\mathrm{pref}}}_0})| \le 12w+22k$.
            \label{it:periodcover:a}
        \item Let  $c_{\mathrm{suff}} \in \fragmentco{0}{\bc(\bG_S)}$ denote the smallest index such that $|\LZ(\rev{T\fragment{\tau^{c_{\mathrm{suff}}}_0}{\tau^{\bc(\bG_S)-1}_0}})| \le 12w+22k$.
            \label{it:periodcover:b}
        \item Let  $c_{\mathrm{lsuff}} \in \fragment{0}{c_\last}$ denote the smallest index such that $|\LZ(\rev{T\fragment{\tau^{c_{\mathrm{lsuff}}}_0}{\tau_{0}^{c_{\mathrm{last}}}}})| \le 12w+22k$.
            \label{it:periodcover:c}
        \item Let  $c_{\mathrm{lpref}} \in \fragmentco{c_\last+1}{\bc(\bG_S)}$ denote the largest index such that $|\LZ(T\fragment{\tau^{c_\last+1}_0}{\tau_{0}^{c_{\mathrm{lpref}}}})| \le 12w+22k$.
            \label{it:periodcover:d}
        \item For \( \fragment{i}{j} \subseteq \fragmentco{0}{\bc(\bG_S)} \), we define \( C_S^{i,j} \subseteq \fragmentco{0}{\bc(\bG_S)} \) as follows.
        \begin{itemize}
            \item If \( \sum_{c=i-1}^{j} \w_S(c) = 0 \), then \( C_S^{i,j} \coloneqq  \emptyset \).
            \item If \( i=j \), then \( C_S^{i,j} \coloneqq  \{i\} \).
            \item Otherwise, let \( h \coloneqq  \floor{(i+j)/2} \).
            Let \( i' \in \fragment{i}{h} \) denote the smallest index such that \( |\LZ(\rev{T\fragment{\tau^{i'}_0}{\tau^h_0}})| \le 12 \sum_{c=i-1}^{j} \w_S(c) \).
            Similarly, let \( j' \in \fragment{h}{j} \) denote the largest index such that \( |\LZ(T\fragmentoc{\tau^h_0}{\tau^{j'}_0})| \le 12 \sum_{c=i-1}^{j} \w_S(c) \).
            In this case, we define recursively \( C_S^{i,j} \coloneqq  C_S^{i,h} \cup C_S^{h+1,j} \cup \fragment{i'}{j'} \).
        \end{itemize}
            \label{it:periodcover:e}
    \end{enumerate}

\end{lemma}

\begin{proof}
    Let $\fragment{a}{b} \subseteq  \fragmentco{0}{\bc(\bG_S)}$ be an interval satisfying
    at least one of the conditions of \cref{def:periodcover_alt}.
    We need to prove that $\fragment{a}{b} \subseteq C_S$.
    Assume $\fragment{a}{b}$ satisfies  \cref{def:periodcover_alt}\eqref{it:periodcover_alt:1},
    that is, $a=0$ and $\selfed(T\fragment{\tau_0^0}{\tau_0^b}) \le 6w+11k$.
    By \cref{prp:lz_selfed} we conclude
     \[
        |\LZ(T\fragment{\tau_0^0}{\tau_0^{b}})| \le 2\selfed(T\fragment{\tau_0^0}{\tau_0^{b}}) \le 12w+22k.
    \]
    Consequently, $\fragment{0}{b}\subseteq \fragment{0}{c_{\mathrm{pref}}} \subseteq C_S$.
    Since $\selfed$ is insensitive to string reversal, similar arguments apply for \cref{def:periodcover_alt}\eqref{it:periodcover_alt:2}\eqref{it:periodcover_alt:3}\eqref{it:periodcover_alt:4}.

    Now, assume \cref{def:periodcover_alt}\eqref{it:periodcover_alt:5} holds for $\fragment{a}{b}$ of size at least two.
    Note, for such interval we have $\selfed(T\fragment{\tau_0^a}{\tau_0^{b}}) > 0$.
    We want to argue that there exist indices $i,j \in \fragmentco{0}{\bc(\bG_S)}$ fulfilling all the three following properties:
    \begin{itemize}
        \item $C^{i,j}_S$ is involved in the recursive construction of $C^{0,\bc(\bG_S)}_S$;
        \item $\fragment{a}{b}\subseteq \fragment{i}{j}$; and
        \item $h \in \fragmentco{a}{b}$ for $h \coloneqq  \floor{(i+j)/2}$.
    \end{itemize}
    For that purpose, note that, in the recursive definition, we start with an interval $\fragment{i}{j}=\fragmentco{0}{\bc(\bG_S)}$ which clearly contains $\fragment{a}{b}$.
    At each step, we recurse on $\fragment{i}{h}$ and $\fragment{h+1}{j}$, where $h=\floor{(i+j)/2}$.
    If $\fragment{a}{b}$ is contained in neither of these intervals, then $i \le a \le h < h+1 \le b \le j$
    and, in particular, $h\in \fragmentco{a}{b}$.
    Consequently, $\tau^a_0 \le \tau^h_0 < \tau^b_0$.
    By \cref{prp:lz_selfed} we conclude that
    \[
       |\LZ(\rev{T\fragment{\tau_0^a}{\tau_0^h}})| \le 2\selfed(T\fragment{\tau_0^a}{\tau_0^{b}}) \le 12\sum_{c=a-1}^b \w_S(c)
       \le 12 \sum_{c=i-1}^{j} \w_S(c).
   \]
   Similarly,
   \[
    |\LZ(T\fragmentoc{\tau_0^h}{\tau_0^b})| \le 2\selfed(T\fragment{\tau_0^a}{\tau_0^{b}})  \le 12\sum_{c=a-1}^b \w_S(c)
    \le 12 \sum_{c=i-1}^{j} \w_S(c).
    \]
   Consequently, $\fragment{a}{b}$ is contained in the interval $\fragment{i'}{j'}\subseteq \fragment{i}{j}$ that is added to $C_S^{i,j}\subseteq C_S$.

   Finally, if $\fragment{a}{a}$ satisfies \cref{def:periodcover_alt}\eqref{it:periodcover_alt:5},
   then $\sum_{c=i-1}^j \w_S(c) > 0$ holds for every interval $\fragment{i}{j}$
   such that $\fragment{a}{a} \subseteq \fragment{i}{j} \subseteq \fragmentco{0}{\bc(\bG_S)}$.
   Consequently, $a\in C_S^{a,a}\subseteq C_S^{0,\bc(\bG_S)-1}$.
\end{proof}

\begin{lemma}
    \label{prp:encode_funcover}
    Let $C_S$ be a period cover obtained as described in \cref{def:periodcover}.
    Then, the information $\{(c, T[\tau_0^c]):c\in C_S\}$ can be encoded in $\mathcal{O}(w\log n + k|S|)$ space.
\end{lemma}

\begin{proof}
    As already argued in the proof of \cref{prp:encode_simple_funcover}, we can
    encode the information $\{(c, \tau_0^c): c\in C_S\}$ using $\mathcal{O}(k|S|)$
    space. Therefore, it suffices to show that
    we can encode in $\mathcal{O}(w\log n + k)$ space
    the intervals $\fragment{\tau^a_0}{\tau^b_0}$ for all $\fragment{a}{b} \subseteq \fragmentco{0}{\bc(\bG_S)}$
    that we add to $C_S$ in any of
    \cref{def:periodcover}\eqref{it:periodcover:a}\eqref{it:periodcover:b}\eqref{it:periodcover:c}\eqref{it:periodcover:e}\eqref{it:periodcover:d}.

    First, note that $\fragment{0}{c_{\mathrm{pref}}}$ can be encoded
    in $\Oh(w+k)$ space since $|\LZ(T\fragment{\tau_0^0}{\tau^{c_{\mathrm{pref}}}_0})| \le 12w+22k$.
    Using similar arguments, we get that the four intervals that we add in any of
    \cref{def:periodcover}\eqref{it:periodcover:a}\eqref{it:periodcover:b}\eqref{it:periodcover:c}\eqref{it:periodcover:d}
    can be encoded in $\mathcal{O}(w + k)$ space.
    It remains to show that all intervals that we add
    in \cref{def:periodcover}\eqref{it:periodcover:e},
    that is, in the recursive construction of $C^{0,\bc(\bG_S)-1}_S$,
    can be encoded in $\mathcal{O}(w\log n)$ space.
    Let $\fragment{i_0}{j_0}, \ldots, \fragment{i_{d-1}}{j_{d-1}}$
    be all the intervals considered by the recursion at a fixed depth.
    Observe that $\sum_{r=0}^{d-1} \sum_{c=i_r-1}^{j_r} \w_S(c) \leq 2w$,
    because for every $c \in \fragmentco{0}{\bc(\bG)-1}$ the term $\w_S(c)$
    does not appear more than twice in the summation.
    Suppose that for $\fragment{i}{j} \in \{\fragment{i_r}{j_r}\}_{r=0}^{d-1}$
    we add $\fragment{i'}{j'}$ to $C_S^{i,j}$.
    Then, we can encode the corresponding
    $T\fragment{\tau_{0}^{i'}}{\tau_{0}^{j'}}$ using $20 \sum_{c=i-1}^{j} \w_S(c)$ space.
    This is because there exists $h$ such that $|\LZ(\rev{T\fragment{\tau_{i'}^0}{\tau_h^0}})| \le 10 \sum_{c=i-1}^{j} \w_S(c)$
    and $|\LZ(T\fragmentoc{\tau_h^0}{\tau_{j'}^0})| \le 10 \sum_{c=i-1}^{j} \w_S(c)$.
    As a consequence, we obtain that we can encode all intervals that we add to $C^{0,\bc(\bG_S)-1}_S$ at a fixed depth
    using at most $\sum_{r=0}^{d-1} 20 \sum_{c=i_r-1}^{j_r} \w_S(c) = \mathcal{O}(w)$ space.
    Since the recursion has depth at most $\mathcal{O}(\log n)$, we conclude that
    the intervals that we add in \cref{def:periodcover}\eqref{it:periodcover:e}
    can be encoded in $\mathcal{O}(w\log n)$ space.
\end{proof}

\subsection{Close Candidate Positions}\label{sec:close}

This (sub)section is devoted to proving the following result:
\prpclose*

\subsubsection{Partition of \texorpdfstring{$P$}{P} and \texorpdfstring{$T$}{T} into Blocks}

\cref{prp:close} heavily relies on the fact that $P$ and $T$ exhibit a periodic structure.
We divide $P$ and $T$ into the blocks ${\{P_j\}}_{j \in \fragmentco{0}{m_0}}$ and ${\{T_i\}}_{i \in \fragmentco{0}{n_0}}$ having a similar structure.

\begin{definition}\label{def:blocks}
    For $j \in \fragmentco{0}{m_0}$ we define
    \[
        P_j \coloneqq
        \begin{cases}
            P\fragmentco{\pi^{0}_{j}}{\pi^{0}_{j+1}} & \text{if $j \neq m_{0}-1$,} \\
            P\fragment{\pi^{0}_{m_0-1}}{\pi^{c_{\last}}_{m_0-1}} & \text{otherwise.}
        \end{cases}
    \]
    Similarly, for $i \in \fragmentco{0}{n_0}$ we define
    \[
        T_i \coloneqq
        \begin{cases}
            T\fragmentco{\tau^{0}_{i}}{\tau^{0}_{i+1}} & \text{if $i \neq n_{0}-1$,} \\
            T\fragment{\tau^{0}_{n_0-1}}{\tau^{c_{\last}}_{n_0-1}} & \text{otherwise.}
        \end{cases}
        \qedhere
    \]
\end{definition}

The notion of similarity between these blocks is also captured by the fact that
for any $j,i$ we
can construct an alignment of $P_j$ onto $T_i$ whose cost is upper bounded by the total
weight of the weight function $\w_S$.
More formally, we can prove the following proposition.

\begin{proposition}
    \label{prop:asji}
    Let $i \in \fragmentco{0}{n_0}, j \in \fragmentco{0}{m_0}$ be arbitrary such that if
    $i = n_0-1$ then $j = m_0-1$.
    Then, there is an alignment $\mA_{S}^{j,i}$ with the following properties.
    \begin{enumerate}[(i)]
        \item $\mA_{S}^{j,i} : P_j \onto T_i$ if $j \neq m_0-1$,
            and $\mA_{S}^{j,i} : P_{m_0-1} \onto T\fragment{\tau_i^0}{\tau_i^{c_{\last}}}$ otherwise.
            \label{it:asji:i}
        \item Let $c \in \fragmentco{0}{\bc(\bG_S)}$ if $j \neq m_0-1$, and $c \in
            \fragment{0}{c_{\last}}$ if $j = m_0-1$.
            Then, $\mA_{S}^{j,i}$ matches $P\position{\pi_{j}^{c}}$ and $T\position{\tau_{i}^{c}}$.
            \label{it:asji:ii}
        \item
            $\mA_{S}^{j,i}$ has cost at most $w$.
            \label{it:asji:iii}
    \end{enumerate}
\end{proposition}

\begin{proof}
    Let $c\in \fragmentco{0}{\bc(\bG_S)}$ if $j \neq m_{0}-1$ and let $c \in
    \fragmentco{0}{c_{\last}}$ if $j = m_{0}-1$.
    From \Cref{def:funcover}\eqref{eq:funcover} follows that there is an
    alignment $\mX_c : P\fragmentco{\pi_j^{c}}{\pi_j^{c+1}} \onto T\fragmentco{\tau_i^{c}}{\tau_i^{c+1}}$
    with cost at most $\w_S(c)$.
    Note, since $P\position{\pi_j^{c}} = T\position{\tau_i^{c}}$, we
    may assume that $\mX_c$ matches $P\position{\pi_j^{c}}$ and $T\position{\tau_i^{c}}$.
    Now, it suffices to set
    \[
        \mA_{S}^{j,i} \coloneqq
        \begin{cases}
            \bigcup_{c=0}^{\bc(\bG_S)-1} \mX_c & \text{if $j \neq m_{0}-1$}, \\
            \left(\bigcup_{c=0}^{c_{\last}-1} \mX_c \right) \cup \{(\pi_{m_0-1}^{c_{\last}}+1, \tau_i^{c_{\last}}+1)\} & \text{if $j = m_{0}-1$}.
        \end{cases}
    \]
    Observe that \( \mA_{S}^{j,i} \) is a valid alignment because the last two characters aligned by \( \mX_{c-1} \) are exactly the first two characters aligned by \( \mX_{c} \).
    It is clear that~\eqref{it:asji:i} and~\eqref{it:asji:ii} hold.
    Regarding~\eqref{it:asji:iii}, note that the cost of \( \mA_{S}^{j,i} \) equals the sum of the costs of all \( \mX_{c} \) whose union defines \( \mA_{S}^{j,i} \).
    This sum is upper bounded by the sum of the corresponding \( \w_S(c) \), and is therefore at most \( w \).
\end{proof}

\begin{lemma}
    \label{prp:coverededit}
    Let $j \in \fragmentco{0}{m_0}$, $i \in \fragmentco{0}{n_0}$ such that if $i = n_0-1$ then $j = m_0-1$, and let $\mA \coloneqq  \mA_S^{j,i}$.
    Let $P\fragmentco{x}{x'}$ and $T\fragmentco{y}{y'}$ be fragments of $P_j$ and $T_i$, respectively, such that $\mA$ aligns $P\fragmentco{x}{x'}$ onto $T\fragmentco{y}{y'}$
    with cost $\delta^{\mA}(P\fragmentco{x}{x'}, T\fragmentco{y}{y'}) > 0$. Let $\mX$ be an optimal alignment of $P\fragmentco{x}{x'}$ onto $T\fragmentco{y}{y'}$.
    If there exists no $(\hx, \hy) \in \mX \cap \mA$ such that both $\mX$ and $\mA$ match $P[\hx]$ and $T[\hy]$, then $\{c\in \fragmentco{0}{\bc(\bG)} : \pi^c_j \in \fragmentco{x}{x'}\} \subseteq C_S$.
\end{lemma}

\begin{proof}
    Denote  $\fragment{a}{b}=\{c\in \fragmentco{0}{\bc(\bG)} : \pi^c_j \in \fragmentco{x}{x'}\}$
    and assume that this interval is non-empty (otherwise, there is nothing to prove).
    Moreover, let $\fragment{a'}{b}=\{c\in \fragmentco{0}{\bc(\bG)} : \pi^c_j \in \fragmentoo{x}{x'}\}$.
    Now, we can apply the following chain of inequalities
    \begin{align}
            \selfed(T\fragment{\tau_i^a}{\tau_i^b}) &\leq \selfed(T\fragmentco{y}{y'}) \label{eq:coverededit:1} \\
                                                                            &\leq 2 \edal{\mA}(P\fragmentco{x}{x'}, T\fragmentco{y}{y'}) \label{eq:coverededit:2} \\
                                                                            &\leq 2 \edal{\mA}(P\fragmentco{\pi_j^{a'-1}}{\pi_j^{b+1}}, T\fragmentco{\tau_i^{a'-1}}{\tau_i^{b+1}}) \label{eq:coverededit:3} \\
                                                                            &= 2 \sum_{c=a'-1}^{b}\edal{\mA}(P\fragmentco{\pi_j^{c}}{\pi_j^{c+1}}, T\fragmentco{\tau_i^{c}}{\tau_i^{c+1}}) \label{eq:coverededit:4} \\
                                                                            &\leq 2 \sum_{c=a'-1}^{b} \w_S(c) \label{eq:coverededit:5}\\
                                                                            &\leq 2 \sum_{c=a-1}^{b} \w_S(c)  \label{eq:coverededit:6}
    \end{align}
    where we have used
    \begin{itemize}
        \item[(\ref{eq:coverededit:1})] monotonicity of $\selfed$ (\cref{prp:prop_selfed});
        \item[(\ref{eq:coverededit:2})] \cref{prp:edimpliesselfed} and $\edal{\mX}(P\fragmentco{x}{x'}, T\fragmentco{y}{y'}) \leq \edal{\mA}(P\fragmentco{x}{x'}, T\fragmentco{y}{y'})$;
        \item[(\ref{eq:coverededit:3})] the definition of $\fragment{a'}{b}$ and the fact that  $(\pi_j^{a'-1}, \tau_i^{a'-1}),(\pi_j^{b+1}, \tau_i^{b+1}) \in \mA$;
        \item[(\ref{eq:coverededit:4})] the fact that $(\pi_j^{c}, \tau_i^{c}) \in \mA$ for all $c \in \fragment{a}{b}$; and
        \item[(\ref{eq:coverededit:5})] \cref{def:funcover}.
    \end{itemize}

    Now, notice that $\ed(T\fragment{\tau_i^{a}}{\tau_i^{b}}, T\fragment{\tau_0^{a}}{\tau_0^{b}}) \leq \sum_{c=a}^{b-1} \ed(T\fragment{\tau_i^{c}}{\tau_i^{c+1}}, T\fragment{\tau_0^{c}}{\tau_0^{c+1}}) \leq 2\sum_{c=a}^{b-1} \w_S(c)$ follows from \cref{def:funcover}.
    From this together with \cref{prp:lz_selfed} and  $\selfed(T\fragment{\tau_i^{a}}{\tau_i^{b}}) \leq 2 \sum_{c=a-1}^{b} \w_S(c)$ follows
    $\selfed(T\fragment{\tau_0^{a}}{\tau_0^{b}}) \leq 6 \sum_{c=a-1}^{b} \w_S(c)$.
    Hence, we can use \cref{def:periodcover_alt} to deduce $\fragment{a}{b} \subseteq C_S$.
\end{proof}

\subsubsection{Recovering the Edit Distance for a Single Candidate Position}

\begin{lemma}\label{prp:recperioded}
    Let $j\in \fragmentco{0}{m_0}$ and $i\in \fragmentco{0}{n_0}$ be such that $j=m_0-1$ if $i=n_0-1$.
    Let $\mX: P_j \onto T\fragmentco{y}{y'}$ be an optimal alignment of $P_j$ onto an arbitrary fragment $T\fragmentco{y}{y'}$.
    If $|\tau_i^0 - y|  \leq w + 4k$ and $\ed(P_j, T\fragmentco{y}{y'})\le k$,
    then $\mX$ aligns $P[\pi^c_j]$ to $T[\tau^c_{i}]$ for all $c\in \fragmentco{0}{\bc(\bG_S)}\setminus C_S$ such that $j\in \fragmentco{0}{m_c}$.
\end{lemma}

\begin{proof}
    We first assume that $j< m_0-1$ and then briefly argue that the case of $j=m_0-1$ can be handled similarly.
    Assume $\fragmentco{0}{\bc(\bG_S)}\setminus C_S \neq \emptyset$; otherwise there is nothing to prove.
    Denote $c_{\mathrm{l}}= \min (\fragmentco{0}{\bc(\bG_S)}\setminus C_S)$ and $c_{\mathrm{r}} = \max(\fragmentco{0}{\bc(\bG_S)}\setminus C_S)$.
    Henceforth, we set $\mA \coloneqq  \mA_{S}^{j,i}$.

    \begin{claim}
        There exist $(x_{\mathrm{l}},y_{\mathrm{l}}), (x_{\mathrm{r}},y_{\mathrm{r}}) \in \mA \cap \mX$
        such that $(x_{\mathrm{l}},y_{\mathrm{l}}) \leq (\pi_j^{c_{\mathrm{l}}}, \tau_i^{c_{\mathrm{l}}})$
        and $(\pi_j^{c_{\mathrm{r}}}, \tau_i^{c_{\mathrm{r}}}) \leq (x_{\mathrm{r}},y_{\mathrm{r}})$.
    \end{claim}

    \begin{claimproof}

    First, we want to argue that there exists $(x_{\mathrm{l}},y_{\mathrm{l}}) \in \mA \cap \mX$ such that $(x_{\mathrm{l}},y_{\mathrm{l}}) \leq (\pi_j^{c_{\mathrm{l}}}, \tau_i^{c_{\mathrm{l}}})$.
    Note, following \cref{def:periodcover_alt} we must have $\selfed(T\fragment{\tau_0^0}{\tau_0^{c_{\mathrm{l}}}}) > 6w + 11k$; otherwise $c_{\mathrm{l}}$ would have been included in $C_S$.
    Suppose that $\mX$ aligns $P\fragmentco{\pi_j^0}{\pi_j^{c_{\mathrm{l}}}}$ with $T\fragmentco{y}{\hy}$.
    For sake of contradiction, suppose that $\mA : P\fragmentco{\pi_j^0}{\pi_j^{c_{\mathrm{l}}}} \onto T\fragmentco{\tau_i^0}{\tau_i^{c_{\mathrm{l}}}}$ and $\mX : P\fragmentco{\pi_j^0}{\pi_j^{c_{\mathrm{l}}}} \onto T\fragmentco{y}{\hy}$ are disjoint.
    By applying \cref{prp:edimpliesselfed} to $\mA : P\fragmentco{\pi_j^0}{\pi_j^{c_{\mathrm{l}}}} \onto T\fragmentco{\tau_i^0}{\tau_i^{c_{\mathrm{l}}}}$ and $\mX : P\fragmentco{\pi_j^0}{\pi_j^{c_{\mathrm{l}}}} \onto T\fragmentco{y}{\hy}$, we obtain
    \begin{align*}
        \selfed(P\fragmentco{\pi_j^0}{\pi_j^{c_{\mathrm{l}}}}) &\leq |\tau_i^0 - y| + \edal{\mA}(P\fragmentco{\pi_j^0}{\pi_j^{c_{\mathrm{l}}}}, T\fragmentco{\tau_i^0}{\tau_i^{c_{\mathrm{l}}}}) + \edal{\mX}(P\fragmentco{\pi_j^0}{\pi_j^{c_{\mathrm{l}}}}, T\fragmentco{y}{\hy}) + |\tau_i^{c_{\mathrm{l}}} - \hy| \\
                                                                &\leq 2|\tau_i^0 - y| + 2\edal{\mA}(P\fragmentco{\pi_j^0}{\pi_j^{c_{\mathrm{l}}}}, T\fragmentco{\tau_i^0}{\tau_i^{c_{\mathrm{l}}}}) + 2\edal{\mX}(P\fragmentco{\pi_j^0}{\pi_j^{c_{\mathrm{l}}}}, T\fragmentco{y}{\hy}) \\
                                                                &\leq 2w + 8k + 2w + 2k = 4w + 10k,
    \end{align*}
    where we have used
    \begin{align*}
        |\tau_i^{c_{\mathrm{l}}} - \hy| &\leq  |\tau_i^0 - y| + |(\tau_i^{c_{\mathrm{l}}} - \tau_i^0) - (\pi_j^{c_{\mathrm{l}}} - \pi_j^0)| + |(\hy - y) - (\pi_j^{c_{\mathrm{l}}} - \pi_j^0)| \\
                                        &\leq |\tau_i^0 - y| + \edal{\mA}(P\fragmentco{\pi_j^0}{\pi_j^{c_{\mathrm{l}}}}, T\fragmentco{\tau_i^0}{\tau_i^{c_{\mathrm{l}}}}) + \edal{\mX}(P\fragmentco{\pi_j^0}{\pi_j^{c_{\mathrm{l}}}}, T\fragmentco{y}{\hy}).
    \end{align*}
    An application of \cref{prp:lz_selfed} on $P\fragmentco{\pi_j^0}{\pi_j^{c_{\mathrm{l}}}}$ and
    $\mA^{j,0}_S$ yields
    \[
        \selfed(T\fragment{\tau_0^0}{\tau_0^{c_{\mathrm{l}}}}) \leq \selfed(P\fragmentco{\pi_j^0}{\pi_j^{c_{\mathrm{l}}}}) + 2w = 6w + 10k.
    \]
    As a consequence, we have
    \[
        \selfed(T\fragment{\tau_0^0}{\tau_0^{c_{\mathrm{l}}}}) \leq
        \selfed(T\fragmentco{\tau_0^0}{\tau_0^{c_{\mathrm{l}}}}) + 1 \leq 6w + 11k,\]
    and we obtain a contradiction.

    Since the alignments $\mA : P_j \onto T_{i}$ and $\mX : P_j \onto T\fragmentco{y}{y'}$ intersect and have costs
    at most $w$ and $k$, respectively, we conclude that $|\tau_{i+1}^0 - y'| \le w+k \le w+4k$.
    Consequently, by an argument symmetric to the above, there exists $(x_{\mathrm{r}},y_{\mathrm{r}}) \in \mA \cap \mX$ such that $(\pi_j^{c_{\mathrm{r}}}, \tau_i^{c_{\mathrm{r}}}) \leq (x_{\mathrm{r}},y_{\mathrm{r}})$.
    \end{claimproof}

    Now, for the sake of contradiction, suppose that there exists $c\in \fragmentco{0}{\bc(\bG_S)}\setminus C_S$ such that $\mX$ does not align $P[\pi^c_j]$ to $T[\tau^c_{i}]$.
    Let $(x_c,y_c) \in \mA \cap \mX$ be the largest $(x_c,y_c)$ such that $(x_c,y_c) \leq (\pi^c_j, \tau^c_{i})$, and let $(x_c',y_c') \in \mA \cap \mX$ be the smallest $(x_c',y_c')$ such that $(\pi^c_j, \tau^c_{i}) \leq (x_c',y_c')$.
    Note, we know that such $(x_c,y_c), (x_c',y_c')$ exist because $(x_{\mathrm{l}},y_{\mathrm{l}}), (x_{\mathrm{r}},y_{\mathrm{r}}) \in \mA \cap \mX$ are such that $(x_{\mathrm{l}},y_{\mathrm{l}}) \leq (\pi_j^{c_{\mathrm{l}}}, \tau_i^{c_{\mathrm{l}}})\leq (\pi^c_j, \tau^c_{i}) \leq (\pi_j^{c_{\mathrm{r}}}, \tau_i^{c_{\mathrm{r}}}) \leq (x_{\mathrm{r}},y_{\mathrm{r}})$.

    \begin{claim}
        $\mX$ and $\mA$ do not both align $P[x_c]$ to $T[y_c]$ at the same time.
    \end{claim}

    \begin{claimproof}
        For the sake of contradiction, assume $\mX$ and $\mA$ both align $P[x_c]$ to $T[y_c]$.
        Thus, $(x_c+1,y_c+1) \in \mX \cap \mA$.
        Since $(\pi^c_j, \tau^c_{i}) \notin \mX$, we must have $(x_c,y_c) \neq (\pi^c_j, \tau^c_{i})$,
        and $(x_c,y_c)$ appears strictly before $(\pi^c_j, \tau^c_{i}) \in \mA$ in $\mA$.
        Consequently, $(x_c+1,y_c+1) \leq (\pi^c_j, \tau^c_{i})$.
        However, this is a contradiction with $(x_c,y_c)$
        being the largest $(x_c,y_c) \in \mX \cap \mA$ such that $(x_c,y_c) \leq (\pi^c_j, \tau^c_{i})$.
    \end{claimproof}

    As a consequence,
    there is no $(\hx, \hy) \in \mA \cap \mX$ such that
    both $\mA : P\fragmentco{x_c}{x_c'} \onto T\fragmentco{y_c}{y_c'}$ and $\mX : P\fragmentco{x_c}{x_c'} \onto T\fragmentco{y_c}{y_c'}$
    align $P\position{\hx}$ to $T\position{\hy}$.
    Therefore, we can use \cref{prp:coverededit} obtaining $c \in C_S$.
    However, this is a contradiction with the definition of $c$.

    The proof for $j=m_0-1$ is almost identical to the proof for $j<m_0-1$: in the proof, we replace every occurrence of $\fragmentco{0}{\bc(\bG_S)}$  with $\fragment{0}{c_{\last}}$ and every occurrence of $\tau_{i+1}^0$ with $\tau_i^{c_{\last}}+1$.
\end{proof}

Finally, to conclude this (sub)section, we prove \cref{prp:close}.

\prpclose
\begin{proof}
    We assume that $C_S \subsetneq  \fragmentco{0}{\bc(\bG_S)}$; otherwise, there is nothing to prove.

    We prove~\eqref{it:close:in} by induction on $j\in \fragmentco{0}{m_0}$.
    Define $y_0,\ldots, y_{m_0}\in \fragment{t}{t'}$ so that $\mX(P_j)=T\fragmentco{y_j}{y_{j+1}}$
    holds for all $j\in \fragmentco{0}{m_0}$.
    Let us first prove that $|\tau_{i+j}^0 - y_{j}|\le w+4k$.
    If $j=0$, we observe that $|\pi_0^0-(y_0-t)|\le k$ because $\mX$ aligns $P\fragmentco{0}{\pi_0^0}$ with $T\fragmentco{t}{y_0}$ at a cost not exceeding $k$.
    Combining this inequality with the assumption $|\tau_{i}^0 - t - \pi_0^0| \leq w + 3k$,
    we conclude that $|\tau_{i}^0-y_0|\le w+4k$.
    If $j>0$, on the other hand, consider $c\in \fragmentco{0}{\bc(\bG_S)}$,
    for which the inductive assumption yields $(\pi^c_{j-1}, \tau^c_{j+i-1}) \in \mX$.
    Thus, $\mX : P_{j-1}\onto T\fragmentco{y_{j-1}}{y_j}$ and $\A_S^{j-1,i+j-1}:P_{j-1}\onto T\fragmentco{\tau_{i+j-1}^0}{\tau_{i+j}^0}$
    are intersecting alignments of costs at most $k$ and $w$, respectively.
    Consequently, $|\tau_{i+j}^0 - y_{j}|\le w+k \le w+4k$.
    In either case, we have proved that $|\tau_{i+j}^0 - y_{j}|\le w+4k$.
    This lets us apply \cref{prp:recperioded} for $\mX : P_j \onto T\fragmentco{y_{j}}{y_{j+1}}$,
    which implies that $\mX$ aligns $P[\pi_j^c]$ to $T[\tau_{i+j}^c]$ for all $c\in \fragmentco{0}{\bc(\bG_S)}\setminus C_S$ such that $j\in \fragmentco{0}{m_c}$,
    completing the inductive argument.

    We proceed to the proof of~\eqref{it:close:out}.
    First, consider $i'\in \fragmentco{0}{i}$ and suppose, for a proof by contradiction, that $\tau_{i'}^c\ge t$.
    \eqref{it:close:in} implies $(\pi_0^c,\tau_i^c)\in \mX$, so $\mX$ aligns $P\fragmentco{0}{\pi_0^c}$
    onto $T\fragmentco{t}{\tau_i^c}$ at cost at most $k$.
    By \cref{def:funcover}\eqref{it:funcover:3}, there exists an alignment $\mA$ of cost at most $w$
    that aligns $P\fragmentco{0}{\pi_0^c}$ onto $T\fragmentco{\hat{t}}{\tau_i^c}$
    for some $\hat{t}\in \fragment{\tau_{i-1}^{\bc(\bG_S)-1}}{\tau_i^0}$.
    Consequently, $|\hat{t}-t|\le w+k$.
    Since $t \le \tau_{i'}^c \le \tau_{i'}^{\bc(\bG_S)-1} \le \tau_{i-1}^{\bc(\bG_S)-1} \le \hat{t}$,
    we conclude $|T\fragmentco{\tau_{i'}^c}{\tau_{i'}^{\bc(\bG_S)-1}}|\le \hat{t}-t \le w+k$.
    Moreover, $|T\fragmentco{\tau_{0}^c}{\tau_{0}^{\bc(\bG_S)-1}}| \le |T\fragmentco{\tau_{i'}^c}{\tau_{i'}^{\bc(\bG_S)-1}}|+\ed(T\fragmentco{\tau_{0}^c}{\tau_{0}^{\bc(\bG_S)-1}},T\fragmentco{\tau_{i'}^c}{\tau_{i'}^{\bc(\bG_S)-1}})
    \le w+k+2w=3w+k$
    and $|T\fragment{\tau_{0}^c}{\tau_{0}^{\bc(\bG_S)-1}}|\le 3w+k+1 \le 3w+2k$.
    Observe that for any string $X$, by only applying insertions/deletions, we have $\selfed(X) \leq 2|X|$.
    Consequently, $\selfed(T\fragment{\tau_{0}^c}{\tau_{0}^{\bc(\bG_S)-1}})\le 6w+4k$,
    contradicting $c\notin C_S$.

    The argument for $i'\in \fragmentco{i+m_c}{n_c}$ is fairly similar.
    For a proof by contradiction, suppose that $\tau_{i'} < t'$.
    \eqref{it:close:in} implies $(\pi_{m_c-1}^c,\tau_{i+m_c-1}^c)\in \mX$, so $\mX$ aligns $P\fragmentco{\pi_{m_c-1}^c}{|P|}$
    onto $T\fragmentco{\tau_{i+m_c-1}^c}{t'}$ at cost at most $k$.
    By \cref{def:funcover}\eqref{it:funcover:5}, there exists an alignment $\mA$ of cost at most $w$
    that aligns $P\fragmentco{\pi_{m_c-1}^c}{|P|}$ onto $T\fragmentco{\tau_{i+m_c-1}^c}{\hat{t}}$
    for some $\hat{t}\in \fragment{\tau_{i+m_0-1}^{c_{\last}}}{\tau_{i+m_0-1}^{c_{\last}+1}}$.
    Consequently, $|\hat{t}-t'|\le w+k$. Now, we need to consider two cases.

    First, suppose that $c \le c_\last$ so that $i'\ge i+m_c = i+m_0$.
    Since $\hat{t} \le \tau_{i+m_0-1}^{c_{\last}+1} \le \tau_{i'}^0 \le \tau_{i'}^c <  t'$,
    we conclude $|T\fragmentco{\tau_{i'}^{0}}{\tau_{i'}^{c}}|\le t'-\hat{t} \le w+k$.
    Moreover, $|T\fragmentco{\tau_{0}^{0}}{\tau_{0}^{c}}| \le |T\fragmentco{\tau_{i'}^{0}}{\tau_{i'}^{c}}|+\ed(T\fragmentco{\tau_{0}^{0}}{\tau_{0}^{c}},T\fragmentco{\tau_{i'}^{0}}{\tau_{i'}^{c}})
    \le w+k+2w=3w+k$
    and $|T\fragment{\tau_{0}^{0}}{\tau_{0}^{c}}|\le 3w+k+1 \le 3w+2k$.
    Consequently, $\selfed(T\fragment{\tau_{0}^{0}}{\tau_{0}^{c}})\le 6w+4k$,
    contradicting $c\notin C_S$.

    Next, suppose that $c > c_\last$ so that $i' \ge i+m_c = i+m_0-1$.
    Since $\hat{t} < \tau_{i+m_0-1}^{c_{\last}+1} \le \tau_{i'}^{c_\last+1} \le \tau_{i'}^c <  t'$,
    we conclude $|T\fragmentco{\tau_{i'}^{c_\last+1}}{\tau_{i'}^{c}}|\le t'-\hat{t} \le w+k$.
    Moreover, $|T\fragmentco{\tau_{0}^{c_\last+1}}{\tau_{0}^{c}}| \le |T\fragmentco{\tau_{i'}^{c_\last+1}}{\tau_{i'}^{c}}|+\ed(T\fragmentco{\tau_{0}^{c_\last+1}}{\tau_{0}^{c}},T\fragmentco{\tau_{i'}^{c_\last+1}}{\tau_{i'}^{c}})
    \le w+k+2w=3w+k$
    and $|T\fragment{\tau_{0}^{c_\last+1}}{\tau_{0}^{c}}|\le 3w+k+1 \le 3w+2k$.
    Consequently, $\selfed(T\fragment{\tau_{0}^{c_\last+1}}{\tau_{0}^{c}})\le 6w+4k$,
    contradicting $c\notin C_S$.
\end{proof}

\subsection{Adding Candidate Positions to \texorpdfstring{$S$}{S}}

\begin{lemma}
    \label{lem:periodhalves}
    Let $\mY : P \onto T\fragmentco{t}{t'}$ be an alignment of cost at most $k$.
    If $|\tau_{i}^{0} - t - \pi_0^{0}|>w+2k$ holds for every $i\in \fragment{0}{n_0-m_0}$, then $\bc(\bG_{S \cup \{\mY\}}) \leq \bc(\bG_{S})/2$.
\end{lemma}

\begin{proof}
    It suffices to show that there is no $c \in \fragmentco{0}{\bc(\bG_{S})}$ such that  $\mY$ aligns $P\position{\pi_0^{c}}$ with
    $T\position{\tau_i^c}$ for some $i\in \fragmentco{0}{n_c}$.
    In this way, every black component of $\bG_S$ is merged with another (black or red) component in $\bG_{S \cup \{\mY\}}$,
    and thus the number of black components at least halves, that is, $\bc(\bG_{S'}) \leq \bc(\bG_{S})/2$.

    For a proof by contradiction, suppose that $(\pi_0^c,\tau_i^c)\in \mY$ for some $i\in \fragmentco{0}{n_c}$.
    We consider two cases.

    If $i\in \fragment{0}{n_0-m_0}$, we note that $\mY$ aligns $P\fragmentco{0}{\pi_0^c}$ to $T\fragmentco{t}{\tau_i^c}$. Since $\mY$ has cost at most $k$, we have $|\tau_i^c - t - \pi_0^c| \leq k$.
    At the same time, $\ed(P\fragmentco{\pi_0^0}{\pi_0^c},T\fragmentco{\tau_i^0}{\tau_i^c})\le w$ follows from \cref{def:funcover},
    which means that $|(\pi_0^c-\pi_0^0)-(\tau_i^c-\tau_i^0)|\le w$.
    Combining the two inequalities, we derive $|\tau_i^0 - t - \pi_0^0| \le w+k$.
    This contradicts the assumption that $|\tau_{i}^{0} - t - \pi_0^{0}|>w+2k$ holds for every $i\in \fragment{0}{n_0-m_0}$.

    If $i\in \fragmentoo{n_0-m_0}{n_0}$, then we recall that $S$ contains an alignment $\mX$ mapping $P$ onto a suffix of $T$, and $\mX$ is an alignment of cost at most $k$ that such that $(\pi_0^0,\tau_{n_0-m_0}^0),(\pi_0^c,\tau_{n_0-m_0}^c),(|P|,|T|)\in \mX$.
    As a result,
    \[|(\tau_{n_0-m_0}^c-\tau_{n_0-m_0}^0)-(\pi_0^c-\pi_0^0)| + |(|T|-\tau_{n_0-m_0}^c)-(|P|-\pi_0^c)|\le k.\]
    Since the alignment $\mY$ has cost at most $k$ such that $(0, t), (\pi_0^c,\tau_i^c),(|P|,t')\in \mY$,
    we conclude that
    \[|(\tau_{i}^c-t)-(\pi_0^c)| + |(t'-\tau_{i}^c)-(|P|-\pi_0^c)|\le k.\]
    Combining these two inequalities, we derive
    \[|(\tau_{n_0-m_0}^c-\tau_{n_0-m_0}^0)-(\pi_0^c-\pi_0^0)-(\tau_{i}^c-t)+(\pi_0^c)| + |(|T|-\tau_{n_0-m_0}^c)-(|P|-\pi_0^c)-(t'-\tau_{i}^c)+(|P|-\pi_0^c)|\le 2k,\]
    which simplifies to
    \[|(\tau_{n_0-m_0}^c-\tau_{n_0-m_0}^0)+\pi_0^0-(\tau_{i}^c-t)| + |(|T|-\tau_{n_0-m_0}^c)-(t'-\tau_{i}^c)|\le 2k.\]
    Since $t'\le |T|$ and $\tau_{i}^c \ge \tau_{n_0-m_0}^c$, we have
    \begin{align*}
        |\tau_{n_0-m_0}^0-t-\pi_0^0|
        &= |(\tau_{n_0-m_0}^c-\tau_{n_0-m_0}^0)+\pi_0^0-(\tau_{i}^c-t)+(\tau_{i}^c - \tau_{n_0-m_0}^c)|\\
        &\le |(\tau_{n_0-m_0}^c-\tau_{n_0-m_0}^0)+\pi_0^0-(\tau_{i}^c-t)|+|\tau_{i}^c - \tau_{n_0-m_0}^c| \\
        &= |(\tau_{n_0-m_0}^c-\tau_{n_0-m_0}^0)+\pi_0^0-(\tau_{i}^c-t)|+\tau_{i}^c - \tau_{n_0-m_0}^c \\
        &\le |(\tau_{n_0-m_0}^c-\tau_{n_0-m_0}^0)+\pi_0^0-(\tau_{i}^c-t)|+\tau_{i}^c - \tau_{n_0-m_0}^c+|T|-t' \\
        & = |(\tau_{n_0-m_0}^c-\tau_{n_0-m_0}^0)+\pi_0^0-(\tau_{i}^c-t)|+|\tau_{i}^c - \tau_{n_0-m_0}^c+|T|-t'|\\
        &= |(\tau_{n_0-m_0}^c-\tau_{n_0-m_0}^0)+\pi_0^0-(\tau_{i}^c-t)| + |(|T|-\tau_{n_0-m_0}^c)-(t'-\tau_{i}^c)| \\
        & \le 2k.
    \end{align*}
    This contradicts the assumption that $|\tau_{n_0-m_0}^{0} - t - \pi_0^{0}|>w+2k$.
\end{proof}

\subsection{Recovering all \texorpdfstring{$k$}{k}-edit Occurrences of
\texorpdfstring{$P$}{P} in \texorpdfstring{$T$}{T}}
\label{sec:recocc}

In this (sub)section we drop the assumption that $S$ encloses $T$ and that $\bc(\bG_{S}) > 0$.
Every time we assume that $S$ encloses $T$ or that $\bc(\bG_{S}) > 0$, we  explicitly mention it.

\begin{definition} \label{def:scomplete}
Let $S$ be a set of $k$-edit alignments of $P$ onto fragments of $T$
and let $T\fragmentco{t}{t'}$ be a $k$-error occurrence of $P$ in $T$.
We say that $S$ \emph{captures} $T\fragmentco{t}{t'}$ if $S$ encloses $T$
and exactly one of the two following holds:
\begin{itemize}
    \item $\bc(\bG_{S}) = 0$; or
    \item $\bc(\bG_{S}) > 0$ and $|\tau_i^{0} - t - \pi_0^{0}| \leq w + 3k$
    holds for some $i\in \fragmentco{0}{n_0}$.\qedhere
\end{itemize}
\end{definition}

\begin{theorem} \label{prp:subhash}
Let $S$ be a set of $k$-edit alignments of $P$ onto fragments of $T$
such that $S$ encloses $T$ and $\bc(\bG_{S}) > 0$.
Construct $P^\#$ and $T^\#$ by replacing, for every $c \notin C_S$, every character in the $c$-th black component with a unique character $\#_c$. Then, the two following hold.
\begin{enumerate}[(i)]
    \item For every $a\in \fragmentco{0}{m}$ and $b \in \fragmentco{0}{n}$, we have
        that $P^\#\position{a} = T^\#\position{b}$ implies  $P\position{a} = T\position{b}$.\\
        (No new equalities between characters are created.)
    \label{it:ed_subhash:i}
        \item If $S$ captures the k-error occurrence $T\fragmentco{t}{t'}$, then $\ed(P, T\fragmentco{t}{t'}) \leq \ed(P^\#, T^\#\fragmentco{t}{t'})$.
        Moreover, for all optimal $\mX \mid P \onto T\fragmentco{t}{t'}$ we have $\sE_{P, T}(\mX) = \sE_{P^\#, T^\#}(\mX)$.\\
        (For captured $k$-error occurrences, the edit information and the edit distance are preserved.)
        \label{it:ed_subhash:ii}
\end{enumerate}
\end{theorem}

\begin{proof}

First, for \eqref{it:ed_subhash:i} observe that we substitute with the same sentinel $\#_c$
characters belonging to the $c$-th black connected component, which we know to be uniform, that is
containing the same character.

We proceed with the proof of \eqref{it:ed_subhash:ii}.
Suppose $\mX : P \onto T\fragmentco{t}{t'}$ is an optimal alignment
of cost $k' \leq k$. We begin by showing that $\ed(P, T\fragmentco{t}{t'}) \leq \ed(P^\#, T^\#\fragmentco{t}{t'})$.

Since $S$ encloses $T$ and $\bc(\bG_{S}) > 0$, there exists $i\in \fragmentco{0}{n_0}$ such that $|\tau_i^0-t-\pi_0^0|\le w+3k$,
and \cref{prp:close} applies to $\mX$.
By \cref{prp:close}\eqref{it:close:out}, for every $c \in \fragmentco{0}{\bc(\bG_S)}$ we have that
$\#_c$ appears in $P^{\#}$ at the positions $\Pi_c = \{\pi_j^c \mid j \in \fragmentco{0}{m_c}\}$
and in $T^{\#}\fragmentco{t}{t'}$ at the positions $\Tau_c = \{\tau_j^c \mid j \in \fragmentco{i}{i + m_c}\}$.

Observe that the difference in cost between $\mX : P \onto T\fragmentco{t}{t'}$ and $\mX : P^{\#} \onto T^{\#}\fragmentco{t}{t'}$ is determined by the characters that $\mX$ matches in $P$ and $T$,
but due to the substitutions they do not match anymore in $P^\#$ and $T^\#$. That is,
$\ed(P^\#, T^\#\fragmentco{t}{t'}) - \ed(P, T\fragmentco{t}{t'}) = |E|$, where
\[
    E \coloneqq \big\{(x, y) \mid \text{$\mX$ matches $P\position{x}$ with $T\position{y}$ but not $P^\#\position{x}$ with $T^\#\position{y}$} \ \big\}
    \text{ is such that }
    E \subseteq \big(\textstyle{\bigcup_c \Pi_c \times \bigcup\nolimits_c \Tau_c} \big) \cap \mX.
\]
By \cref{prp:close}\eqref{it:close:in}, we have that $\mX$ aligns $P\position{\pi^c_j}$ to $T\position{\tau^c_{i+j}}$ for every $j \in \fragmentco{0}{m_c}$, meaning that
\[
    E \subseteq \{(\pi^c_j, \tau^c_{i+j}) \mid j \in \fragmentco{0}{m_c}, c \in \fragmentco{0}{\bc(\bG_S)}\}.
\]
Moreover, as $\pi^c_j$ and $\tau^c_{i+j}$, are in the same black connected component, we have $P^{\#}\position{\pi_j^c} = T^{\#}\position{\tau_j^c}$.
Consequently, $E = \emptyset$ and $\ed(P^\#, T^\#\fragmentco{t}{t'}) = \ed(P, T\fragmentco{t}{t'})$.

For the second part of \eqref{it:ed_subhash:ii}, it suffices
to notice that all characters that are substituted in $P,T$
(or equivalently all the hashes in $P^{\#},T^{\#}$)
are always matched by such $\mX$.
Since the characters that are matched by $\mX$
are stored explicitly neither in $\sE_{P, T}(\mX)$ nor in $\sE_{P^\#, T^\#}(\mX)$,
we conclude $\sE_{P, T}(\mX) = \sE_{P^\#, T^\#}(\mX)$.
\end{proof}

As \cref{prp:subhash}\eqref{it:ed_subhash:i} implies that the edit distance between $P$ and $T\fragmentco{t}{t'}$
cannot be larger than the edit distance between $P^\#$ and $T^\#\fragmentco{t}{t'}$,
we obtain the following corollary.

\begin{corollary} \label{cor:subhash}
Let $S$ be a set of $k$-edit alignments of $P$ onto fragments of $T$
such that $S$ encloses $T$ and $\bc(\bG_{S}) > 0$.
Construct $P^\#$ and $T^\#$ as in \cref{prp:subhash}.
If $S$ captures all $k$-error occurrences, then:
\begin{enumerate}[(i)]
    \item $\ed(P, T\fragmentco{t}{t'}) = \ed(P^\#, T^\#\fragmentco{t}{t'})$ and $\sE_{P, T}(\mX) = \sE_{P^\#, T^\#}(\mX)$ for all optimal $\mX \mid P \onto T\fragmentco{t}{t'}$ of cost at most $k$.
    \item $\ed(P, T\fragmentco{t}{t'}) \leq \ed(P^\#, T^\#\fragmentco{t}{t'})$ for all $t,t' \in \fragment{0}{n}$. \lipicsEnd
\end{enumerate}
\end{corollary}

\subsection{The Communication Complexity of Edit Distance}

\ccompl

\begin{proof}
    We first prove the second part of the theorem, that is, we describe a protocol allowing Alice to encode using $\Oh(n/m\cdot k\log m \log(m|\Sigma|))$ bits the edit information for all alignments in
    \[S' \coloneqq \{\mX : P \onto T\fragmentco{t}{t'} \mid t,t' \in \fragment{0}{n}, \edal{\mX}(P, T\fragmentco{t}{t'}) = \ed(P, T\fragmentco{t}{t'}) \leq k\}.\]

    We may assume $k \leq m/4$; otherwise, it suffices to send $P$ and $T$ to Bob.
    We want to argue that we may further restrict ourselves to the case where $S'$ encloses $T$.

    \begin{claim}
        If such a protocol
        using $\Oh(k\log m\log(m|\Sigma|))$ bits exists in the case where $S'$ encloses $T$,
        then such a protocol also exists for the general case.
    \end{claim}

    \begin{claimproof}
        Divide $T$ into $\Oh(n/m)$ contiguous blocks of length $m-2k\ge m/2$ (with the last block potentially being shorter).
        For the $i$-th block starting at position $b_i \in \fragmentco{0}{|T|}$, define $B_i = \fragment{b_i}{b_i + 2m-2k} \cap \fragment{0}{|T|}$ and $S_i = \{\mX : P \onto T\fragmentco{t}{t'} \mid t,t' \in B_i, \edal{\mX}(P, T\fragmentco{t}{t'}) = \ed(P, T\fragmentco{t}{t'}) \leq k\}$.
        Since $S' = \bigcup_{i} S_i$, it suffices to show that there exists such a protocol using $\Oh(k\log m\log(m|\Sigma|))$ bits for the case where instead of $T$ we consider an arbitrary $B_i$.

        Therefore, let $i$ be arbitrary.
        If $S_i = \emptyset$, then Alice does not need to send anything.
        On the other hand, if $|S_i| = 1$, then Alice directly sends to Bob the edit information of the unique alignment contained in $S_i$ using $\Oh(k\log(m|\Sigma|))$ bits.
        Lastly, if $|S_i| \geq 2$, Alice uses the protocol from the assumption of the claim on $S_i$ and $T\fragmentco{\ell_i}{r_i}$, where $\ell_i = \min \{t \mid \mX : P \onto T\fragmentco{t}{t'} \in S_i\}$ and $r_i = \max \{t' \mid \mX : P \onto T\fragmentco{t}{t'} \in S_i\}$.
        This is indeed possible, as $S_i$ encloses $T\fragmentco{\ell_i}{r_i}$ because $|B_i|-1 \leq 2m - 2k$ and $|S_i| \geq 2$.
    \end{claimproof}

    Henceforth, suppose $S'$ encloses $T$.
    We construct a subset $S \subseteq S'$ that
    captures all $k$-error occurrences as follows.
    First, we add to $S$ the alignments $\mXpref, \mXsuf \in S'$
    such that $(0, 0) \in \mXpref$ and $(|P|, |T|) \in \mXsuf$.
    This ensures that $S$ encloses $T$.
    Next, we iteratively add alignments to $S$ according to the following rules:
    \begin{itemize}
        \item if $S$ captures all $k$-error occurrences,
        then we stop;
        \item otherwise, $\bc(\bG_S) > 0$, and we add to $S$ an arbitrary alignment $\mY : P \onto T\fragmentco{t}{t'}$ of cost at most $k$ such that $|\tau_{i}^{0} - t - \pi_0^{0}|>w+2k$ holds for every $i\in \fragment{0}{n_0-m_0}$
        ($w = \Oh(k|S|)$ denotes the total weight of the weight function from \cref{prp:encode_simple_funcover}).
    \end{itemize}
    Note that we never apply the second rule more than $\Oh(\log n)$ times:
    from \cref{lem:periodhalves} follows that every time we add $\mY$ to $S$,
    $\bc(\bG_{S \cup \{\mY\}}) \leq \bc(\bG_{S})/2$, and if $\bc(\bG_{S}) = 0$,
    then the first rule applies.
    Consequently, $|S| = \Oh(\log n)$.

    Now, if $\bc(\bG_{S}) = 0$, then the information that Alice sends to Bob is
    $\{\sE_{P,T}(\mX) : \mX : P \onto T \in \mX\}$.
    As already argued before, this allows Bob to fully reconstruct $P$ and $T$.

    Otherwise, if $\bc(\bG_{S}) > 0$, then Alice additionally sends the encoding of
    $\{(c, T\position{\tau_0^c}):c\in C_S\}$ to Bob, where $C_S$ is the
    black cover from \cref{prp:encode_simple_funcover}.
    This allows Bob to construct the strings $P^\#$ and $T^\#$ from \cref{prp:subhash}.
    \cref{cor:subhash} guarantees that by computing the answer over $P^\#$ and $T^\#$
    instead of $P$ and $T$, Bob still obtains the information required by the protocol.

    Note, $\{\sE_{P,T}(\mX) : \mX : P \onto T\fragmentco{t}{t'} \in S\}$ and $\{(c, T\position{\tau_0^c}):c\in C_S\}$
    can be stored using $\mathcal{O}(w + k|S|) = \mathcal{O}(k\log n)$ space.
    Since a space unit consists of either a position of $P$/$T$,
    or a character from $\Sigma$,
    we conclude that the second part of the theorem holds.

    Lastly, we prove the first part of the theorem.
    If there is no requirement to send the edit information for every $\mX \in S$,
    we encode the characters in $P$ to an alphabet $\Sigma'$ with a size of at most $m+1$.
    This new alphabet includes an additional special character representing
    all characters in $T$ that are not present in $P$.
    By mapping all characters of $P$ to $T$ to the corresponding characters of $\Sigma'$,
    $S'$ and all alignments contained in $S'$ remain unchanged.
    Thereby, we reduce the number of bits needed by the protocol to $\Oh(n/m \cdot k\log^2 m)$.
\end{proof}

\cclb
\begin{proof}
    Let $p = \lfloor{n/(2m-2)}\rfloor$ and $T=S_0\cdot S_0\cdot S_1\cdot S_1 \cdots S_{p-1}\cdot S_{p-1}\cdot \zero^{n-p(2m-2)}$, where $S_0,\ldots,S_{p-1}\in \{\mathtt{0},\mathtt{1}\}^{m-1}$ are strings that contain exactly $k$ copies of $\one$ and $m-1-k$ copies of $\zero$.
    We shall prove that, for every $q\in \fragmentco{0}{p}$ and $i\in \fragmentco{0}{m-1}$, we have $S_q\position{i}=\mathtt{0}$ if and only if $q(2m-2)+i\in \OccE_{k}(P,T)$. Consequently, $T$ can be recovered from $\OccE_{k}(P,T)$, and the theorem follows because the number of possibilities for $T$ is $\binom{m-1}{k}^{p} = 2^{\Omega(n/m\cdot k\log(m/k))}$.

        If $S_q\position{i}=\mathtt{0}$, then $T\fragmentco{q(2m-2)+i}{q(2m-2)+i+m}=S_q\fragmentco{i}{m-1}\cdot S_q\fragment{0}{i}$ contains exactly $k$ copies of $\one$ and $m-k$ copies of $\zero$, so $q(2m-2)+i\in \OccE_{k}(P,T)$.

        Similarly, if $S_q\position{i}=\mathtt{1}$, then $T\fragmentco{q(2m-2)+i}{q(2m-2)+i+m}$ contains exactly $k+1$ copies of $\one$ and $m-k-1$ copies of $\zero$.
        For $s\ge 0$, every alignment $P\onto T\fragmentco{q(2m-2)+i}{q(2m-2)+i+m+s}$ makes at least $k+1$ insertions or substitutions (to account for $\mathtt{1}$s).
        Every alignment $P\onto T\fragmentco{q(2m-2)+i}{q(2m-2)+i+m-s}$, on the other hand, makes at least $s$ deletions (to account for the length difference) and at least $k+1-s$  insertions or substitutions (to account for $\mathtt{1}$s).
        The overall number of edits is at least $k+1$ in either case, so $q(2m-2)+i\notin \OccE_{k}(P,T)$.
\end{proof}

\section{Quantum Algorithms on Strings}

We assume the input string $S \in \Sigma^n$
can be accessed in a quantum query model \cite{ambainis2004quantum,DBLP:journals/tcs/BuhrmanW02}:
there is an input oracle $O_S$ that performs the unitary mapping
$O_s \colon \lvert i, b \rangle  \mapsto \lvert i, b \oplus S[i] \rangle$ for any index $i \in \fragmentco{0}{n}$ and any $b \in \Sigma$.

The \emph{query complexity} of a quantum algorithm (with $2/3$ success probability)
is the number of queries it makes to the input oracles.
The \emph{time complexity} of the quantum algorithm additionally
counts the elementary gates \cite{PhysRevA.52.3457}
for implementing the unitary operators that are independent of the input.
Similar to prior works \cite{GS22,AJ22,ambainis2004quantum}, we assume quantum random access quantum memory.

We say an algorithm succeeds \emph{with high probability (w.h.p)},
if the success probability can be made at least $1 - 1/n^c$
for any desired constant $c > 1$.
A bounded-error algorithm can be boosted to succeed w.h.p.\ by $O(\log n)$ repetitions.
In this paper, we do not optimize sub-polynomial factors
of the quantum query complexity (and time complexity) of our algorithms.

For our quantum algorithms, we rely on the following primitive quantum operation.

\begin{theoremq}[\GS (Amplitude amplification) \cite{DBLP:conf/stoc/Grover96,brassard2002quantum}]
    \label{prp:grover}
    Let $f\colon \fragmentco{0}{n} \to \{0,1\}$ denote a function, where $f(i)$ for each $i\in
    \fragmentco{0}{n}$ can be evaluated in time $T$.
    There is a quantum algorithm that, with high probability and in time $\tilde O(\sqrt{n}\cdot T)$,
    finds an $x\in f^{-1}(1)$ or reports that $f^{-1}(1)$ is empty.
\end{theoremq}

In designing lower bounds for our quantum algorithm, we rely on the following framework.

\begin{theoremq}[Adversary method {\cite{Ambainis02}}]\label{thm:adv}
    Let $f(x_0 , ..., x_{n-1})$ be a function of $n$ binary variables and $X,Y\subseteq \{0,1\}^n$
    be two sets of inputs such that $f(x) \ne f(y)$ if $x \in X$ and $y \in Y$. Let $R \subseteq X\times Y$  be such that:
    \begin{enumerate}
        \item For every $x \in X$, there exist at least $m$ different $y \in Y$ such that $(x, y) \in R$.
        \item For every $y \in Y$, there exist at least $m'$ different $x \in X$ such that $(x, y) \in R$.
        \item For every $x \in X$ and $i \in \fragmentco{0}{n}$, there is at most one $y \in  Y$ such
        that $(x, y) \in R$ and $x_i \ne y_i$.
        \item For every $y \in X$ and $i \in \fragmentco{0}{n}$, there is at most one $x \in  X$ such
        that $(x, y) \in R$ and $x_i \ne y_i$.
    \end{enumerate}
    Then, any quantum algorithm computing $f$ uses $\Omega\left(\sqrt{mm'}\right)$ queries
\end{theoremq}

Throughout this paper,
we build upon several existing quantum algorithms for string processing.
We summarize below the key techniques and results that are relevant to our work.

\subsection{Essential Quantum Algorithms on Strings}

\begin{theoremq}[Quantum Exact Pattern Matching~\cite{HV03}]\label{thm:quantum_matching}
    There is an $\Ohtilde(\sqrt{n})$-time quantum algorithm that, given a pattern $P$ of length $m$ and a text $T$ of length~$n\ge m$,
    finds an exact occurrence of $P$ in $T$ (or reports that \(P\) does not occur in \(T\)).
\end{theoremq}

Kociumaka, Radoszewski, Rytter, and Waleń~\cite{ipm} established that the computation of $\per(S)$ can be reduced to $\Oh(\log |S|)$ instances of exact pattern matching and longest common prefix operations involving substrings of $S$.
Consequently, the following corollary holds.

\begin{corollaryq}[Finding Period]\label{cor:quantum_per}
There is an $\Ohtilde(\sqrt{n})$-time quantum algorithm that, given a string $S$ of length~$n$, computes $\per(S)$.
\end{corollaryq}

We modify the algorithm of~\cite{HV03} such that it outputs the whole set $\Occ(P, T)$.

\begin{lemma}\label{thm:quantum_matching_all}
    There is an $\Ohtilde(\sqrt{\mathsf{occ} \cdot n})$-time quantum algorithm that, given a pattern $P$ of length $m$ and a text $T$ of length $n\ge m$, outputs $\Occ(P, T)$, where $\mathsf{occ} = |\Occ(P, T)|$.
\end{lemma}
\begin{proof}
    Let $t_1 < t_2 < \cdots < t_{\mathsf{occ}}$ denote the positions that are contained in $\Occ(P, T)$.

    Suppose that we have just determined $t_i$ for some $i\in \fragmentco{0}{\mathsf{occ}}$ (if $i = 0$, then we set $t_0 = -1$).
    To find $t_{i+1}$ we use \cref{thm:quantum_matching} combined with exponential search.
    More specifically, at the $j$-th jump of the exponential search, we use \cref{thm:quantum_matching} to check  whether $\Occ(P, T\fragmentco{t_{i}+1}{\min(\sigma_i + 2^{j - 1},|T|)}) \neq \emptyset$.
    Once we have found the first $j$ for which this set is non-empty, we use \cref{thm:quantum_matching} combined with binary search to find the exact position of $t_{i+1}$.
    By doing so, we use the algorithm of \cref{thm:quantum_matching} as subroutine at most $\Oh(\log (t_{i+1} - t_{i})) = \Oh(\log n)$ times, each time requiring at most $\Ohtilde(\sqrt{t_{i+1} - t_i})$ time.

    As a consequence, finding all $t_1, \ldots, t_{\mathsf{occ}}$  requires $\Ohtilde(\sum_{i=0}^{\mathsf{occ}-1} \sqrt{t_{i+1} - t_i})$ time.
    By the Cauchy--Schwarz inequality, $\sum_{i=0}^{\mathsf{occ}-1} \left(1\cdot \sqrt{t_{i+1} - t_i}\right) \leq \sqrt{\mathsf{occ}} \cdot \sqrt{t_{\mathsf{occ}}-t_0} \le \sqrt{\mathsf{occ} \cdot n}$.
    Therefore, the total time complexity becomes $\Ohtilde(\sqrt{\mathsf{occ} \cdot n})$.
\end{proof}

Via \cref{thm:quantum_matching} we may also check for rotations of primitive strings.

\begin{lemma}\label{claim:hd_compstructure_rot}
    Given a string \( X \) and a primitive string \( Q \) of the same length $n$, we can
    determine in \(\Ohtilde(\sqrt{n})\) quantum time whether there exists an index \( a \in
    \fragmentco{0}{n} \) such that \(\rot^a(Q) = X\). If such an index exists, we
    can find the unique \( a \) that satisfies this condition.
\end{lemma}
\begin{proof}
    It suffices to use \cref{thm:quantum_matching} to check whether \( Q \) occurs in \( XX\fragmentco{0}{n-1} \).
    If an occurrence is found, return its starting position.
    Since \( Q \) is primitive, the occurrence must be unique.
\end{proof}

Another useful tool is the following quantum algorithm to compute the edit distance of two strings.

\begin{theoremq}[Quantum Bounded Edit Distance {\cite[Theorem 1.1]{GJKT24}}]\label{prp:quantumed}
    There is a quantum algorithm that, given quantum oracle access to strings $X,Y$ of length at most $n$, computes their edit distance $k \coloneqq \ed(X,Y)$, along with a sequence of $k$ edits transforming $X$ into $Y$.
    The algorithm has a query complexity of $\Ohtilde(\sqrt{n + kn})$ and a time complexity of $\Ohtilde(\sqrt{n + kn} + k^2)$.
\end{theoremq}

The algorithm of \cref{prp:quantumed} returns lists edits from left to right, with each edit specifying the positions and the characters involved (one for insertions and deletions, two for substitutions).
In the following corollary, we show that this representation can be easily transformed into the edit information $\sE_{X, Y}(\mA)$ of some optimal alignment $\mA : X \onto Y$.

\begin{corollary}\label{prp:quantumed_w_info}
    There is a quantum algorithm that, given quantum oracle access to strings $X,Y$ of length at most $n$,
    computes their edit distance $k \coloneqq \ed(X,Y)$ and the edit information $\sE_{X, Y}(\mA)$ of some optimal alignment $\mA : X \onto Y$.
    The algorithm has the same query and time complexity as in \cref{prp:quantumed}.
\end{corollary}
\begin{proof}
    We use \cref{prp:quantumed} to retrieve the sequence of $k$ edits, and then we construct $\sE \coloneqq \sE_{X, Y}(\mA)$ as follows.
    We iterate through the returned list of edits in the reverse (right-to-left) order and add elements to $\sE$ according to the following rules, maintaining a pair $(x',y')$ initialized as $(|X|,|Y|)$.
    \begin{itemize}
        \item If the current edit substitutes $X\position{x}$ with $Y\position{y}$, we set $(x',y')\coloneqq (x,y)$ and add $(x', X\position{x}, y', Y\position{y})$ to $\sE$.
        \item If the current edit inserts $Y\position{y}$, we set $(x',y')\coloneqq (x'+y-y'+1,y)$ and add $(x', \epsilon, y', Y\position{y})$ to $\sE$.
        \item Lastly, if the current edit deletes $X\position{x}$, we set $(x',y')\coloneqq (x,y'+x-x'+1)$ and add $(x', X\position{x}, y', \epsilon)$ to $\sE$.
    \end{itemize}
    Observe that we can reconstruct the underlying alignment $\mA$ by a similar iterative scheme used to construct $\sE$:
    just before setting $(x',y')$ to $(x'',y'')$, we add to $\mA$ all pairs of the form $(x'-i,y'-i)$ with $i\in \fragmentco{0}{\min(x'-x'',y'-y'')}$.
    Moreover, after processing the entire edit sequence, we add to $\mA$ all the pairs of the form $(x'-i,y'-i)$ with $i\in \fragment{0}{x'}$; we then have $x'=y'$.
    It is easy to verify at this point that the following construction of $\mA$ leads to a correctly defined alignment of cost $k$ such that $\sE = \sE_{X, Y}(\mA)$.

    Lastly, note that the construction of $\sE$ requires at most $\Oh(k)$ additional time and no queries.
\end{proof}

From the work that established \cref{prp:quantumed},
we also obtain an algorithm for computing the LZ factorization of a string.

\begin{theoremq}[Quantum LZ77 factorization {\cite[Theorem 1.2]{GJKT24}}]\label{prp:quantumlz}
    There is a quantum algorithm that, given quantum oracle access to a string $X$ of length $n$ and an integer threshold $z\ge 1$, decides whether $|\LZ(X)|\le z$ and, if so, computes the LZ factorization $\LZ(X)$.
    The algorithm uses $\Ohtilde(\sqrt{zn})$ query and time complexity.
\end{theoremq}

\begin{corollary}\label{cor:testselfed}
    There is a quantum algorithm that, given quantum oracle access to a string $X$ of length $n$ and an integer threshold $k\ge 1$, decides whether $\selfed(X)\le k$ using $\Ohtilde(\sqrt{kn})$  queries and $\Ohtilde(\sqrt{kn}+k^2)$ time.
\end{corollary}
\begin{proof}
    First, we apply \cref{prp:quantumlz} with $z=2k$. If $|\LZ(X)|>2k$, we return \no.
    Otherwise, we check $\selfed(X)\le k$ as follows.

    We construct a classical data structure for LCE (Longest Common Extension) queries in $X$ following the method outlined in \cite{tomohiro}.
    This data structure can be built in $\Ohtilde(k)$ time and supports LCE queries in $\Oh(\log z)$ time, where $z = |\LZ(X)|$.

    We then employ an algorithm, described in Lemma 4.5 of \cite{CKW23},
    which is a simple modification of the classic Landau-Vishkin algorithm \cite{LV88}.
    This algorithm requires only a blackbox for computing LCE queries.
    By using this algorithm, we can check whether $\selfed(X)\le k$ using $\Oh(k^2)$ LCE queries to the data structure,
    thus requiring $\Ohtilde(k^2)$ time.

    Regarding the correctness anaylsis, if $|\LZ(X)|>2k$, then according to \cref{prp:lz_selfed}, $\selfed(X) > k$.
    Otherwise, the correctness follows from the correctness of the algorithm presented in \cite{CKW23}.
\end{proof}

\subsection{Quantum Algorithm for Edit Distance Minimizing Suffix}

In our Quantum \PMwE algorithm, there is a recurring need for a procedure that identifies a suffix of a text that minimizes the edit distance to a specific pattern.
To this end, we define
\[
    \edpx{X}{Y} \coloneqq \min_{y \in \fragment{0}{|Y|}} \ed(X, Y\fragmentco{y}{|Y|}).
\]
Similarly, we define $\edsx{X}{Y} \coloneqq \min_{y \in \fragment{0}{|Y|}} \ed(X,Y\fragmentco{0}{y})$ for prefixes.
In this (sub)section, we discuss a quantum subroutine that calculates $\edpx{X}{Y}$.
Note, by reversing the order of its input strings, the same subroutine computes $\edsx{X}{Y}$.

\begin{lemmaq}\label{prp:candidatepos}
    Let $T$ denote a text of length $n$, let $P$ denote a pattern of length $m$, and let $k > 0$ denote a positive threshold.
    Then, we can check whether $\edpx{P}{T} \leq k$ (and, if this is the case, calculate the value of $\edpx{P}{T}$ and the suffix of $T$ minimizing the edit distance) in $\Ohtilde(\sqrt{km} + k^2)$ quantum time using $\Ohtilde(\sqrt{km}+k)$ quantum queries.
\end{lemmaq}
\begin{proof}
    Our algorithm relies on the following combinatorial claim.
    \begin{claim}\label{clm:candidatepos}
        Let $\$$ denote a character that occurs in neither $P$ nor $T$.
        If $n \leq m+k$ and $\edpx{P}{T}\le k$, then $\ed(\$^{2k}P,T)=\edpx{P}{T}+2k$.
        Moreover, if an optimum alignment $\mA : \$^{2k}P\onto T$ substitutes exactly $s$ among $\$$ characters, then $\edpx{P}{T}=\ed(P,T\fragmentco{s}{n})$.
    \end{claim}
    \begin{claimproof}
            Consider a position $i\in \fragment{0}{n}$ such that $\ed(P,T\fragmentco{i}{n})=\edpx{P}{T}\le k$.
            Observe that $|m-(n-i)|\le k$, so the assumption $n\le m+k$ implies $i\le k+n-m \le 2k$.
            Consequently, $\ed(\$^{2k},T\fragmentco{0}{i})\le \max(2k,i) = 2k$ and $\ed(\$^{2k}P,T) \le \ed(\$^{2k},T\fragmentco{0}{i})+\ed(P,T\fragmentco{i}{n})\le 2k+\edpx{P}{T}$.

            Next, consider an optimal alignment $\mA : \$^{2k}{P}\onto T$ and a position $i\in \fragment{0}{n}$ such that $(2k,i)\in \mA$.
            Note that $\ed(\$^{2k}P,T)=\edal{\mA}(\$^{2k},T\fragmentco{0}{i})+\edal{\mA}(P,T\fragmentco{i}{n}) \ge 2k+\edpx{P}{T}$; this is because $\$$ does not occur in $T$ (so its every occurrence needs to be involved in an edit) and due to the definition of $\edpx{P}{T}$.

            Since $\ed(T,\$^{2k}P)\le 2k+\edpx{P}{T}$, we have $\edal{\mA}(T\fragmentco{0}{i},\$^{2k})=2k$ and $\edal{\mA}(T\fragmentco{i}{n},P)=\edpx{P}{T}$.
            In particular, $\mA$ substitutes exactly $s=i$ characters $\$$, and $\edpx{P}{T}=(T\fragmentco{s}{n},P)$ holds.
    \end{claimproof}

    First, observe that we can assume that $n \leq m+k$.
    Otherwise, we can truncate $T$ to $T\fragmentco{n - m - k}{n}$ since the edit distance between all suffixes of $T$ that we left out and the pattern $P$ exceeds $k$.

    Following \cref{clm:candidatepos}, we use \cref{prp:quantumed} to check whether $\ed(T, \$^{2k}P)\le 3k$.
    If that is not the case, we report that $\edpx{P}{T}>k$.
    Otherwise, we retrieve $\ed(T, \$^{2k}P)$ along with an optimal sequence of edits transforming $T$ into $\$^{2k}P$.
    We report $\ed(T, \$^{2k}P)-2k$ as the distance $\edpx{P}{T}$ and $T\fragmentco{s}{n}$, where $s$ is the number of characters $\$$ that $\mA$ substitutes, as the minimizing suffix.

    The correctness of this approach follows directly from \cref{clm:candidatepos}.
    The complexity is dominated by the application of \cref{prp:quantumed}: $\Ohtilde(\sqrt{km}+k^2)$ time and $\Ohtilde(\sqrt{km}+k)$ queries.
\end{proof}

Given a string $Q$ and a threshold $k$, \cref{prp:candidatepos} also allows us to check if  $\eds{P}{Q} \leq k$ holds,
and to return $\eds{P}{Q}$ using $\Ohtilde(\sqrt{kn} + k^2)$ time and $\Ohtilde(\sqrt{kn})$ queries.
In fact, it is easy to verify that $\eds{P}{Q} = \edpx{P}{Q^{\infty}\fragmentco{q - |P| - k}{q}}$,
where $q = |Q| \cdot \ceil{(|P| + k)/|Q|}$.
Throughout the rest of the paper, we assume that, whenever we use \cref{prp:candidatepos}
to compute expressions of the form $\eds{P}{Q}$, we return also the negative index $q-x$, where $x \in \fragment{q - |P| - k}{q}$ is the positions of the suffix minimizing $\edpx{P}{Q^{\infty}\fragmentco{q - |P| - k}{q}}$.
Similar consideration holds for $\ed(P,Q^*)$, for which together with $\ed(P,Q^*)$,
we return the index $y$ that minimizes $\ed(P,Q^\infty\fragmentco{0}{y})$.

Next, we show that via \cref{prp:candidatepos}
we can also compute $\edl{S}{Q}$, provided $Q$ is primitive.

\begin{lemma}\label{lem:find_min_approx_per}
    Let $S$ be a string, let $Q$ be a primitive string,
    and let $k > 0$ be a positive threshold
    such that $|S| \geq (2k+1)|Q|$ or $|Q|=1$.
    Then, there is a quantum algorithm that either verifies that $\edl{S}{Q} > k$,
    or outputs the value $\edl{S}{Q}$ in the contrary case.
    In the latter case, it also outputs $\ell,r$
    such that $\edl{S}{Q} = \ed(S, Q^\infty\fragmentco{\ell}{r})$.
    The algorithm requires $\Oh(\sqrt{k|S|})$ time and $\Oh(\sqrt{k|S|}+|S|^2)$ queries.
\end{lemma}

\begin{proof}
        We first sort out the case when $|Q|=1$.
        In such case we have $\edl{S}{Q} = \ed(S,Q^*)$,
        and it suffices to use \cref{prp:candidatepos}.
        By doing so, we also stay within in the claimed time and query complexities.

        In the remaining part of the proof we handle the case $|Q| > 1$.
        Consider the following procedure that returns two indices and a cost value.
        \begin{enumerate}[(i)]
        \item Divide $S$ into blocks of length $|Q|$, obtaining $\floor{|S|/|Q|} \geq 2k+1$ full blocks.
        \label{lem:find_min_approx_per:i}
        \item Sample u.a.r. a block $B = S\fragmentco{b}{b+|Q|}$,
            and via \cref{claim:hd_compstructure_rot} find the rotation $a \in \fragmentco{0}{|Q|}$
            for which $B=Q_a$ holds, where $Q_a \coloneqq \rot^{a}(Q)$ (or verify that there is no such $a$).
        \label{lem:find_min_approx_per:ii}
        \item If there is no such $a$, report a failure and return.
        \label{lem:find_min_approx_per:iii}
        \item Combining \cref{prp:candidatepos} with an exponential/binary search,
            verify whether $\eds{S\fragmentco{0}{b}}{Q_a} \leq k$ and $\ed(S\fragmentco{b+|Q|}{m},Q_a^*) \leq k$ hold.
            If this is not the case, report a failure and return.
        \label{lem:find_min_approx_per:iv}
        \item Otherwise, let $b^-,b^+$ be the returned indices by \cref{prp:candidatepos} in the form as remark shortly before \cref{lem:find_min_approx_per}.
        Return the cost value $k' \coloneqq \eds{S\fragmentco{0}{b}}{Q_a} + \ed(S\fragmentco{b+|Q|}{m},Q_a^*) \leq k$ together with $\ell \coloneqq b^-+b+j|Q|+a, r \coloneqq b+(j+1)|Q|+b^++a$ where $j$ is the smallest $j \in \Zz$ such that $\ell,r \geq 0$.
        \label{lem:find_min_approx_per:v}
    \end{enumerate}

   The following claim analyzes a single run of the procedure.
   \begin{claim}\label{claim:find_min_approx_per}
       With constant probability we do not fail and
       $\edl{S}{Q} = k'$ holds.
       Moreover, given we do not fail and return a cost $k'$,
       then $k' = \ed(S,Q^{\infty}\fragmentco{\ell}{r})$ and $\edl{S}{Q} \leq k'$
       always holds.
   \end{claim}

   \begin{claimproof}
    Consider an optimal alignment $\mX$ involved in $\edl{S}{Q}$.
    Observe that $\mX$ can make edits in at most $k$ of the full blocks.
    In each of the at most $k+1$ other blocks,
    no edit occurs, which implies that they are all equal to a rotation of $Q$.
    Consequently, with probability at least $(k+1)/(2k+1) \geq 1/2$, we have that $\mX$ matches $B$ with $Q_b$ for some rotation of $Q_b$ of $Q$.

    Further, assume that this even happens.
    Thus, $Q_a = Q_b$ and the algorithm proceeds to \eqref{lem:find_min_approx_per:iv},
    setting $a$ such that $a \equiv_{|Q|} b$.
    By decomposing $\mX$ into three parts, we obtain:
    \[
        \edl{S}{Q} = \eds{S\fragmentco{0}{b}}{Q_b} + \underbrace{\ed(B, Q_b)}_{= 0} + \ed(S\fragmentco{b+|Q|}{m},Q_b^*).
    \]
    As $\edl{S}{Q} \leq k$,
    the first and the third term of this last sum are upper bounded by $k$.
    Moreover, these two terms must equal to $\eds{S\fragmentco{0}{b}}{Q_a}$ and $\ed(S\fragmentco{b+|Q|}{m},Q_a^*)$, respectively.
    We conclude that the algorithm proceeds to \eqref{lem:find_min_approx_per:v},
    and sets $k' = \edl{S}{Q}$.

    For the second part of the claim, we observe that
    by the definition of $\ell,r$ we have
    $\eds{S\fragmentco{0}{b}}{Q_a} = \eds{S\fragmentco{0}{b}}{Q^{\infty}\fragmentco{\ell}{b+j|Q|+a}}$ and
    $\ed(S\fragmentco{b+|Q|}{m},Q_a^*) = \eds{S\fragmentco{0}{b}}{Q^{\infty}\fragmentco{\ell}{b+j|Q|+a}}$.
    Consequently, $k' = \ed(S,Q^{\infty}\fragmentco{\ell}{r})$.
    Lastly, by the triangle inequality, we conclude
    \[
        k' =
        \eds{S\fragmentco{0}{b}}{Q_b} + \ed(B, Q_b) + \ed(S\fragmentco{b+|Q|}{m},Q_b^*)
        \leq \edl{S}{Q}.
        \claimqedhere
    \]
   \end{claimproof}

    To recover a quantum algorithm with the same as described in the statement
    of this lemma, it suffices to repeat $\Oh(\log n)$ times the procedure.
    If we obtained at least once not a failure, we select the minimum obtained cost across all runs that did not fail,
    and return it together with the two corresponding indices $\ell,r$.
    Otherwise, we return that $\edl{S}{Q} > k$.
    The correctness follows from \cref{claim:find_min_approx_per}.

   The quantum and query time required for a single execution of the procedure is
   dominated by \eqref{lem:find_min_approx_per:iv} which requires $\Oh(\sqrt{k|S|})$ queries and $\Oh(\sqrt{|S|m}+|S|^2)$ time. The total quantum time for the algorithm differs from that of a single execution by only a logarithmic factor.
\end{proof}

\subsection{Quantum Algorithm for Gap Edit Distance}
\label{sec:quantumgaped}

In this section we  show how to devise a quantum subroutine for the \gaped problem which is defined as follows.

\defproblem{$(\beta,\alpha)$-\gaped}%
{strings $X,Y$ of length $|X| = |Y| =n$, and integer thresholds $\alpha \geq \beta \geq
0$.}%
{\yes if $\ed(X,Y) \leq \beta$, \no if $\ed(X,Y) > \alpha$, and an arbitrary answer
otherwise.}

We formally prove the following.

\begin{lemma}\label{lem:quantum_gap_ed}
    For all positive integers $k$ and $n$, there exists a positive integer $\ell = \Oh(n^{1+o(1)})$
    and a function $f: {\{0, 1\}}^{\ell} \times \Sigma^n \times \Sigma^n \rightarrow \{\text{\yes}, \text{\no}\}$ such that:
    \begin{itemize}
         \item if $\mathrm{seed}\in {\{0, 1\}}^{\ell}$ is sampled uniformly at random, then,
         for every $X, Y \in \Sigma^n$, the value $f(\mathrm{seed},X, Y)$ is with high probability a correct answer to the $(k, kn^{o(1)})$-\gaped instance with input $X$ and $Y$;
        \item there exists a quantum algorithm that computes $f(\mathrm{seed},X, Y)$
        in $\Oh(n^{1 + o(1)})$ quantum time using $\Oh(n^{0.5 + o(1)})$
        queries providing quantum oracle access to  $X, Y \in \Sigma^n$ and $\mathrm{seed}\in {\{0, 1\}}^{\ell}$.\qedhere
    \end{itemize}
\end{lemma}

The quantum algorithm developed for the gap edit distance problem
is a direct adaptation of Goldenberg, Kociumaka, Krauthgamer, and Saha's classical algorithm~\cite{gapED} for the same problem.
The paper~\cite{gapED} presents several algorithms for the gap edit problem,
each progressively improving performance. Among these algorithms,
we opt to adapt the one from Section 4,
stopping short of the faster implementation discussed in Section 5.

Notably, all algorithms presented in~\cite{gapED} are non-adaptive,
meaning they do not require access to the string to determine which recursive calls to make.
This characteristic allows us to pre-generate randomness using a sequence of bits from $\{0, 1\}^{\ell}$.
Consequently, when provided with a fixed sequence of bits,
the algorithm's output becomes deterministic,
producing the same output not only for cases where $\ed(X,Y) \leq \beta$ and $\ed(X,Y) > \alpha$
but also for scenarios where $\beta < \ed(X,Y) \leq \alpha$.

In \cref{sec:gapedpresentation}, we  first briefly present the algorithm of~\cite{gapED}.
Then, in \cref{sec:gapedadaptation}, we argue how the algorithm can be extended to the quantum framework.
For a more rigorous proof, an intuition for why the algorithm works, and a more comprehensive introduction to the problem, we recommend referring directly to~\cite{gapED}.

\subsubsection{Gap Edit Distance and Shifted Gap Edit Distance}
\label{sec:gapedpresentation}

The algorithm presented in \cite{gapED} solves an instance of \gaped by querying an oracle solving smaller instances of a similar problem: the \shifted \ \gaped problem.
Before properly defining this problem, we need to introduce the notion of \emph{shifted edit distance}.

\begin{definition}
Given two strings $X,Y$ and a threshold $\beta$, define the \emph{$\beta$-shifted edit distance} $\edshifted{\beta}(X,Y)$ as
\[
    \edshifted{\beta}(X,Y) \coloneqq \min \left(\bigcup_{\Delta=0}^{\min(|X|,|Y|,\beta)} \left\{\ed(X\fragmentco{\Delta}{|X|}, Y\fragmentco{0}{|Y| - \Delta}), \ed(X\fragmentco{0}{|X| - \Delta}), Y\fragmentco{\Delta}{|Y|}\right\}\right).\qedhere
\]
\end{definition}

The \shifted \ \gaped problem consists in either verifying that the shifted edit distance between two strings $X,Y$ is small, or that the (regular) edit distance between $X$ and $Y$ is large.
More formally:

\defproblem{$\beta$-\shifted $(\gamma,3\alpha)$-\gaped}
{strings $X,Y$ of length $|X| = |Y| =n$, and integer thresholds $\alpha \geq \beta \geq
\gamma \geq 0$.}%
{\yes if $\edshifted{\beta}(X,Y) \leq \gamma$, \no if $\ed(X,Y) > 3\alpha$, and an
arbitrary answer otherwise.}

The insight of~\cite{gapED} is that an instance of \shifted \ \gaped can be reduced back to smaller instances of \gaped.
In this way,~\cite{gapED} obtains a recursive algorithm switching between instances of \shifted \ \gaped and of \gaped.
The size of the instances shrinks throughout different calls until the algorithm is left with instances having trivial solutions.

\subparagraph*{From Gap Edit Distance to Shifted Gap Edit Distance}

The algorithm that reduces \gaped instances to \shifted \ \gaped instances
partitions the strings $X$ and $Y$ into blocks of different lengths.
It samples these blocks and utilizes them as input for instances of the \shifted \ \gaped problem.
If a limited number of instances return \no, the routine outputs \yes.
Otherwise, it returns \no.
A description is provided in \cref{alg:gapED}.

\begin{algorithm}[t]
\KwInput{An instance $(X,Y)$ of $(\beta, \alpha)$-\gaped, and a parameter $\phi \in
\mathbb{Z}_+$ satisfying
\[
    \phi \geq \beta \geq \psi \coloneqq \floor{({112\beta \phi
\ceil{\log n}})/{\alpha}}.
\]}
Set $\rho \coloneqq {84\phi}/{\alpha}$\;
\For{$p \in \fragment{\ceil{\log (3\phi)}}{\floor{\log(\rho n)}}$}{
    Set $m_p \coloneqq \ceil{{n}/{2^p}}$ to be number of blocks of length $2^p$ we partition $X,Y$ into (last block might be shorter)\;
    For $i \in \fragmentco{0}{m_p}$ set $X_{p,i} = X\fragmentco{i \cdot 2^p}{\min(n, (i+1)2^p)}$ and $Y_{p,i} = Y\fragmentco{i \cdot 2^p}{\min(n, (i+1)2^p)}$\;
    \For {$t \in \fragmentco{0}{\ceil{\rho m_p}}$}{
        Select $i \in \fragmentco{0}{m_p}$ uniformy at random\;
        Solve the instance $(X_{p,i},Y_{p,i})$ of the $\beta$-\shifted $(\psi,3\alpha)$-\gaped problem\; \label{ln:gapedrec}
    }
}
\If{we obtained at most 5 times \no at \cref{ln:gapedrec}}{
    \Return{\yes}
}\Else{
    \Return{\no}
}
    \caption{Randomized Algorithm reducing \gaped to \shifted \ \gaped}\label{alg:gapED}
\end{algorithm}

Goldenberg, Kociumaka, Krauthgamer, and Saha proved in~\cite{gapED} the following.
\begin{lemma}[Lemma 4.3 of~\cite{gapED}]
Suppose all oracle calls at \cref{ln:gapedrec} of \cref{alg:gapED}, return the correct answer.
Then, \cref{alg:gapED} solves an instance $(X,Y)$ of $(\beta, \alpha)$-\gaped with error probability at most $e^{-1}$. \lipicsEnd
\end{lemma}

\subparagraph*{From Shifted Gap Edit Distance to Gap Edit Distance}

In contrast to the previous algorithm,
the algorithm reducing \shifted \ \gaped instances to \gaped instances is deterministic.
For an instance $(X,Y)$ of $\beta$-\shifted $(\gamma, 3\alpha)$-\gaped
satisfying $\alpha \geq 3\gamma$,
the algorithm makes several calls to an oracle that solves
$(\gamma, 3\alpha)$-\gaped instances.
The routine returns \yes if and only if
at least one such oracle call returned \yes.
A description is provided in \cref{alg:shiftedgapED}.

\begin{algorithm}[t]
\KwInput{An instance $(X,Y)$ of $\beta$-\shifted $(\gamma, 3\alpha)$-\gaped
satisfying $\alpha \geq 3\gamma$.}
Set $\xi \coloneqq \floor{\sqrt{(1+\beta)(1+\gamma)}} - 1$ and $n' = n - \beta$\;
For all $x \in \fragment{0}{\beta}$ such that $x \equiv_{1+\xi} \beta$ or $x \equiv_{1+\xi} 0$,
all $y \in \fragment{0}{\xi}$ such that $y \equiv_{1+\gamma} 0$,
all $y \in \fragment{\beta-\xi}{\beta}$ such that $y \equiv_{1+\gamma} \beta$,
solve the instance $(X\fragmentco{x}{x+n'}, Y\fragmentco{y}{y+n'})$ of the
$(3\gamma, \alpha)$-\gaped problem\; \label{ln:shiftedgapedrec}

\If{we obtained at least one time \yes at \cref{ln:shiftedgapedrec}}{
    \Return{\yes}
}\Else{
    \Return{\no}
}
    \caption{Deterministic Algorithm reducing \shifted \ \gaped to \gaped.}\label{alg:shiftedgapED}
\end{algorithm}

\begin{lemma}[Lemma 4.3 of \cite{gapED}]
    Suppose all oracle calls at \cref{ln:shiftedgapedrec} of \cref{alg:shiftedgapED}, return correct answer.
    Then, \cref{alg:shiftedgapED} solves an instance $(X,Y)$ of $\beta$-\shifted $(3\gamma, \alpha)$-\gaped. \lipicsEnd
\end{lemma}

\subparagraph*{The Base Case}
The algorithm is structured to switch
between instances of \cref{alg:gapED} and \cref{alg:shiftedgapED}
until it reaches an instance of the $(\beta, \alpha)$-\gaped problem with $\beta = 0$.
This instance represents the base case of the algorithm.
To handle it, the algorithm simply checks whether there is at least one position $i \in \fragmentco{0}{n}$ such that $X[i] \neq Y[i]$, that is,
it checks whether the Hamming distance $\hd(X, Y)$ is non-zero.
If such a position is found, the algorithm returns \no;
otherwise, it returns \yes.
This check is sufficient because if $\ed(X, Y) = 0$,
then no such position exists and $\hd(X, Y) = 0$.
Conversely, if $\ed(X, Y) \geq \alpha$,
there must be at least $\alpha$ positions $i \in \fragmentco{0}{n}$ such that $X[i] \neq Y[i]$, that is, $\hd(X, Y) \geq \ed(X, Y) \geq \alpha$.
If we sample ${c n}/{(1 + \alpha)}$ positions,
the probability of sampling at least one of those $\alpha$ positions is at most
\[
    {\left(1 - {\alpha}/{n}\right)}^{{c n}/{(1 + \alpha)}} \leq e^{-{c \alpha}/({1 +
    \alpha})}.
\]
By choosing $c = \Oh(1)$ appropriately, we may assume any constant lower bound for the success probability.

\subsubsection{Adapting the Algorithm to the Quantum Setting}
\label{sec:gapedadaptation}

The algorithm
is formally constructed through Proposition 4.6 and Proposition 4.7 in \cite{gapED}.
These propositions prove, respectively, the existence of an algorithm for the \gaped problem and for the \shifted \ \gaped problem which, together with the information for the respective instance, take in input a parameter $h$, describing the depth of the recursive calls we make until we are at the base described in the previous (sub)section.

In this (sub)section we want to argue that we can take Proposition 4.6 and Proposition 4.7 in \cite{gapED} and adapt them to the quantum setting, by creating subroutines that replicate their behavior when provided in input the same sequence of bits.
The adaptation to the quantum case is rather straightforward, since the outcome of both
\cref{alg:gapED} and \cref{alg:shiftedgapED} relies on independent calls to other oracles
that simply return \yes or \no, and such structure can be exploited by \GS.

\begin{proposition}(Compare with Proposition 4.6 of  \cite{gapED})\label{prp:quantum_gap_ed_time}
    There exists a quantum algorithm with the following properties:
    \begin{itemize}
        \item it takes as input $h \in \mathbb{Z}_{\geq 0}$, $\epsilon \in \mathbb{R}_{+}$,
        and an instance of the $(\beta, \alpha)$-\gaped problem, satisfying
        $\beta < (336\ceil{\log n})^{{-h}/{2}}\alpha^{{h}/{(h+1)}}$,
        and a sequence of bits $\mathrm{seed} \in \{0, 1\}^{\ell}$
        of length
        \[\ell = \mathcal{O}\left({(1+\beta)}/{(1+\alpha)}\cdot n \cdot
            \log^{\mathcal{O}(h)} n \cdot \log {\epsilon}^{-1} \cdot
             2^{\mathcal{O}(h)}\right);
        \]
        \item it outputs with high probability the same output of the non-adaptive
            algorithm from Proposition 4.6 of \cite{gapED}
            with input $h$, $\epsilon$, the instance of the $(\beta, \alpha)$-\gaped problem,
            and $\mathrm{seed}$; and
        \item it uses \[
                \mathcal{O}\left(({1+\beta})/({1+\alpha})\cdot n \cdot \log^{\mathcal{O}(h)} n
                \cdot \log {\epsilon}^{-1} \cdot 2^{\mathcal{O}(h)}\right)
            \text{ time} \]
            and
            \[
                \mathcal{O}\left({\left(({1+\beta})/({1+\alpha})\right)}^{{1}/{2}}\cdot
                n^{{1}/{2}} \cdot \log^{\mathcal{O}(h)} n \cdot \log {\epsilon}^{-1} \cdot 2^{\mathcal{O}(h)}\right)
                \text {queries.}
            \] \lipicsEnd
    \end{itemize}
\end{proposition}

\begin{proposition}(Compare with Proposition 4.7 of  \cite{gapED}) \label{prp:quantum_shifted_gap_ed_time}
    There exists a quantum algorithm with the following properties:
    \begin{itemize}
        \item it takes as input $h \in \mathbb{Z}_{\geq 0}$, $\epsilon \in \mathbb{R}_{+}$,
        and an instance of the $\beta$-\shifted $(\gamma, 3\alpha)$-\gaped problem, satisfying
        $\gamma < {1}/{3} \cdot(336\ceil{\log n})^{{-h}/{2}}\alpha^{{h}/({h+1})}$,
        and a sequence of bits $\mathrm{seed} \in \{0, 1\}^{\ell}$
        of length
        \[\ell = \mathcal{O}\left(({1+\beta})/({1+\alpha})\cdot n \cdot
            \log^{\mathcal{O}(h)} n \cdot \log ({n}/{\epsilon}) \cdot
    2^{\mathcal{O}(h)}\right); \]
        \item it outputs with high probability the same output of the non-adaptive
            algorithm from Proposition 4.7 of \cite{gapED}
            with input $h$, $\epsilon$, the instance of the $\beta$-\shifted $(\gamma, 3\alpha)$-\gaped problem,
            and $\mathrm{seed}$; and
        \item it uses \[
            \mathcal{O}\left(({1+\beta})/({1+\alpha})\cdot n \cdot \log^{\mathcal{O}(h)} n
            \cdot \log ({n}/{\epsilon}) \cdot 2^{\mathcal{O}(h)}\right)
            \text{ time,}
            \] and \[
            \mathcal{O}\left({\left(({1+\beta})/({1+\alpha})\right)}^{{1}/{2}}\cdot
            n^{{1}/{2}} \cdot \log^{\mathcal{O}(h)} n \cdot \log ({n}/{\epsilon}) \cdot 2^{\mathcal{O}(h)}\right)
            \text{queries.}
            \tag*{\lipicsEnd}
        \]
    \end{itemize}
\end{proposition}

The aforementioned results
\cite[Propositions 4.6 and 4.7]{gapED}
are proved by induction over $h$,
alternating between Proposition 4.6 and Proposition 4.7 of \cite{gapED},
mirroring the structure of the recursive calls of the algorithm.

In order to ensure correctness with probability at least $\epsilon$,
the propositions use the fact that
we can boost an algorithm's success probability from any constant
to $1 - \epsilon$
by repeating the algorithm $\mathcal{O}(\log {\epsilon}^{-1})$ times.
We prove \cref{prp:quantum_gap_ed_time} and \cref{prp:quantum_shifted_gap_ed_time}
using the same structure of induction.

\begin{proof}[Proof of \cref{prp:quantum_gap_ed_time} and \cref{prp:quantum_shifted_gap_ed_time}]
    As mentioned in the previous (sub)section,
    the base case is represented by \cref{prp:quantum_gap_ed_time} when $\beta=0$ or $h=0$.
    There, it suffices to interpret the seed as sampled positions
    from $\fragmentco{0}{n}$, and to execute \GS over all sampled positions checking whether
    there exists a sampled position $i \in \fragmentco{0}{n}$ such that $X\position{i} \neq Y\position{i}$,
    thereby obtaining the claimed query time and quantum time.

    \begin{description}
        \item[Induction step for \cref{prp:quantum_gap_ed_time}]

        If $\beta \neq 0$ and $h \neq 0$,
        the algorithm from Proposition 4.6 of \cite{gapED}
        uses \cref{alg:gapED} with $\phi = \beta$,
        employing as an oracle Proposition 4.7 of \cite{gapED} with parameters $h - 1$ and
        $\Theta({1}/{n})$.

        We adapt \cref{alg:gapED} to the quantum setting as follows.
        First, we extract the randomness required
        by the routine, that is, for each level $p \in \fragment{\ceil{\log (3\phi)}}{\floor{\log(\rho n)}}$
        we have to read
        \[\ceil{\rho m_p} = \ceil{{84\phi n}/{\alpha 2^p}}\]
        new positions from the sequence bits from the input.

        This allows us to execute \GS over all
        selected instances separately for each of the levels.
        Our goal is to detect at least
        five oracle calls returning \no.
        Since we are only looking for constantly many of those, it suffices to pay an
        additional $\tilde{\mathcal{O}}(1)$ factor in \GS.
        The quantum time becomes
        \begin{align*}
            & \sum_{p=\ceil{\log (3\phi)}}^{\floor{\log(\rho n)}}
            \mathcal{O}\left(\ceil{\rho m_p} + \ceil{\rho m_p}^{0.5} \cdot {\beta}/{\phi}
            \cdot 2^{p} \cdot \log^{\mathcal{O}(h-1)} n \cdot \log n \cdot
        2^{\mathcal{O}(h-1)}\right)\\
            & \leq \sum_{p=\ceil{\log (3\phi)}}^{\floor{\log(\rho n)}}
            \mathcal{O}\left(\ceil{\rho m_p} \cdot {\beta}/{\phi} \cdot 2^{p} \cdot \log^{\mathcal{O}(h-1)} n \cdot \log n \cdot 2^{\mathcal{O}(h-1)}\right) \\
            & \leq \mathcal{O}\left( {\beta}/{\alpha} \cdot n \cdot \log^{\mathcal{O}(h)} n \cdot 2^{\mathcal{O}(h)} \right)
        \end{align*}
        Similarly, the query complexity becomes
        \[
            \sum_{p=\ceil{\log (3\phi)}}^{\floor{\log(\rho n)}}
            \mathcal{O}\left(\ceil{\rho m_p}^{\frac{1}{2}}
            {\left({\beta}/{\phi}\right)}^{0.5} \cdot 2^{{p}/{2}} \cdot
        \log^{\mathcal{O}(h-1)} n \cdot \log n \cdot 2^{\mathcal{O}(h-1)}\right)
            = \mathcal{O}\left( {\left({\beta}/{\alpha}\right)}^{0.5} n^{0.5} \cdot \log^{\mathcal{O}(h)} n \cdot 2^{\mathcal{O}(h)} \right).
        \]
        \item[Induction step for \cref{prp:quantum_shifted_gap_ed_time}]
            The algorithm from \cite[Proposition 4.7]{gapED}
        uses \cref{alg:shiftedgapED} employing as an oracle
        \cite[Proposition 4.6]{gapED} with parameters $h$ and $\Theta({\epsilon}/{n})$.

        Adapting \cref{alg:shiftedgapED} is even more straightforward, since it suffices
        to use \GS over all
        \[
            4\ceil{({1+\beta})/({1+\xi})}\ceil{({1+\xi})/({1+\gamma})} \leq 16 \cdot
            ({1+\beta})/({1+\gamma})
        \] independent oracle calls, where $1+\xi = \floor{\sqrt{(1+\beta)(1+\gamma)}}$.

        Observe that each of the oracle calls involves a string of length at most $n$.
        Therefore, the quantum time becomes
        \begin{align*}
            &\mathcal{O}\left({\left(({1+\beta})/({1+\gamma})\right)}^{0.5} \cdot
                ({1+3\gamma})/({1+\alpha}) \cdot n \cdot \log^{\mathcal{O}(h)} n \cdot \log n
        \cdot 2^{\mathcal{O}(h)}\right)\\
        &\quad \leq
            \mathcal{O}\left(({1+\beta})/({1+\alpha}) \cdot n \cdot \log^{\mathcal{O}(h)} n \cdot \log n \cdot 2^{\mathcal{O}(h)}\right).
        \end{align*}
        Similarly, the query complexity becomes
        \begin{align*}
            &
            \mathcal{O}\left({\left(({1+\beta})/({1+\gamma})\right)}^{0.5} \cdot
            {\left(({1+3\gamma})/({1+\alpha})\right)}^{0.5} \cdot n^{0.5}
        \cdot \log^{\mathcal{O}(h)} n \cdot \log n \cdot 2^{\mathcal{O}(h)}\right) \\
        &\quad =
            \mathcal{O}\left({\left(({1+\beta})/({1+\alpha})\right)}^{0.5}
            \cdot n^{0.5} \cdot \log^{\mathcal{O}(h)} n \cdot \log n \cdot 2^{\mathcal{O}(h)}\right).
            \tag*{\qedhere}
        \end{align*}
    \end{description}
\end{proof}

\begin{proof}[Proof of~\cref{lem:quantum_gap_ed}]
    We apply, \cref{prp:quantum_gap_ed_time},
    setting $\beta = k$, $\alpha = kn^{o(1)}$ and $\epsilon$ polynomially small.
    Further, we set $h$ to be the smallest integer such that
    \[\beta < (336\ceil{\log n})^{{-h}/{2}}\alpha^{{h}/({h+1})}.\]
    By doing so, we obtain the desired quantum time and query time.
\end{proof}

\section{Near-Optimal Quantum Algorithm for Pattern Matching with Edits}
\label{sec:pmwe}

In the following section, we  demonstrate how the combinatorial results from \cref{sec:combres} pave the way
for a quantum algorithm with nearly optimal query time up to a sub-polynomial factor.
As already mentioned earlier in \cref{sec:overview}, the development of such an algorithm also relies on other components
such as the combinatorial results of~\cite{CKW20} and the quantum subroutine for the Gap Edit Distance problem from \cref{sec:quantumgaped}.

The primary result we focus on is not \cref{thm:qpmwe} itself, but rather \cref{thm:qedfindproxystrings}.
Specifically, given a pattern \( P \) of length \( m \), a text \( T \) of length \( n \le 3/2 \cdot m \), and an integer threshold \( k > 0 \), we show how to compute proxy strings $P^\#, T^\#$ that are equivalent to the original $P,T$ for the purpose of finding $k$-error occurrences.

\begin{restatable*}{theorem}{qedfindproxystrings}\label{thm:qedfindproxystrings}
    There exists a quantum algorithm that, given a pattern $P$ of length $m$, a text $T$ of length $n \le 3/2 \cdot m$, and an integer threshold $k > 0$, outputs strings $P^\#$ and $T^\#$ such that, for every fragment $T\fragmentco{t}{t'}$:
    \begin{enumerate}[(a)]
    \item if $\ed(P,T\fragmentco{t}{t'})\le k$, then $\ed(P, T\fragmentco{t}{t'}) = \ed(P^\#, T^\#\fragmentco{t}{t'})$ and $\sE_{P, T}(\mX) = \sE_{P^\#, T^\#}(\mX)$ holds for every optimal alignment $\mX : P \onto T\fragmentco{t}{t'}$;
    \item $\ed(P, T\fragmentco{t}{t'}) \leq \ed(P^\#, T^\#\fragmentco{t}{t'})$.
    \end{enumerate}
    The algorithm uses $\Ohhat(\sqrt{km})$ queries and $\Ohhat(\sqrt{k}m+k^2)$ time.
\end{restatable*}

We show how to derive \cref{thm:qpmwe} from \cref{thm:qedfindproxystrings}.

\qpmwe

\begin{proof}
    We first show how to prove \cref{thm:qpmwe} assuming that we have a quantum algorithm
    $\mathsf{A}$ that solves \PMwE when $n \leq 5/4 \cdot m+k$, using $\Ohhat(\sqrt{km})$ queries and $\Ohhat(\sqrt{k}m + k^{3.5})$ time.

    Consider partitioning \(T\) into \(\Oh(n/m)\) contiguous blocks of length \(\ceil{m/4}\) each (with the last block potentially being shorter).
    For every block that starts at position $i \in \fragmentco{0}{|T|}$,  we consider a segment $T_i\coloneqq T\fragmentco{i}{\min(n, i + \floor{5/4 \cdot m}+k)}$ of length at most $5/4 \cdot m+k$.
    Every $k$-error occurrence \(T\fragmentco{t}{t'}\) starting at $t\in \fragmentco{i}{i+\ceil{m/4}}$ is of length at most $m+k$ and thus is fully contained in the segment $T_i$.

    For the first claim, we iterate over all such segments, applying $\mathsf{A}$ to \(P\), the current segment $T_i$, and \(k\).
    The final set of occurrences is obtained by taking the union of all sets returned by the quantum algorithm.

    For the second claim, we employ \GS over all segments, utilizing a function that determines whether the set of $k$-error occurrences of $P$ in the segment $T_i$ is empty or not.

    For the purpose of showing how to satisfy the assumption,  we distinguish three different regimes of $k$:
    \begin{itemize}
        \item If $k < m/4$, we set as $\mathsf{A}$ an algorithm that first uses \cref{thm:qedfindproxystrings} to obtain the proxy strings $P^\#$ and $T^\#$ and then applies the classical algorithm of \cref{prp:classical_pmwe} to compute $\OccE_k(P^\#, T^\#)$.
        The result is returned as $\OccE_k(P, T)$.
        Observe that the conditions for \cref{thm:qedfindproxystrings} are satisfied since $5/4\cdot m +k < 3/2 \cdot m$.
        \item If $m/4 \le k \leq m$, we set as $\mathsf{A}$ an algorithm that first reads $P$ and $T$ completely and then directly applies the classical algorithm of \cref{prp:classical_pmwe} to compute $\OccE_k(P, T)$.
        From $k = \Omega(m)$ and $n =\Oh(m)$ follows that this requires $m+n = \Oh(\sqrt{km})$ queries and $\Ohtilde(m + k^{3.5})$ time.
        \item If $m \leq k$, then $\mathsf{A}$ simply returns $\OccE_k(P,T)=\fragment{0}{n}$ represented as an arithmetic progression.
        \qedhere
    \end{itemize}
\end{proof}

We prove that the query complexity is optimal for $k = o(m)$.

\qpmwelb
\begin{proof}
   Let $p = \lfloor{(n+m)/(2m)}\rfloor$. For a tuple $\Ab=(A_0,\ldots,A_{p-1})$ of $p$ subsets of  $\fragmentco{0}{m}$, consider a text $T_{\Ab}\in \{\zero,\one\}^n$ such that:
    \[T_{\Ab}[i] = \begin{cases}
        \zero & \text{if }i=2qm+r\text{ for }q\in \fragmentco{0}{p}\text{ and }r\in \fragmentco{0}{m}\setminus A_{q},\\
        \one & \text{otherwise.}
    \end{cases}\]

    \begin{claim}\label{claim:lb}
        For $q\in \fragmentco{0}{p}$, we have $\OccE_k(P,T_{\Ab})\cap \fragmentco{2qm-m}{2qm+m} \ne \emptyset$ if and only if $|A_q|\le k$.
    \end{claim}
    \begin{claimproof}
        Observe that $\ed(P,T_{\Ab}\fragmentco{2qm}{2qm+m})=|A_q|$ because $T_{\Ab}\fragmentco{2qm}{2qm+m}$ consists of $|A_q|$ copies of $\one$ and $m-|A_q|$ copies of $\zero$. Thus, $|A_q|\le k$ implies $2qm\in \OccE_k(P,T_{\Ab})$.

        For the converse implication, suppose that $\ed(P,T\fragmentco{i}{j})\le k$ holds for some fragment $T_{\Ab}\fragmentco{i}{j}$ with $i\in \fragmentco{2qm-m}{2qm+m}$.
        In particular, this means that $T_{\Ab}\fragmentco{i}{j}$ contains at least $m-k$ copies of $\zero$ and at most $k$ copies of $\one$.
        The fragment $T_{\Ab}\fragmentco{2qm}{2qm+m}$ is separated from other $\zero$s in $T_{\Ab}$ by at least $m>k$ copies of $\one$ in both directions, so all $\zero$s contained in $T_{\Ab}\fragmentco{i}{j}$ must lie within  $T_{\Ab}\fragmentco{2qm}{2qm+m}$.
       The number of $\zero$s within  $T_{\Ab}\fragmentco{2qm}{2qm+m}$ is $m-|A_q|$.
       Consequently, $m-|A_q|\ge m-k$.
    \end{claimproof}

    We prove the query complexity lower bounds using \cref{thm:adv}.
    Consider a set $X$ consisting of texts $T_{\Ab}$ such that $|\{q : |A_q|=k+1\}|=\lceil(p+1)/2\rceil$ and $|\{q : |A_q|=k\}|=\lfloor(p-1)/2\rfloor$,
    as well as set a $Y$ consisting of all texts $T_{\Ab}$ such that $|\{q : |A_q|=k+1\}|=\lceil(p-1)/2\rceil$ and $|\{q : |A_q|=k\}|=\lfloor(p+1)/2\rfloor$.
    Moreover, let $R\subseteq X\times Y$ be a relation consisting of instances that differ by exactly one substitution in $T_{\Ab}$. Observe that:
    \begin{itemize}
        \item Every instance in $X$ is in relation with exactly $\lfloor(p+1)/2\rfloor \cdot  (k+1)$ instances in $Y$.
        \item Every instance in $Y$ is in relation with exactly $\lceil(p+1)/2\rceil \cdot  (m-k)$ instances in $X$.
        \item Instances in $R$ differ at exactly one bit.
    \end{itemize}
    By \cref{claim:lb}, any (quantum) algorithm computing $\OccE_k(P,T)$ must be able to distinguish instances in $X$ from instances in $Y$.
    Consequently, the quantum algorithm has query complexity of at least \[\Omega\left(\sqrt{\lfloor(p+1)/2\rfloor \cdot  (k+1)\cdot \lceil(p+1)/2\rceil \cdot  (m-k)}\right)=\Omega\left(n/m\cdot \sqrt{k(m-k)}\right).\]

    Next, consider a set $X'$ consisting of all texts $T_{\Ab}$ such that $|\{q : |A_q|=k+1\}|=p$,
    as well as a set $Y'$ consisting of all texts $T_{\Ab}$ such that $|\{q : |A_q|=k+1\}|=p-1$ and $|\{q : |A_q|=k\}|=1$.
    Moreover, let $R'\subset X'\times Y'$ be a relation consisting of instances that differ by exactly one substitution in $T_{\Ab}$.
    Observe that:
    \begin{itemize}
        \item Every instance in $X$ is in relation with exactly $p \cdot  (k+1)$ instances in $Y'$.
        \item Every instance in $Y$ is in relation with exactly $m-k$ instances in $X'$.
        \item Instances in $R$ differ at exactly one bit.
    \end{itemize}
    Since $\max\OccE_k(P,T_{\Ab})< 2pm-m$, by \cref{claim:lb}, any (quantum) algorithm deciding $\OccE_k(P,T)\ne \emptyset$ must be able to distinguish instances in $X'$ from instances in $Y'$.
    Consequently, the quantum algorithm uses at least $\Omega\left(\sqrt{p\cdot (k+1) \cdot  (m-k)}\right)=\Omega\left(\sqrt{n/m}\cdot \sqrt{k(m-k)}\right)$ queries.
\end{proof}

\Cref{sec:pmwe} is organized as follows.
First, in \cref{sec:combres_implementation}, we explain how to construct the combinatorial structures
used throughout \cref{sec:combres}.
Next, in \cref{sec:proxy_strings_ed}, we prove \cref{thm:qedfindproxystrings} by utilizing the communication results along with two additional components:
the \gaped algorithm from \cref{lem:quantum_gap_ed},
and an algorithm capable of finding a candidate based on the findings of \cite{CKW20}.
We use the latter component in \cref{sec:proxy_strings_ed} without a proof.
We defer its proof to \cref{sec:qed_intro_analyze}.

\subsection{Construction of the Combinatorial Structures from \cref{sec:combres}}
\label{sec:combres_implementation}

In this (sub)section, we present algorithmic constructions for the combinatorial results from \cref{sec:combres}. Many of these constructions could be optimized for faster execution. However,  as our focus does not lie in achieving sublinear quantum time anyway, we present simpler constructions.

\begin{proposition} \label{prop:construct_g_s}
    Let $P$ be a string of length $m$, let $T$ be a string of length $n \leq 3/2 \cdot m$, and let $k > 0$ be a threshold.
    Further, let $S$ be a set of alignments of $P$ onto substrings of $T$ of cost at most $k$ such that $S$ encloses $T$.

    Then, there exists a (classical) algorithm that, given $\sE_{P,T}(\mX)$ for all $\mX : P \onto T\fragmentco{t}{t'}$ such that $\mX \in S$, constructs $\bG_S$ from \cref{def:bg} and computes $\bc(\bG_S)$.
    Moreover, if $\bc(\bG_S) > 0$, the algorithm finds $\tau_i^c$ and $\pi_j^c$ from \cref{def:pitau} for all $c \in \fragmentco{0}{\bc(\bG_S)}$,  $i \in \fragmentco{0}{n_c}$ and $j \in \fragmentco{0}{m_c}$.
    The algorithm requires $\Ohtilde(m \cdot |S|)$ time.
\end{proposition}

\begin{proof}

    Consider the following algorithm.
    \begin{enumerate}[(i)]
        \item Reconstruct all $\mX : P \onto T\fragmentco{t}{t'} \in S$
            using the edit information $\sE_{P, T}(\mX)$.
            \label{lem:constrgs:i}
        \item Construct $\bG_S$ from \cref{def:bg}: mark an edge red if it corresponds to at least one edit in any alignment in $S$, otherwise mark it as black.
            Counting the connected components containing only black edges yields $\bc(\bG_S)$.
            If $\bc(\bG_S) = 0$, report that $\bc(\bG_S) = 0$ and return.
            \label{lem:constrgs:ii}
        \item Select for every black connected component the smallest position of a character of $P$ contained in it.
            Then, sort in non-decreasing order these $\bc(\bG_S)$ positions and mark the black components
            containing the $c$-th character (0-indexed) in this sorted sequence
            as the $c$-th black connected component.
            \label{lem:constrgs:iii}
        \item For every $c \in \fragmentco{0}{\bc(\bG_S)}$, sort
            in non-decreasing order the characters of $P$ and $T$ contained in the $c$-th black component
            w.r.t their position at which they appear in $P$ and $T$.
            The $i$-th character in the first sorted sequence corresponds to $\tau_i^c$,
            and the $j$-th correponds to $\pi_j^c$.
            \label{lem:constrgs:iv}
    \end{enumerate}

    Regarding the time complexity,
    note that reconstructing each of the alignments in \(S\) results in \(|S|\) alignments,
    each with a size of \(\mathcal{O}(m + k)\).
    Consequently, the graph \(\bG_S\) has at most \(\mathcal{O}(m)\) vertices and \(\mathcal{O}(|S|\cdot m)\) edges.
    Hence, steps \eqref{lem:constrgs:i}, \eqref{lem:constrgs:ii}, \eqref{lem:constrgs:iii}, and \eqref{lem:constrgs:iv}
    can all be implemented in \(\Ohtilde(|S|\cdot m)\) time.
\end{proof}

Next, we introduce a data structure that, given a position in $T$ or $P$,
supports \emph{projection queries}.

\begin{definition}
    Let $P$ be a string of length $m$, $T$ a string of length $n \leq 3/2 \cdot \cdot m$, and let $k > 0$ be an integer threshold. Additionally, let $S$ be a set of alignments of $P$ onto $T$ such that $S$ encloses $T$ and $\bc(\bG_S) \neq 0$.

    A \emph{projection query} can take one of the following two forms:
    \begin{enumerate}[(a)]
        \item Given $y \in \fragmentco{0}{n}$, output $c$ such that $\Pi_T(y) = \tau_i^c$ if $\Pi_T(y) \neq -1$, or output $\bc(\bG_S) - 1$  otherwise.
        \label{def:proj_ds:a}
        \item Given $x \in \fragmentco{0}{m}$, output $c$ such that $\Pi_P(x) = \pi_i^c$ if $\Pi_P(x) \neq -1$, or output $\bc(\bG_S) - 1$ otherwise.
        \label{def:proj_ds:b} \qedhere
    \end{enumerate}
\end{definition}

\begin{proposition}\label{prop:proj_ds}
    Given the same input as \cref{prop:construct_g_s} such that $\bc(\bG_S) \neq 0$,
    we can construct a data structure in $\Ohtilde(m \cdot |S|)$ time
    that supports $\Oh(1)$-time projection queries.
\end{proposition}
\begin{proof}
    We briefly argue how to support queries of type \eqref{def:proj_ds:a};
    queries of type \eqref{def:proj_ds:b} can be handled similarly.
    Using \cref{prop:construct_g_s}, we first compute and sort all $\tau_i^c$ in non-increasing order.
    Iterating over all $y \in \fragmentco{0}{n}$ from $0$ to $n-1$,
    we keep a pointer to the largest $\tau_i^c$ of the sorted sequence such that $\tau_i^c \leq y$.
    By doing so, we can directly store the answer to the query with input $y$ in an array and retrieve them in $\Oh(1)$ time.
\end{proof}

Through a data structure supporting projection queries we can build the function $\w_S$ from \cref{lem:exfuncover}.

\begin{proposition} \label{prp:w_interval_sum}
    There exists a (classical) data structure that, given the same input as \cref{prop:construct_g_s} and access to $t_q$-time projection queries, can be constructed in $\Ohtilde(k|S|\cdot q_q)$ time.
    The data structure supports the following queries: Given $\fragmentco{\ell}{r} \subseteq \fragmentco{0}{\bc(\bG_S)}$ output $\sum_{c=\ell}^{r-1} \w_S(c)$ in $\Ohtilde(1)$ time, where $\w_S$ is defined as \cref{lem:exfuncover}.
\end{proposition}

\begin{proof}
    We build a balanced binary search tree where the keys are the identifiers in \(\fragmentco{0}{\bc(\bG_S)}\).
    Our final goal is to store all \(c \in \fragmentco{0}{\bc(\bG_S)}\) such that \(\w_S(c) \neq 0\).
    We store a key \(c \in \fragmentco{0}{\bc(\bG_S)}\) only if \(\w_S(c) \neq 0\), and we store along with it \(\w_S(c)\).
    Using common techniques to support range sums,
    we can output \(\sum_{c=i}^{j-1} \w_S(c)\) for any \(\fragmentco{i}{j} \subseteq \fragmentco{0}{\bc(\bG_S)}\) in $\Ohtilde(1)$ time.

    To construct \(\w_S\),
    we strictly follow the construction described in \cref{lem:exfuncover}.
    We begin with an empty binary search tree, meaning \(\w_S(c) = 0\) for all \(c \in \fragmentco{0}{\bc(\bG_S)}\).
    We iterate over the edits as described in \cref{lem:exfuncover}, and use
    a projection query to find the key that needs to be updated.
\end{proof}

Next, we show that how the black cover from \cref{prp:quantum_blackcover}
can be constructed in the quantum setting.

\begin{lemma} \label{prp:quantum_blackcover}
    The construction of a black cover described in \cref{def:periodcover} can be adapted to the quantum setting.

    That is, there exists a quantum algorithm that,
    given oracle access to \(P\) and \(\Ohtilde(1)\)-time oracle access to interval sums of a function \(\w_S\)
    that covers \(S\) with total weight \(w\),
    outputs a period cover \(C_S \subseteq \fragmentco{0}{\bc(\bG_S)}\) of compressed size \(\Ohtilde(w + k|S|)\).
    This algorithm requires \(\Ohtilde(\sqrt{wm})\) queries and runs in \(\Ohtilde(\sqrt{wm} + w^2)\) time.

    The representation of \(C_S\) consists of $\mathsf{sz} = \Ohtilde(w)$ intervals
    $\fragment{a_1}{b_1}, \ldots, \fragment{a_\mathsf{sz}}{b_\mathsf{sz}} \subseteq \fragmentco{0}{\bc(\bG_S)}$
    such that $C_S = \bigcup_{i=1}^{\mathsf{sz}} \fragment{a_i}{b_i}$. For each $i \in \fragment{1}{\mathsf{sz}}$ either
    \begin{center}
        \quad $\LZ(T\fragment{\tau_0^{a_i}}{\tau_0^{b_i}})$ \quad or \quad $\LZ(\rev{T\fragment{\tau_0^{a_i}}{\tau_0^{b_i}}})$
    \end{center}
    is provided. For such factorizations it holds
    \[
        \sum_{i=1}^{\mathsf{sz}} (1 - \mathrm{I}_i) \cdot |\LZ(T\fragment{\tau_0^{a_i}}{\tau_0^{b_i}})|
        + \mathrm{I}_i \cdot \LZ(\rev{T\fragment{\tau_0^{a_i}}{\tau_0^{b_i}}})
        = \Ohtilde(w + k|S|),
    \]
    where $\mathrm{I}_i \in \{0,1\}$ indicates whether the compressed representation along with the $i$-th interval is reversed.
\end{lemma}

\begin{proof}
    The algorithm computes $c_{\mathrm{pref}}$, $c_{\mathrm{suff}}$, $c_{\mathrm{lsuff}}$,
    and $c_{\mathrm{lpref}}$ and retrieves the $\LZ$ compression of the relative strings
    by combining \cref{prp:quantumlz} with an exponential/binary search.
    Furthermore, the algorithm computes $C^{0,\bc(\bG_S)-1}_S$
    using the recursive construction described in \cref{def:periodcover}.
    In the recursion, when computing $C_S^{i,j}$ for $\fragment{i}{j} \subseteq \fragmentco{0}{\bc(\bG_S)}$,
    the algorithm again uses \cref{prp:quantumlz} combined with an exponential/binary search
    to find the corresponding indices $i'$ and $j'$.
    Once the computation of $C_S^{i,j}$ is complete,
    the recursion returns $i'$, $h$, $j'$, $\LZ(\rev{T\fragment{\tau^{i'}_0}{\tau^h_0}})$,
    and $\LZ(T\fragmentoc{\tau^h_0}{\tau^{j'}_0})$,
    along with the same information from the lower recursion levels.
    The collection of intervals $\fragment{i'}{h}$ and $\fragmentoc{h}{j'}$
    (the latter may be extended to $\fragment{h}{j'}$ without changing $C_S$, as $h$ is already included in the former intervals) across all recursion calls form
    $\fragment{a_1}{b_1}, \ldots, \fragment{a_\mathsf{sz}}{b_\mathsf{sz}}$.
    The bound on $\mathsf{sz}$ and on the sum of the compressed representation is provided in the proof of \cref{prp:encode_funcover}.

    For the complexity analysis,
    note that the computation of $c_{\mathrm{pref}}$, $c_{\mathrm{suff}}$, $c_{\mathrm{lsuff}}$,
    and $c_{\mathrm{lpref}}$ via \cref{prp:quantumlz} combined with an exponential/binary search
    requires $\Ohtilde(\sqrt{(k + w)m} + (k + w)^2) = \Ohtilde(\sqrt{km} + k^2)$
    quantum time and $\Ohtilde(\sqrt{(k + w)m}) = \Ohtilde(\sqrt{km})$ query time.

    Regarding the analysis of the recursion,
    let $\fragment{i_0}{j_0}, \ldots, \fragment{i_{d-1}}{j_{d-1}}$
    be the intervals considered by it at a fixed depth.
    Further, define $w_{i,j} \coloneqq \sum_{c=i-1}^{j} \w_S(c)$ for $i, j \in \fragmentco{0}{\bc(\bG_S)}$.
    As already noted in the proof of \cref{prp:encode_funcover},
    $\sum_{r=0}^{d-1} w_{i_r, j_r} = \Oh(w)$.

    Suppose that for $\fragment{i}{j} \in \{\fragment{i_r}{j_r}\}_{r=0}^{d-1}$
    we add $\fragment{i'}{j'}$ to $C_S^{i,j}$.
    Computing $i'$ and $j'$ by combining \cref{prp:quantumlz}
    with an exponential/binary search requires
    $\Ohtilde(\sqrt{(j - i + 1) \cdot w_{i,j}})$ query time and
    $\Ohtilde(\sqrt{(j - i + 1) \cdot w_{i,j}} + w_{i,j}^2)$ quantum time.
    The expression $\sum_{r=0}^{d-1} \sqrt{(j_r - i_r + 1) \cdot w_{i_r,j_r}}$
    is maximized when the lengths $j_r - i_r + 1$ and the weights $w_{i_r,j_r}$
    are all roughly the same size.
    From $\sum_{r=0}^{d-1} w_{i_r, j_r} = \Oh(w)$ and
    $\sum_{r=0}^{d-1} (j_r - i_r + 1) = \Oh(m)$, we obtain
    that for the level of recursion we fixed the algorithm requires
    \begin{align*}
        \sum_{r=0}^{d-1} \Ohtilde\left(\sqrt{(j_r - i_r) \cdot w_{i_r,j_r}} + w_{i_r,j_r}^2 \right)
        = \sum_{r=0}^{d-1} \Ohtilde\left(\sqrt{{m}/{d}} \cdot \sqrt{{w}/{d}}\right) + \Ohtilde\left(w^2\right)
        = \Ohtilde (\sqrt{wm} + w^2)
        = \Ohtilde (\sqrt{km} + k^2)
    \end{align*}
    quantum time and $\Ohtilde(\sum_{r=0}^{d-1} \sqrt{(j_r - i_r + 1) \cdot w_{i_r,j_r}}) \leq \Ohtilde(\sqrt{km})$ query time.
    The overall claimed complexities follow by using the fact that the recursion has depth at most $\Ohtilde(1)$.
\end{proof}

Another operation which we needs to be implemented is the one
of checking whether a $k$-error occurrence $T\fragmentco{t}{t'}$ is captured.
To this end, it suffices to give a data structure that, given $t \in \fragmentco{0}{n}$,
computes $\min_{i \in \fragment{0}{n_0 - m_0}}|\tau_{i}^0 - (t+\pi_0^0)|$.

\begin{lemma}\label{prop:capture_ds}
    There exists a (classical) data structure that, given the same input as \cref{prop:proj_ds},
    can be constructed in $\Ohtilde(k|S|)$ time and can answer queries in $\Oh(1)$ time of the following type:
    Given $t \in \fragmentco{0}{n}$, compute $\min_{i \in \fragment{0}{n_0 - m_0}}|\tau_{i}^0 - (t+\pi_0^0)|$.
\end{lemma}

\begin{proof}
    Using \cref{prop:construct_g_s}, we first compute and sort all $\tau_i^0$ for $i \in \fragmentco{0}{n_0}$
    in non-increasing order.
    Iterating over all $y \in \fragmentco{0}{n}$ from $0$ to $n-1$,
    we keep a pointer to the largest $\tau_i^0$ of the sorted sequence such that $\tau_i^c \leq y+\pi_0^0$.
    By doing so, we can store for each $y \in \fragmentco{0}{n}$ the value $\min_{i \in \fragment{0}{n_0 - m_0}}|\tau_{i}^0 - (y+\pi_0^0)|$,
    as the minimum is achieved either by the current $\tau_i^0$ pointed to by the pointer, or by the next value $\tau_{i+1}^0$, if it exists.
\end{proof}

Finally, we conclude this (sub)section by showing
how to construct the strings $P,T$ and $P^\#,T^\#$
in the two respective cases that $\bc(\bG_S) = 0$
and $\bc(\bG_S) > 0$.

\begin{proposition} \label{prp:reconstruct_pt}
    There exists a (classical) algorithm that, given the same input as \cref{prop:construct_g_s},
    whenever $\bc(\bG_S) = 0$, reconstructs $P$ and $T$ in $\Ohtilde(m \cdot |S|)$ time.
\end{proposition}

\begin{proof}
    The edit information of $\mX$ for all $\mX \in S$ allows us to
    associate to each node in $\bG_S$ incident to a red edge
    the corresponding character from $P$ or $T$.
    We propagate characters to other nodes contained in red components of
    $\bG_S$ by executing a depth search using exclusively black edges.
    By doing so, we are guaranteed to retrieve the correct characters of $P$ and $T$
    because a black edge between two nodes indicates equality between
    the two corresponding characters.
    Moreover, from $\bc(\bG_S) = 0$ follows that all characters
    are contained in red components, and therefore
    this procedure delivers us a full reconstruction of
    the characters of $P$ and $T$.
\end{proof}

\begin{proposition} \label{prp:reconstruct_pt_hash}
    There exists a (quantum) algorithm that,
    given the same input as \cref{prop:construct_g_s}
    such that $\bc(\bG_S) > 0$,
    constructs $P^\#, T^\#$ from \cref{prp:subhash}
    using $\Ohtilde(\sqrt{k|S| \cdot m})$ queries and $\Ohtilde(m \cdot |S|)$ time.
\end{proposition}

\begin{proof}
    First, we identify all characters associated with nodes contained in red components, following the approach used in the proof of \cref{prp:reconstruct_pt}. For each character in the \( c \)-th black component, we assign a unique sentinel character \( \#_c \).

    Next, we retrieve a black cover. To achieve this, we first use the projection data structure from \cref{prop:proj_ds} to compute the data structure from \cref{prp:w_interval_sum}, which provides \(\Oh(1)\) interval sums of the function \( \w_S \) that covers \( S \) with a total weight \( w = \Oh(k|S|) \). This allows us to construct the black cover using \cref{prp:quantum_blackcover}.
    We then decompress all intervals from the black cover. For any character within a decompressed interval that belongs to a black component, we propagate its known value throughout all the nodes in that component, replacing the associated sentinel character.

    Finally, the strings \( P^\# \) and \( T^\# \) are obtained by selecting the associated character at each position.
\end{proof}

\subsection{Construction of the Proxy Strings
    \texorpdfstring{$P^\#$}{P-sharp}
    and
    \texorpdfstring{$T^\#$}{T-sharp}
via Candidate Sets}
\label{sec:proxy_strings_ed}

The goal of this (sub)section is to prove \cref{thm:qedfindproxystrings}.
One of the fundamental ingredients we use to prove \cref{thm:qedfindproxystrings} is \cref{lem:qed_analyze} that we state here, but whose proof we defer to the next (sub)section.

\begin{restatable*}{lemmaq}{qedanalyze}\label{lem:qed_analyze}
   There is a quantum algorithm that, given a pattern $P$ of length $m$, a text $T$ of length $n$, and an integer threshold $k > 0$ such that $n < 3/2 \cdot m$, computes a set $C$ of \emph{candidate positions} such that $\OccE_k(P,T) \subseteq C \subseteq \fragment{0}{n}$, and outputs one of two compressed representations of $C$:
    \begin{itemize}
        \item either it outputs $\floor{C/k}$ of size $|\floor{C/k}|=\Ohtilde(k)$, \textbf{or}
        \item it outputs $q \in \Z_{>0}$ and $I \subseteq \fragmentco{0}{q}$ of size $|I| = \Oh(k)$ such that $C = \{ t \in \fragment{0}{n-m-k} : t \ \mathrm{mod} \ q \in I\}$
        and $\OccE_k(P,T) \subseteq C \subseteq \OccE_{44k}(P,T)$.
    \end{itemize}
    The algorithm takes $\Ohtilde(\sqrt{km})$ query time and $\Ohtilde(\sqrt{km}+k^2)$ quantum time. \qedhere
\end{restatable*}

At a high level, our goal is to use the candidate set from \cref{lem:qed_analyze} to construct a set \( S \) of alignments of \( P \) onto fragments of \( T \) that not only captures all \( k \)-error alignments but also avoids unnecessarily costly alignments. If the candidate set \( C \) is already in the second compressed form from \cref{lem:qed_analyze}, we know that \( C \) inherently avoids such costly alignments.

However, if \( |C| = \Ohtilde(k) \), we aim to apply the quantum \gaped algorithm to filter out candidate positions where costly \( k \)-error occurrences arise. A challenge we face is that the \gaped algorithm requires two input strings of the same length, while the candidate set only provides the starting positions of the \( k \)-error occurrences.

The following corollary of \cref{lem:quantum_gap_ed} helps to handle this mismatch in formulation.

\begin{corollary}[of \cref{lem:quantum_gap_ed}]\label{cor:quantum_gap_ed}
    For all positive integers $k$, $m$, and $n$, there exists an integer $\ell=m^{1+o(1)}$ and a function $\phi : \{0,1\}^\ell \times \Sigma^m \times \Sigma^n \times \fragment{0}{n}\to \{\text{\yes}, \text{\no}\}$ such that:
    \begin{itemize}
        \item if $\mathrm{seed}\in {\{0, 1\}}^{\ell}$ is sampled uniformly at random, then, with high probability,
        $\OccE_k(P,T)\subseteq \{t\in \fragmentco{0}{n} : \phi(\mathrm{seed},P,T,t)=\text{\yes}\}\subseteq \OccE_{K}(P,T)$ holds for some $K=k\cdot m^{o(1)}$;
       \item there exists a quantum algorithm that computes $\phi(\mathrm{seed},P,T,t)$
       in $\Oh(m^{1 + o(1)})$ time using $\Oh(m^{0.5 + o(1)})$ queries to $P$, $T$, and $\mathrm{seed}$.
    \end{itemize}
\end{corollary}
\begin{proof}
Denote $T_\$ \coloneqq T\cdot \$^{m-1}$ and $\Sigma_\$=\Sigma\cup\{\$\}$, where $\$\notin \Sigma$ is a unique special character.
Let $f : \{0,1\}^\ell\times \Sigma_\$^m\times \Sigma_\$^m$ be the function of \cref{lem:quantum_gap_ed} for length $n$ and threshold $2k$.
We set $\phi(\mathrm{seed},P,T,t)=f(\mathrm{seed},P,T_\$\fragmentco{t}{t+m})$.
By \cref{lem:quantum_gap_ed}, this function can be evaluated in $\Oh(m^{1 + o(1)})$ time
using $\Oh(m^{0.5 + o(1)})$ queries to $P$, $T$, and $\mathrm{seed}$.

If $t\in \OccE_k(P,T)$, then there is $t'\in \fragment{t}{n}$ such that $\ed(P, T\fragmentco{t}{t'})\le k$,
and thus $\ed(P, T_\$\fragmentco{t}{t'})\le k$.
Moreover, $\ed(P, T_\$\fragmentco{t}{t+m}) \le \ed(P, T_\$\fragmentco{t}{t'})+|t+m-t'| = \ed(P, T\fragmentco{t}{t'})+\big||P|-|T\fragmentco{t}{t'}|\big| \le 2\ed(P, T\fragmentco{t}{t'}) \le 2k$.
Hence, $\phi(\mathrm{seed},P,T,t)=f(\mathrm{seed},P,T_\$\fragmentco{t}{t+m})=\text{\yes}$ holds w.h.p.

Finally, note that $\phi(\mathrm{seed},P,T,t)=\text{\no}$ holds with high probability unless $\ed(P,T_\$\fragmentco{t}{t+m})\le (2k)\cdot m^{o(1)}$.
If $t+m \le n$, the latter immediately implies $\ed(P,T\fragmentco{t}{t+m})\le (2k)\cdot m^{o(1)}$.
Otherwise, $\ed(P,T_\$\fragmentco{t}{t+m})=\ed(P\fragmentco{0}{p},T\fragmentco{t}{n})+\ed(P\fragmentco{p}{m},\$^{t+m-n})$ holds for some $p\in \fragment{0}{m}$, and thus $\ed(P,T\fragmentco{t}{n}) \le \ed(P\fragmentco{0}{p},T\fragmentco{t}{n})+ \ed(P\fragmentco{p}{m},\varepsilon) =   \ed(P\fragmentco{0}{p},T\fragmentco{t}{n}) + m-p \le  \ed(P\fragmentco{0}{p},T\fragmentco{t}{n})+\ed(P\fragmentco{p}{m},\$^{t+m-n})=\ed(P,T_\$\fragmentco{t}{t+m}) \le (2k)\cdot m^{o(1)}$.
In either case, $t\in \OccE_{K}(P,T)$ for $K = (2k)\cdot m^{o(1)}$.
\end{proof}

Now, we present an oracle that, given a set \( F \subseteq \fragmentco{0}{n} \), either identifies a position \( t \in F \cap \OccE_{K}(P, T) \), where \( K = k \cdot m^{o(1)} \), or confirms that \( F \cap \OccE_k(P, T) = \emptyset \).
Even though \cref{lem:verifier} is formulated for general $F$, we  use \( F \) as the starting points of all \( K \)-error occurrences that the set \( S \) already captures. Then, the oracle helps us in either finding an uncaptured \( K \)-error occurrence or confirming that all \( k \)-error occurrences are already captured.

\begin{lemma}\label{lem:verifier}
There exists a quantum algorithm that, given a pattern $P$ of length $m$, a text $T$ of length $n\le 3/2\cdot m$, an integer $k>0$, and a set $F\subseteq \fragment{0}{n}$, outputs one of the following:
\begin{itemize}
    \item a position $t\in F\cap \OccE_{K}(P,T)$, where $K=k\cdot m^{o(1)}$; or
    \item $\bot$, indicating that $F\cap \OccE_k(P,T)=\emptyset$.
\end{itemize}
The algorithm uses $\Ohhat(\sqrt{k}m+k^2)$ time and $\Ohhat(\sqrt{km})$ queries to $P$, $T$, and the characteristic function of $F$.

Within the same complexity bounds, we can guarantee that the reported position satisfies $t\le \min(F\cap \OccE_k(P,T))$ or, alternatively, that it satisfies $t \ge \max(F\cap \OccE_k(P,T))$.\footnote{We assume $\min \emptyset = +\infty$ and $\max \emptyset = -\infty$.}
\end{lemma}
\begin{proof}
First, we apply \cref{lem:qed_analyze}, which returns $\floor{C/k}$ for a set $\OccE_k(P,T)\subseteq C \subseteq \fragment{0}{n}$ of candidate positions.
The two cases of \cref{lem:qed_analyze} govern further behavior of the algorithm.

If $|\floor{C/k}|=\Ohtilde(k)$, we use \cref{cor:quantum_gap_ed} for threshold $2k$ and a fixed $\mathrm{seed}\in \{0,1\}^\ell$ chosen uniformly at random.
We apply \GS to find $c\in \floor{C/k}$ such that
$\Phi(\mathrm{seed},P,T,ck)=\text{\yes}$ and $\fragmentco{ck}{ck+k}\cap F \ne \emptyset$
(using \GS again to test the latter condition).
If the search is successful, the algorithm reports an element of $\fragmentco{ck}{ck+k}\cap F$; otherwise, it outputs $\bot$.

If $t\in F\cap \OccE_k(P,T)$, then $\floor{t/k}\in \floor{C/k}$, $k\cdot \floor{t/k}\in \OccE_{2k}(P,T)$, and $\fragmentco{\floor{t/k}k}{\floor{t/k}k+k}\cap F \ne \emptyset$, so the search is successful with high probability.
Moreover, if both the outer and the inner \GS return the smallest (the largest) witnesses, then the reported value is guaranteed to be at most (at least, respectively) $t$.
For the converse implication, note that the reported position $t\in \fragmentco{ck}{ck+k}\cap F$ satisfies $t\in \OccE_{K}(P,T)$ for $K=K'+k$, where $K'=(2k)\cdot m^{o(1)}$ is the higher threshold of \cref{cor:quantum_gap_ed} such that $ck\in \OccE_{K'}(P,T)$ holds with high probability.

As for the complexity, observe that a single evaluation of the oracle for the outer \GS takes $\Ohhat(\sqrt{m}+\sqrt{k})$ queries and $\Ohhat(m+\sqrt{k})$ time.
Since the search space is of size $|\floor{C/k}|=\Ohtilde(\min(k,{m}/{k}))$, the total query and time complexity is $\Ohhat(\sqrt{km})$ and $\Ohhat(m\sqrt{k})$, respectively.

In the remaining case, we have $C\subseteq \OccE_{44k}(P,T)$.
We apply \GS to find $x\in \fragment{0}{n-m-k}$ such that $x \in F$ and $x \ \mathrm{mod} \ q \in I$.
If the search is successful, the algorithm reports such $x$; otherwise, it outputs $\bot$.

If $t\in F\cap \OccE_k(P,T)$, then from $C \subseteq \OccE_k(P,T)$, follows that the search is successful.
For the converse implication, note that for the reported $t$ always $t \in C\subseteq \OccE_{44k}(P,T)$ holds.

Since the search space is of size $\Oh(m)$, the total query and time complexity is $\Ohhat(\sqrt{m})$.
\end{proof}

With all components in place, we now proceed to prove \cref{thm:qedfindproxystrings}.

\qedfindproxystrings

\begin{proof}
    For this proof, let \( K \coloneqq k\cdot m^{o(1)} \) be defined as in \cref{lem:verifier} for $m$, $n$, and $k$.
    Additionally, set \( K' \coloneqq 2K + k \).

    We first show how to prove the lemma in the case we have two alignments that enclose $T$.

     \begin{claim}\label{claim:qedfindproxystringscropped}
        Suppose we are given $\sE_{P,T}(\mX_{\pref})$ and $\sE_{P,T}(\mX_{\suf})$ for alignments $\mX_{\pref}$ and $\mX_{\suf}$ of cost at most $K'$ such that $S \coloneqq \{\mX_{\pref}, \mX_{\suf}\}$ encloses $T$.

        Then, we can compute strings $P^\#,T^\#$ for which the two following hold:
        \begin{enumerate}[(a)]
    \item For every $a\in \fragmentco{0}{m}$ and $b \in \fragmentco{0}{n}$, we have
        that $P^\#\position{a} = T^\#\position{b}$ implies  $P\position{a} = T\position{b}$.
        \label{claim:qhdfindproxystringscropped:it:a}
        \item for all optimal alignments $\mX : P \onto T\fragmentco{t}{t'}$ of cost at most $k$, we have $\ed(P, T\fragmentco{t}{t'}) = \ed(P^\#, T^\#\fragmentco{t}{t'})$ and $\sE_{P, T}(\mX) = \sE_{P^\#, T^\#}(\mX)$.
        \label{claim:qhdfindproxystringscropped:it:b}
\end{enumerate}
        We use $\Ohhat(\sqrt{km})$ queries and $\Ohhat(\sqrt{k}m+k^2)$ quantum time.
    \end{claim}

    \begin{claimproof}
        Consider the following procedure.
        \begin{enumerate}[(i)]
        \item Start with $S = \{\mX_{\pref}, \mX_{\suf}\}$ containing alignments of cost at most $K'$. Iterate through the following steps, adding new elements to $S$.
        Instead of storing each $\mX \in S$ explicitly, store only $\sE_{P,T}(\mX)$.
        Continue this process until the loop is not interrupted.
        \label{alg:qedfindproxystringscropped:i}
        \begin{enumerate}[(1)]
            \item Check whether $\bc(\bG_S) = 0$ via \cref{prop:construct_g_s}.
            If $\bc(\bG_S) = 0$, interrupt the loop.
             \label{alg:qedfindproxystringscropped:1}
            \item Otherwise, if $\bc(\bG_S) > 0$, initialize the data structure of
                \cref{prop:capture_ds}.
                Let $w$ be total weight of $\w_S$ from \cref{lem:exfuncover}, that is, the
                sum of the costs of all alignments in $S$.
                Apply \cref{lem:verifier} for $F = \{t \in \fragment{0}{n} : \forall_{i\in
                \fragment{0}{n_0-m_0}}\, |\tau_i^0-t-\pi_0^0|> w+3K'\}$, using the data
                structure to test in $\Oh(1)$ time whether a given position belongs to
                $F$.
                If \cref{lem:verifier} returns $\bot$, interrupt the loop.
                \label{alg:qedfindproxystringscropped:2}
            \item If \cref{lem:verifier} returns $t\in F$, apply
                \cref{prp:quantumed_w_info} to retrieve $\sE_{P,T}(\mY)$ for an optimal
                alignment $\mY : P \onto T\fragmentco{t}{\min(t + m, n)}$ and $\mY$ to
                $S$.
            \label{alg:qedfindproxystringscropped:3}
        \end{enumerate}
        \item If the loop was interrupted in \eqref{alg:qedfindproxystringscropped:1},
            reconstruct $P,T$ via \cref{prp:reconstruct_pt} and return $P,T$.
        Otherwise, if the loop was interrupted in
        \eqref{alg:qedfindproxystringscropped:2}, construct $P^\#,T^\#$ via
        \cref{prp:reconstruct_pt_hash} and return $P^\#,T^\#$.
    \end{enumerate}

    Each time we add \( \mY \) to \( S \), the corresponding position $t$ satisfies \( |\tau_{i}^{0} - t - \pi_0^{0}| > w + 3K' \) for all \( i \in \fragment{0}{n_0-m_0} \) and $\ed(P,T\fragmentco{t}{t'})\le K$ for some $t'\in \fragment{t}{n}$.
    Since $|\min(t+m,n)-t'| \le |t+m-t'|\le  \ed(P,T\fragmentco{t}{t'}) \le K$, we conclude that $\ed(P,T\fragmentco{t}{\min(t+m,n)})\le 2K$, and thus the cost of $\mY$ is at most $2K \le K'$.
    Consequently, adding $\mY$ to $S$ preserves the invariant and, by \cref{lem:periodhalves}, halves the number of black components, that is, \( \bc(S \cup \{\mY\}) \leq \bc(S)/2 \).

    If the loop is interrupted in \eqref{alg:qedfindproxystringscropped:2}, then there is no $k$-error occurrence $T\fragmentco{t}{t'}$ of $P$ such that $t\in F$, and thus $S$ captures all $k$-error occurrences of $P$ in $T$, so the correctness of the procedure follows from \cref{prp:reconstruct_pt_hash}.
    If the loop is interrupted in \eqref{alg:qedfindproxystringscropped:1}, then the correctness follows from \cref{prp:reconstruct_pt}.

    Regarding computational complexity, observe that the main loop iterates $\Oh(\log n)$ times because the number of black components is initially $\Oh(n+m)$, it halves upon the successful completion of every iteration, and the algorithm exits the loop as soon as the number of black components reaches zero.
    For each iteration, the runtime is dominated by the application of \cref{lem:verifier}, which takes $\Ohhat(\sqrt{km})$ queries and $\Ohhat(m\sqrt{k}+k^2)$ time.
    Subsequent applications of \cref{prp:reconstruct_pt,prp:reconstruct_pt_hash} require $\Ohtilde(\sqrt{K'|S|m})=\Ohhat(\sqrt{km})$ queries and $\Ohtilde(m|S|)=\Ohtilde(m)$ time.
    \end{claimproof}

    The following routine reconnects the procedure from \cref{claim:qedfindproxystringscropped} to the general case.
    \begin{enumerate}[(i)]
        \item If \( 4K' > m \), read the entire strings \( P \) and \( T \), and return them directly.
        \label{alg:qedfindproxystrings:i}
        \item Apply \cref{lem:verifier} for $F=\fragment{0}{n}$, which either outputs $\bot$, indicating that $\OccE_k(P,T)=\emptyset$,
        or returns positions $t_{\pref},t_{\suf}\in \OccE_{K}(P,T)$ such that $t_\pref \le \min \OccE_k(P,T)$ and $t_\suf \ge \max \OccE_k(P,T)$.
        \label{alg:qedfindproxystrings:iii}
        \item If \cref{lem:verifier} outputs $\bot$, return two strings \( P^\# \) and \( T^\# \) of the original length, each with a unique sentinel character at every position. Otherwise, use \cref{prp:quantumed_w_info} to collect \( \sE_{P,T}(\mX_{\pref}) \) and \( \sE_{P,T}(\mX_{\suf}) \) for some optimal alignments \( \mX_{\pref} : P \onto  T\fragmentco{t_{\pref}}{\min(t_{\pref} + m, n)} \) and \( \mX_{\suf} : P \onto T\fragmentco{t_{\suf}}{t_{\mathrm{end}}} \), where $t_{\mathrm{end}} \coloneqq \min(t_{\suf} + m+k, n)$.
        \label{alg:qedfindproxystrings:iv}
        \item Apply \cref{claim:qedfindproxystringscropped} on the pattern \( P \), the cropped text \( T' = T\fragmentco{t_{\pref}}{t_{\mathrm{end}}} \), the edit information \( \sE_{P,T}(\mX_{\pref}) \) and \( \sE_{P,T}(\mX_{\suf}) \) with indices shifted so that they are defined on \( T' \).
        \label{alg:qedfindproxystrings:v}
        \item Finally, modify the returned \( P^\# \) and \( T^\# \) by prepending \( t_{\pref} \) unique sentinels to \( T^\# \) and appending \( n - t_{\mathrm{end}} \) unique sentinels to \( T^\# \). Return \( P^\# \) and \( T^\# \).
        \label{alg:qedfindproxystrings:vi}
    \end{enumerate}

   Next, we argue the correctness of the routine.

    If the routine returns at \eqref{alg:qedfindproxystrings:i}, then \( P \) and \( T \) trivially satisfy the claimed properties. Additionally, since \( 4K' > m \) implies \( k = \hat{\Omega}(K) = \hat{\Omega}(m) \), the \(\Oh(m)\) queries needed to read the strings are within the claimed query complexity of \(\Ohhat(\sqrt{km})\).

    On the other hand, if the routine returns at \eqref{alg:qedfindproxystrings:iv}, then \( \OccE_k(P,T) = \emptyset \).
    In this case, the returned strings \( P^\# \) and \( T^\# \) satisfy \( \ed(P, T\fragmentco{t}{t'}) \leq \ed(P^\#, T^\#\fragmentco{t}{t'}) \) for all $t,t'$, thereby fulfilling the condition \( \OccE_k(P^\#,T^\#) = \emptyset\) as desired.

    Finally, let us consider the case where we neither return at \eqref{alg:qedfindproxystrings:i} nor at \eqref{alg:qedfindproxystrings:iii}.
    Since $t_{\pref},t_\suf \in \OccE_{K}(P,T)$ and the underlying $K$-error occurrences have lengths in $\fragment{m-K}{m+K}$, the costs of $\mX_{\pref}$ and $\mX_\suf$ are at most $2K<K'$ and $2K+k=K'$, respectively.
    Furthermore, $4K' \le m$ implies $|T'|\le n \le 3/2 \cdot m \le 2m-2K'$, so the shifted alignments $\mX_{\pref}'$ and $\mX_{\suf}'$, which are defined on $T'$ instead of $T$ and are used as input for \cref{claim:qedfindproxystringscropped}, enclose $T'$.
    Therefore, the assumptions of \cref{claim:qedfindproxystringscropped} are satisfied.

    Let $\mX : P \onto T\fragmentco{t}{t'}$ be an alignment with cost $k' \leq k$. Our goal is to prove the first part of the statement, speciffically that $\mX : P^\# \onto T^\#\fragmentco{t}{t'}$ also has a cost of $k' \leq k$.
    By construction, $t\in \OccE_k(P,T)$, so $t_\pref \le t \le t_\suf$ and $t' \le \min(n, t+m+k) \le \min(n,t_\suf+m+k)=t_{\mathrm{end}}$.
    This implies the existence of a corresponding alignment $\mX' : P \onto T'\fragmentco{t-t_{\pref}}{t'-t_{\pref}}$ with shifted indices and the same cost $k'$. According to \cref{claim:qedfindproxystringscropped}\eqref{claim:qhdfindproxystringscropped:it:b}, this alignment is preserved in the resulting $P^\#$ and $T^\#$. Prepending sentinel characters to $T^\#$ then ensures that $\mX : P^\# \onto T^\#\fragmentco{t}{t'}$ keeps cost $k'$.

    For the reverse direction, observe that $P^\#$ and $T^\#$, even after the addition of sentinel characters, satisfy \cref{claim:qedfindproxystringscropped}\eqref{claim:qhdfindproxystringscropped:it:a}. Therefore, for any alignment $\mX : P \onto T\fragmentco{t}{t'}$, we have $\edal{\mX}(P, T\fragmentco{t}{t'}) \leq \edal{\mX}(P^\#, T^\#\fragmentco{t}{t'})$. This proves the other direction of the statement.
\end{proof}

\subsection{Proof of Lemma~\ref{lem:qed_analyze} via Structural Insights of
    \texorpdfstring{$k$}{k}-Error
Occurrences}
\label{sec:qed_intro_analyze}

In this (sub)section we prove \cref{lem:qed_analyze}.

\qedanalyze

We briefly outline the proof. The first key to this is the structural insights from the pattern analysis in \cite{CKW20}, which distinguish between three possible structures that $P$ may exhibit.

\EI

In \cref{sec:decomp} we devise a quantum algorithm
computing this structural decomposition.

\begin{restatable*}{lemma}{lemanalyze}\label{lem:analyzeP}
    Given a pattern $P$ of length $m$, we can find a structural decomposition as described
    in \cref{prp:EI} in $\tilde{\mathcal{O}}(\sqrt{km})$ query time and
    $\tilde{\mathcal{O}}(\sqrt{km} + k^2)$ quantum time.
\end{restatable*}

\Cref{lem:analyzeP}, combined with the next three lemmas, is sufficient to prove \cref{lem:qed_analyze}.

\begin{restatable*}{lemma}{edbreakcase}\label{lem:ed_breakcase}
    Let $P$ denote a pattern of length $m$, let $T$ denote a text of length $n \leq 3/2 \cdot m$, and let $k > 0$ denote an integer threshold. Suppose for $P$ case \eqref{item:a:prp:EI} of \cref{prp:EI} applies, that is, $P$ contains $2k$ disjoint \emph{breaks} $B_1, \ldots, B_{2k}$ each
    having period $\per(B_i) > m/128k$ and length $|B_i| = \lfloor m/8k \rfloor$.

    Then, there exists a quantum algorithm that, in $\Ohtilde(\sqrt{km})$ time, outputs $\floor{C/k}$, where $C$ is a set of     \emph{candidate positions} such that $\OccE_k(P,T) \subseteq C \subseteq \fragment{0}{n}$   and $|\floor{C/k}|=\Ohtilde(k)$.
\end{restatable*}

\begin{restatable*}{lemma}{edrepregion}\label{lem:ed_repregion}
    Let $P$ denote a pattern of length $m$, let $T$ denote a text of length $n \leq 3/2 \cdot m$, and let $k > 0$ denote an integer threshold. Suppose for $P$ case \eqref{item:b:prp:EI} of \cref{prp:EI} applies, that is, $P$ contains disjoint \emph{repetitive regions} $R_1,\ldots, R_{r}$ of total length $\sum_{i=1}^r |R_i| \ge \deltavN/\deltavD \cdot m$ such that each region $R_i$ satisfies $|R_i| \ge m/\betav k$ and has a primitive \emph{approximate period} $Q_i$ with $|Q_i| \le m/\alphav k$ and $\edl{R_i}{Q_i} = \ceil{\betav k/m\cdot |R_i|}$.

    Then, there exists a quantum algorithm that, using $\Ohtilde(\sqrt{km})$ queries and $\Ohtilde(\sqrt{km}+k^2)$ time, outputs $\floor{C/k}$, where $C$ is a set of  \emph{candidate positions} such that $\OccE_k(P,T) \subseteq C \subseteq \fragment{0}{n}$   and $|\floor{C/k}|=\Ohtilde(k)$.
\end{restatable*}

\begin{restatable*}{lemma}{edapproxpercase}\label{lem:ed_approxpercase}
    Let $P$ denote a pattern of length $m$, let $T$ denote a text of length $n$, and let $k > 0$ denote an integer threshold such that $n < 3/2 \cdot m + k$. Suppose for $P$ case \eqref{item:c:prp:EI} of \cref{prp:EI} applies, that is, $P$ has a primitive \emph{approximate period} $Q$ with $|Q|\le m/\alphav k$ and $\edl{P}{Q} < \betav k$.

    There is a quantum algorithm that computes a set $C$ of \emph{candidate positions} such that $\OccE_k(P,T) \subseteq C \subseteq \fragment{0}{n}$, and outputs one of two compressed representations of $C$:
    \begin{itemize}
        \item either it outputs $\floor{C/k}$ of size $|\floor{C/k}|=\Ohtilde(k)$, \textbf{or}
        \item it outputs $q \in \Z_{>0}$ and $I \subseteq \fragmentco{0}{q}$ of size $|I| = \Oh(k)$ such that $C = \{ t \in \fragment{0}{n-m-k} : t \ \mathrm{mod} \ q \in I\}$
        and $\OccE_k(P,T) \subseteq C \subseteq \OccE_{44k}(P,T)$.
    \end{itemize}
    The algorithm takes $\Ohtilde(\sqrt{km})$ query time and $\Ohtilde(\sqrt{km}+k^2)$ quantum time.
\end{restatable*}

\begin{proof}[Proof of \cref{lem:qed_analyze}]
    Use \cref{lem:analyzeP} to determine which case of \cref{prp:EI} applies, and based on that outcome, use the appropriate lemma among \cref{lem:ed_breakcase}, \cref{lem:ed_repregion}, or \cref{lem:ed_approxpercase}.
\end{proof}

\subsubsection{Structural Decomposition of the Pattern}
\label{sec:decomp}

This (sub)section aims to prove the following lemma.

\lemanalyze*

Charalampopoulos, Kociumaka, and Wellnitz~\cite{CKW20} introduced a high-level algorithm (see \cref{alg:E1}) that computes such a structural decomposition and can be readily adapted to various computational settings.
This algorithm is also easily adaptable to the quantum case.
It is worth noting that \cref{alg:E1} provides a constructive proof of \cref{prp:EI} on its own.
For a more detailed proof, we recommend referring directly to~\cite{CKW20}.

\begin{algorithm}[t]
    $\mathcal{B} \gets \{\}; \mathcal{R} \gets \{\}$\;
    \While{\bf true}{
        Consider fragment $P' = P\fragmentco{j}{j+\floor{m/\betav k}}$
        of~the next $\floor{m/\betav k}$ unprocessed characters of~$P$\;
        \If{$\per(P') > m/\alphav k$}{
            $\mathcal{B} \gets \mathcal{B} \cup \{P'\}$\;
            \lIf{$|\mathcal{B}| = 2k$}{\Return{breaks $\mathcal{B}$}}
            }\Else{
            $Q \gets P\fragmentco{j}{j+\per(P')}$\;
            Search for prefix $R$ of~$P\fragmentco{j}{m}$
            with $\edl{R}{Q} =
            \lceil\betav k/m\cdot |R|\rceil$ and $|R|>|P'|$\;\label{ln:efwd}
            \If{such $R$ exists}{
                $\mathcal{R} \gets \mathcal{R} \cup \{(R, Q)\}$\;
                \If{$\sum_{(R, Q) \in \mathcal{R}} |R|\ge \deltavN/\deltavD \cdot
                    m$}{
                    \Return{repetitive regions (and their corresponding
                    periods) $\mathcal{R}$}\;
                }
                }\Else{
                Search for suffix $R'$ of~$P$
                with $\edl{R'}{Q} = \lceil\betav
                k/m\cdot |R'|\rceil$ and $|R'|\ge m-j$\;\label{ln:ebcw}
                \lIf{such $R'$ exists}{%
                    \Return{repetitive region $(R',Q)$}%
                    }\lElse{%
                    \Return{approximate period $Q$}%
                }
            }
        }
    }
    \caption{A constructive proof~of~\cref{prp:EI} \cite[Algorithm 9]{CKW20}.}\label{alg:E1}
\end{algorithm}

\Cref{alg:E1} maintains an index $j$ indicating
the position of the string that has been processed
and returns as soon as one of the structural properties
described in \cref{prp:EI} is found.
At each step, the algorithm considers the fragment $P' = P\fragmentco{j}{+\lfloor m/8k \rfloor}$.
If $\per(P') \geq m/128k$, $P'$ is added to $\mathcal{B}$.
If not, the algorithm attempts to extend $P'$ to a repetitive region.
Let $Q$ be the string period of $P'$.
The extension involves searching for a prefix $R$ of $P\fragmentco{j}{m}$
such that $\edl{R}{Q} = \lceil\betav k/m\cdot |R|\rceil$.
If such a prefix is found, $R$ is added to $\mathcal{R}$.
If not, it indicates that there were not enough mismatches between $Q$ and $P\fragmentco{j}{m}$.
At this point, the algorithm attempts to extend $P\fragmentco{j}{m}$ backward to a repetitive region.
If successful, a sufficiently long repetitive region has been found to return.
Otherwise, it suggests that $\edl{P}{Q}$ is small,
implying that $P$ has approximate period $Q$.

In adapting \cref{alg:E1} to the quantum setting,
the primary focus lies on the procedure for identifying
a suitable prefix $R$ of $P\fragmentco{j}{m}$ such that $\edl{R}{Q} = \lceil\betav k/m\cdot |R|\rceil$.
The first step towards accomplishing this consists into,
given a prefix $R$ of $P\fragmentco{j}{m}$,
to determine whether $\edl{R}{Q}$ is less than,
equal to, or greater than $\lceil\betav k/m\cdot |R|\rceil$.
Note that we already know that for any prefix $R$ of $P\fragmentco{j}{m}$ such that $|R|
\leq \floor{m/\betav k}$, we have $\edl{R}{Q} = 0$.
Therefore, we only consider the case where $|R| > \floor{m/\betav k}$.
This allows us to later combine an exponential search/binary search to verify the existence
of a prefix with the desired property.

For the sake of conciseness, we introduce the following definition.

\begin{definition}
    Let $P$ denote a pattern of length $m$, and let $Q$ denote a primitive string of
    length $|Q| \leq m/\alphav k$ for some positive threshold $k$.
    For $j' \in \fragmentoc{j+\floor{m/\betav k}}{m}$, define
    $\Delta(j') = \edl{P\fragmentco{j}{j'}}{Q} - \lceil\betav k/m\cdot |j'-j'|\rceil$.
\end{definition}

First, notice that via \cref{lem:find_min_approx_per} we can verify
whether $\Delta(j') < 0$, $\Delta(j') = 0$ or $\Delta(j') > 0$ holds for a given $j' \in \fragmentoc{j+\floor{m/\betav k}}{m}$.
Indeed, it suffices to apply \cref{lem:find_min_approx_per}
on $P\fragmentco{j}{j'}$, $Q$ and $\lceil\betav k/m\cdot |j'-j|\rceil$.
Note that the assumptions for it is satisfied, as from $|Q| \leq m/128k$ follows that $|j'-j| \geq (2 \lceil\betav k/m\cdot |j'-j|\rceil + 1)|Q|$ holds.
Each call to compute the sign of $\Delta(j')$ requires $\tilde{\mathcal{O}}(|j' - j| \cdot \sqrt{k/m})$ query time and
$\tilde{\mathcal{O}}(|j' - j|\cdot\sqrt{k/m} + |j' - j|^2 \cdot k^2/m^2)$ quantum time.

Next, we show how this helps us in finding a prefix $R$ of $P\fragmentco{j}{m}$
such that $\edl{R}{Q} = \lceil\betav k/m\cdot |R|\rceil$, that is, in finding
$j' \in \fragmentoc{j}{m}$ such that $\Delta(j') = 0$.

\begin{lemma}
    \label{prp:findprefix}
    Let $P$ denote a pattern of length $m$, and let $Q$ denote a primitive string of
    length $|Q| \leq m/\alphav k$ for some threshold $k$.

    Then, there is a quantum procedure that either
    \begin{enumerate}
        \item finds a prefix $R$ of $P\fragmentco{j}{m}$ such that $\edl{R}{Q} =
            \lceil\betav k/m\cdot |R|\rceil$ in $\tilde{\mathcal{O}}(|R| \cdot
            \sqrt{k/m})$ query time and $\tilde{\mathcal{O}}(|R|\cdot\sqrt{k/m} + |R|^2
            \cdot k^2/m^2)$ quantum time; or
        \item verifies that $\edl{P\fragmentco{j}{m}}{Q} \leq \lceil\betav k/m\cdot
            |R|\rceil$ in $\tilde{\mathcal{O}}(\sqrt{km})$ query time and
            $\tilde{\mathcal{O}}(\sqrt{km} + k^2)$ quantum time.
            Moreover, it marks the search for $R$ as unsuccessful.
    \end{enumerate}
\end{lemma}
\begin{proof}
    The algorithm proceeds as follows:
    \begin{enumerate}[(i)]
        \item Via \cref{lem:find_min_approx_per} verify whether there exists $i \in \fragmentco{0}{\ceil{\log(m-j)}}$ such that $\Delta(\min(2^{i},m-j)) \geq 0$;
        \item if such $i$ does not exist then mark the search of $R$ as not successful;
        \label{it:findprefix:ii}
        \item else if $\Delta(\min(2^{i},m-j)) = 0$, return the prefix $R = P\fragmentco{j}{\min(2^{i},m-j)}$; lastly
        \label{it:findprefix:iii}
        \item if neither of the two previous two cases holds, execute a binary search.
            The two indices $\ell,h$ are kept as a lower bound and an upper bound,
            and they are initialized to $\ell \coloneqq j + |Q|$ and $h \coloneqq \min(j + 2^i, m)$.
            In every iteration, set $\mathsf{mid}$ to be the middle between position $\ell$ and $h$,
            and computed the sign of $\Delta(\mathsf{mid})$.
            If $\Delta(\mathsf{mid}) = 0$, then return the prefix $R = P\fragmentco{j}{\mathsf{mid}}$.
            If $\Delta(\mathsf{mid}) > 0$ recurse on the right side.
            Lastly, if $\Delta(\mathsf{mid}) < 0$, recurse on the left side.
            \label{it:findprefix:iv}
    \end{enumerate}

    Clearly, the algorithm is correct when it returns in \eqref{it:findprefix:ii} or in \eqref{it:findprefix:iii}.
    It remains to show that if the algorithm proceeds to \eqref{it:findprefix:iv}, then it finds $j' \in \fragmentoc{j}{m}$ such that $\Delta(j') = 0$.
    It is easy to verify that the algorithm maintains the invariant $\Delta(\ell) < 0$ and $\Delta(h) > 0$.
    Indeed, from $Q = P\fragmentco{j}{j + |Q|}$ and from how we perform the exponential search follows
    that the invariant holds in the beginning.
    Moreover, from~\eqref{it:findprefix:iv} follows that we preserve the invariant
    at every step of the algorithm.
    For the sake of contradiction, assume the binary search has ended without finding such a prefix $R$,
    that is, $h = \ell+1$ and the invariant still holds.
    Let $R_{\ell} = P\fragmentco{j}{\ell}$ and $R_{h} = P\fragmentco{j}{h}$
    Observe that for $h = l+1$, we have
    \[
        \edl{R_l}{Q} = \edl{R_h}{Q} + c \text{ and }
        \lceil\betav k/m\cdot |R_l| \rceil = \lceil\betav k/m\cdot |R_h| \rceil + c'
        \text{ for } c,c' \in \{0,1\}.
    \]
    However, there are no values of $c,c'$ such that
    \[
        \edl{R_l}{Q} < \lceil\betav k/m\cdot |R_l|\rceil
        \text{ and }
        \edl{R_h}{Q} > \lceil\betav k/m\cdot |R_h|\rceil.
    \]
    Hence, the binary search terminates and finds a prefix $R$ of
    $P\fragmentco{j}{m}$ such that $\edl{R}{Q} = \lceil\betav k/m\cdot |R|\rceil$.

    Regarding the complexity analysis,
    notice that if the algorithm returns in
    \eqref{it:findprefix:ii} then the quantum and query time is dominated by
    computing the sign of $\Delta(j')$ with $|j' - j| = \Oh(m)$;
    otherwise, it is dominated by computing the sign of $\Delta(j')$ with $|j' - j| = \Oh(|R|)$.
\end{proof}

Finally, we show how to adapt \cref{alg:E1} to the quantum setting.

\lemanalyze

\begin{proof}
    We show how to implement \cref{alg:E1} in the claimed query and quantum time.
    There is no need to modify its pseudocode.
    We go through line by line and show how we can turn it into a quantum algorithm by using suitable quantum subroutines.

    First, observe that the while loop at line 2 requires at most $\mathcal{O}(k)$
    iterations, as we process at least $\lfloor m/8k \rfloor$ characters in each
    iteration.

    Since we can find the period of $P'$ using \cref{cor:quantum_per}
    in $\tilde{\mathcal{O}}(\sqrt{m/k})$ quantum time and query time,
    we can conclude that line 4 summed over all iterations runs in
    $\tilde{\mathcal{O}}(k\cdot \sqrt{m/k}) = \tilde{\mathcal{O}}(\sqrt{km})$ time.
    To check for the existence of the prefix $R$ at line 9 we use \cref{prp:findprefix}.
    Using $\sum_{(R, Q)} |R| \leq m$, we conclude that summed over all $R \in \mathcal{R}$ this results in
    $\tilde{\mathcal{O}}(\sqrt{km})$ query time and
    \[
        \sum_{(R, Q)} \tilde{\mathcal{O}}(|R| \cdot \sqrt{k/m} + |R|^2 \cdot k/m) =
        \tilde{\mathcal{O}}(\sqrt{k^2/m^2}) \sum_{(R, Q)} |R| +
        \tilde{\mathcal{O}}(k^2/m^2) \sum_{(R, Q)} |R|^2 \leq
        \tilde{\mathcal{O}}(\sqrt{km} + k^2) \text{ quantum time.}
    \]

    If we did not manage to find such a prefix $R$, we incur the same quantum and query
    complexity, and we return either at line 15 or at line 17.
    Before returning, we still have to determine if $R'$ exists.
    Similarly, by using again \cref{prp:findprefix}, we can argue that the search for $R'$
    requires no more than the claimed query and quantum time.
\end{proof}

\subsection{Finding a Candidate Set for the Periodic Case}
\label{sec:periodic_case}

\begin{lemma}
\label{lem:periodiccase_ed_progression}
Let $P$ denote a pattern of length $m$, let $T$ denote a text of length $n$,
and a positive threshold $k > 0$.
Suppose there is a positive integer $d \geq 2k$ and a primitive string $Q$ with $|Q| \leq m/8d$, $\ed(P, Q^\infty\fragmentco{\ell_P}{r_P}) \leq d$,
and $\ed(T, Q^\infty\fragmentco{\ell_T}{r_T}) \leq 4d$
for some $\ell_P \leq r_P$ and $\ell_T \leq r_T$.

Then, the set
\[
C \coloneqq \{x \in \fragment{0}{n-m+k} \mid \text{there exists $x' \in \Occ_{\infty}$ such that $|x-x'| \leq 6d$}\},
\]
where $\Occ_{\infty} = \{\ell_P- \ell_T + j \cdot |Q| : j \in \mathbb{Z}\}$, satisfies
\begin{enumerate}[(a)]
    \item $\OccH_k(P,T) \subseteq C$, and
    \label{lem:periodiccase:a}
    \item $C \subseteq \OccH_{22d}(P,T)$.
    \label{lem:periodiccase:b}
\end{enumerate}
\end{lemma}

\begin{proof}
    Let $Q_T = Q^{\infty}\fragmentco{\ell_T}{r_T}$, $Q_P = Q^{\infty}\fragmentco{\ell_P}{r_P}$, and let $\mX_P : P \onto Q_P$, $\mX_T : T \onto Q_T$ be optimal alignments.

    First, we prove~\eqref{lem:periodiccase:a}.
    Suppose $\mX : P \onto T\fragmentco{t}{t'}$ is an alignment of cost at most $k$.
    Note, since $\mX$ has cost at most $k$, we have $t \in \fragment{0}{n-m+k}$.

    We want to prove that there exists $x' \in \Occ_{\infty}$ such that $|x' - t| \leq 6d$.
    Consider the alignment  $\mX_T' \subseteq \mX_T$  such that $\mX_T' : T\fragmentco{t}{t'} \onto Q_T\fragmentco{q}{q'}$ for some $q,q' \in \fragment{0}{|Q_T|}$,
    and the alignment $\mY = \mX_P^{-1} \circ \mX \circ \mX_T'$ such that $\mY : Q_P \onto Q_T\fragmentco{q}{q'}$.
    From $1 \leq |Q| \leq m/8d$ and $|Q_P| \leq m - d$,
    follows that $Q_P$ contains
    \[
    \left\lfloor\frac{|Q_P|}{|Q|}\right\rfloor
    \geq \left\lfloor \frac{m-d}{|Q|}\right\rfloor
    \geq \left\lfloor 8d - \frac{8d^2}{m} \right\rfloor
    \geq \left\lfloor 8d - d \right\rfloor
    = 7d
    \]
    full occurrences of $\hat{Q} \coloneqq Q_P\fragmentco{0}{|Q|}$.
    Since $\mY$ has cost at most $d + k + 4d$, at least $\floor{|Q_P|/|Q|} - d - k - 4d \geq 7d - d - k - 4d \geq 1$ full occurrences of $\hat{Q}$ are matched without edits in $\mY$.

    Consider such an arbitrary occurrence that starts at position $q_P \in \fragment{0}{|Q_P|-|Q|}$ in $Q_P$ and is matched by $\mY$ to
    a occurrence of $\hat{Q}$ starting at position $q_T \in \fragment{q}{q'-|Q|}$ in $Q_T$.
    Further, let $\hat{r}$ and $\hat{t}$ be such that
    $(\hat{r}, q_R) \in \mX_R$, $(\hat{r}, \hat{t}) \in \mX$, $(\hat{t}, q_T) \in \mX_T'$.
    Note, for $\hat{r}$ and $\hat{t}$ we have that
    $|(q_T - q_P) - t| \leq d$,
    $|\hat{r} - (\hat{t} - t)| \leq k$,
    $|\hat{t} - q_T| \leq 4d$.
    By using the triangle inequality, we obtain
    \[
        |(q_T - q_P) - t|
        \leq |\hat{r} - q_P| + |\hat{r} - (\hat{t} - t)| + |\hat{t} - q_T|
        \leq 5d + k \leq 6d.
    \]
    To conclude the proof of~\eqref{lem:periodiccase:a}
    it suffices to argue that $q_T-q_P \in \Occ_{\infty}$ (therefore we can set $x' = q_T - q_P$).
    Since $Q$ is a primitive string, and thus also $\hat{Q}$,
    we have that $q_P$ and $q_T$ are such that
    $q_P + \ell_P \equiv_{|Q|} 0$ and  $q_T + \ell_T \equiv_{|Q|} 0$.
    Consequently, $q_P - q_T \equiv_{|Q|} \ell_T - \ell_P$ and $q_T-q_P \in \Occ_{\infty}$.

    We proceed to the proof of~\eqref{lem:periodiccase:b}.
    Consider an arbitrary $x \in \fragment{0}{n-m+k}$
    such that there exists $x' \in  \Occ_{\infty}$ with $|x - x'| \leq 6d$.
    Let $y = \min(x + m, n)$. We would like to argue that
    \begin{align} \label{eq:periodiccase:3}
        m - (y - x) \leq k.
    \end{align}
    If $y + m \leq n$, then~\eqref{eq:periodiccase:3} trivally holds because $x - y = m$.
    Otherwise, if $y + m > n$,
    then from $x \leq n-m+k$ directly follows $m - (y - x) = m - n + x \leq k$ and~\eqref{eq:periodiccase:3} holds.

    Next, consider $\mY_T \subseteq \mX_T$ such that $\mY_T : T\fragmentco{x}{y} \onto Q_T\fragmentco{\ell_T' - \ell_T}{r_T' - r_T}$ for some $\ell_T', r_T'$
    for which $Q^{\infty}\fragmentco{\ell_T'}{r_T'} =  Q_T\fragmentco{\ell_T' - \ell_T}{r_T' - r_T}$ holds.
    We want to argue that the optimal alignment $\mA : Q_T \onto Q_P$ has cost at most $17d$.
    This is sufficient in order to prove~\eqref{lem:periodiccase:b},
    because then $\mX_P \circ \mA^{-1} \circ \mY_T^{-1} : P \onto T\fragmentco{y}{y'}$ has cost at most $22d$.

    From $(x, \ell_T' - \ell_T) \in \mX_T$
    together with $|x - x'| \leq 6d$ follows
    $|x' - (\ell_T' - \ell_T)| \leq |y - (\ell_T' - \ell_T)| + |y - x'| \leq 4d + 7d \leq 11d$.
    Now, since $x' \in  \Occ_{\infty}$ there exists $j \in \mathbb{Z}$ such that $x' = \ell_P- \ell_T + j \cdot |Q|$.
    Consequently, for such $j$ it holds that
    \begin{equation} \label{eq:periodiccase:1}
        |\ell_P - \ell_T' + j \cdot |Q|| = |x' - (\ell_T' - \ell_T)| \leq 11d.
    \end{equation}
    Moreover, from $\ed(T\fragmentco{x}{y}, Q^{\infty}\fragmentco{\ell_T'}{r_T'}) \leq 4d$,
    $\ed(P, Q^{\infty}\fragmentco{\ell_P}{r_P}) \leq d$ and~\eqref{eq:periodiccase:3} follows
    \begin{equation} \label{eq:periodiccase:2}
        |(r_P - \ell_P) - (r_T' - \ell_T')|
        \leq |m - (r_P - \ell_P)| + |m - (y-x)| + |(y - x) - (r_T' - \ell_T')| \leq 4d+k+d \leq 6d.
    \end{equation}
    By combining~\eqref{eq:periodiccase:1} and~\eqref{eq:periodiccase:2},
    we conclude that the cost of $\mA$ is at most $17d$.
\end{proof}

\begin{lemma}[{\cite[Lemma 6.9]{CKW20}}]\label{lem:locked}
    Consider a string $S$, a primitive string $Q$, and an alignment $\mA : S \onto
    Q^{\infty}\fragmentco{\ell}{r}$ of optimal cost $d \coloneqq \edl{S}{Q}$.
    There is a (classical) algorithm that, given $\mA$ and $|Q|$, in $\Oh(d+1)$ time\footnote{The original proof in \cite{CKW20} starts by constructing $\mA$; here, we assume that $\mA$ is already provided.} constructs a set $\mathcal{L}$ of disjoint fragments of $S$ such that
    \begin{itemize}
        \item $\edl{L}{Q} > 0$ for every $L \in \mathcal{L}$,
        \item $\sum_{L\in \mathcal{L}} \edl{L}{Q} = d$, and
        \item $\sum_{L\in \mathcal{L}} |L| \le (5|Q|+1)d+2|Q|$.\lipicsEnd
    \end{itemize}
\end{lemma}

\begin{lemma}\label{lem:periodiccase_ed_small_c}
Let $P$ denote a pattern of length $m$, let $T$ denote a text of length $n$,
and a positive threshold $k > 0$.
Suppose there is a positive integer $d \geq 2k$ and a primitive string $Q$ with $|Q| \leq m/8d$, $\edl{P}{Q} = d$,
and $\edl{T}{Q} \leq 4d$.

Then, there exists a quantum algorithm that outputs $\floor{C/d}$, where $C$ is a set of     \emph{candidate positions} such that $\OccE_k(P,T) \subseteq C \subseteq \fragment{0}{n}$   and $|\floor{C/d}|=\Ohtilde(d)$
using $\Ohtilde(\sqrt{dm})$ queries and $\Ohtilde(\sqrt{dm}+d^2)$ time.
\end{lemma}

\begin{proof}
    Consider the following procedure.
    \begin{enumerate}[(i)]
        \item Via \cref{lem:find_min_approx_per} compute $\ell_T,r_T$ and $\ell_P,r_P$
        such that $Q_T \coloneqq Q^{\infty}\fragmentco{\ell_T}{r_T}$ and $Q_P \coloneqq Q^{\infty}\fragmentco{\ell_P}{r_P}$ are such to minimize $\edl{T}{Q}$ and $\edl{P}{Q}$, respectively.
        Moreover, via \cref{prp:quantumed_w_info} compute $\sE_{P, Q_P}(\mX_P)$ and $\sE_{T, Q_T}(\mX_T)$
        for some optimal alignments $\mX_P : P \onto Q_P$, $\mX_T : T \onto Q_T$.
        \label{it:periodiccase_ed_small_c:i}
        \item Apply \cref{lem:locked} to compute a family $\mathcal{L} = \{L_1, \ldots, L_f\}$ of disjoint fragments of $P$ such that $\sum_{i=1}^{f} \edl{L_i}{Q}=d$ and $\sum_{i=1}^{f} |L_i| \le (5|Q|+1)d+2|Q| \le 7|Q|d$.
        \label{it:periodiccase_ed_small_c:ii}
        \item Select $L=P\fragmentco{\ell}{r}\in \mathcal{L}$ at random with probability proportional to $\edl{L}{Q}$.
        \label{it:periodiccase_ed_small_c:iii}
        \item If $|L| > 28|Q| \cdot \edl{L}{Q}$, then report a failure.
        \label{eq:periodiccase_ed_small_c:iv}
        \item Return $\floor{C/d}$, where
        \[C \coloneqq \{x \in C' : \mX_T \text{ makes at
            least }1/8 \cdot \edl{L}{Q}\text{ edits in
    }T\fragmentco{\max\{0,x+\ell-k\}}{\min\{n,x+r+k\}}\}\] and $C'$ is the set from
        \cref{lem:periodiccase_ed_progression}.
        \label{it:periodiccase_ed_small_c:v}
    \end{enumerate}

    First, we give an upper bound on the size of the returned set.

    \begin{claim}\label{claim:periodiccase_ed_small_c:size}
        The candidate set $C$ satisfies $|\floor{C/d}|=\Oh(d)$.
    \end{claim}
    \begin{claimproof}
        We divide $\Z$ into \emph{blocks} of the form $\fragmentco{jd}{(j+1)d}$.
        We want to bound the number $b$ of blocks
        that contain some $x \in C$ such that $\mX_T$ makes at least $1/8\cdot \edl{L}{Q}$ edits in $T\fragmentco{\max\{0,x+\ell-k\}}{\min\{n,x+r+k\}}\}$.

        For this purpose, for each edit $(e,e') \in \mX_T$, let $b_e$ be the number
        of such blocks that contain some $x \in C$ such that
        $e \in \fragmentco{\max\{0,x+\ell-k\}}{\min\{n,x+r+k\}}\}$.
        Note that $b \leq 8 / \edl{L}{Q} \cdot \sum_e b_e$.
        Next, we perform a case distinction.

        If $|Q| < 12d$, we observe that for each edit $(e,e') \in \mX_T$, we have that
        $\fragmentco{e-r-k}{e-\ell+k}$ overlaps with at most
        $2 + (r - \ell + 2k)/d \leq 2+(|L| + 2k)/d \leq 3 + |L|/d$ blocks.
        As this latter condition is necessary for a block to be counted in $b_e$,
        we have $b_e \leq (3+|L|/d)$.
        Since $\mX_T$ contains at most $4d$ edits, we get
        \[
            b
            \leq 8/ \edl{L}{Q} \cdot 4d \cdot (3+|L|/d)
            \leq 8/ \edl{L}{Q} \cdot 4d \cdot (3 + 28|Q| \cdot \edl{L}{Q}/d)
            \leq 96d+896|Q| = \Oh(d).
        \]

        Otherwise, if $|Q| \geq 12d$,
        we use that $x$ that lies at distance at most $6d$
        from a position in $\Occ_{\infty}$,
        where this last set is defined as in \cref{lem:periodiccase_ed_progression}\eqref{lem:periodiccase:a}.
        That is,
        \[
            x \in I_z
            \coloneqq \fragment{c+z|Q|-6d}{c+z|Q|+6d}
        \]
        for some $z \in \Z$ and fixed $c\in \Z$. Let $o_z$ be the number of blocks that $I_z$ overlaps with. Observe that $o_z \leq 2 + (14d + 1)/d \leq 15$.
        Moreover, for each edit $(e,e') \in \mX_T$, we have that
        $\fragmentco{e-r-k}{t-\ell+k}$ overlaps with at most
        $2 + (r - \ell + 2k)/|Q| \leq 2+(|L| + 2k)/d \leq 3 + |L|/d$ intervals of the form $I_z$. For every such $I_z$ we have that this latter condition is necessary for each of the $o_z$ blocks to be counted in $b_e$. Consequently,
        \[
            b_e \leq 15 \cdot (2 + (d+|L|)/|Q|) = 30 + d/|Q| + |L|/|Q| \leq 30 + 12 + 28 = 70,
        \]
        which together with the bound on the number of edits in $\mX_T$,
        yields
        \[
            b \leq 8 / \edl{L}{Q} \cdot \textstyle{\sum_e} b_e \leq 8 / \edl{L}{Q} \cdot 4d \cdot 70 = \Oh(d).
            \claimqedhere
        \]
    \end{claimproof}

    Next, we observe that, by Markov's inequality,
    the probability that we report a failure in \eqref{eq:periodiccase_ed_small_c:iv} is
    \[
        \pr{|L| / \edl{L}{Q} > 28|Q|} \leq
        \frac{\sum_{i=1}^{f} |L_i| / \edl{L}{Q} \cdot \edl{L}{Q}/d}{28|Q|}
        = \frac{7|Q|}{28|Q|}
        = \frac{1}{4}.
    \]
    We also aim to give a bound on the probability that we fail
    to include an arbitrary position $t \in \OccE_k(P, T)$ to $C$.
    The following claim serves this purpose and holds without the assumption that we succeed in \eqref{eq:periodiccase_ed_small_c:iv}.

    \begin{claim}
        Fix an optimal $k$-edit alignment $\mX : P \onto T\fragmentco{t}{t'}$. With probability at most $4/7$ we have $t \notin C$.
    \end{claim}
    \begin{claimproof}
        Write $L_i = P\fragmentco{\ell_i}{r_i}$ for the $i$-th fragment in $\mathcal{L}$.
        Moreover, suppose $\mX$ aligns $L_i$ onto $L_i'$, that is $L_i' = \mX(L_i)$,
        and $\mX_T$ aligns $L_i'$ onto $Q_i \substr Q^\infty$, that is $Q_i = \mX_T(L_i')$.

        Suppose $\mX_T$ makes $k_i$ edits in
        $T\fragmentco{\max\{0,t+\ell_i-k\}}{\min\{n,t+r_i+k\}}$ and
        define $I \coloneqq \{i \in \fragmentco{i}{f} \mid k_i \geq 1/8 \cdot \edl{L_i}{Q}\}$.
        By \cref{lem:periodiccase_ed_progression}\eqref{lem:periodiccase:a} we have $t \in C'$,
        and thus or goal is to find an upper bound to
        \begin{equation}\label{eq:prob_l_bound}
            \pr{t \in C}
            = \pr{L \in \{L_i\}_{i\in I}}
            = \frac{\sum_{i\in I}\edl{L_i}{Q}}{\sum_{i=1}^{f} \edl{L_i}{Q}}
            = \frac{\sum_{i\in I}\edl{L_i}{Q}}{d}.
        \end{equation}
        To this end, we begin by observing that we have
        $L_i' \subseteq \fragmentco{\max\{0,t+\ell_i-k\}}{\min\{n,t+r_i+k\}}$
        because $\mX$ makes at most $k$ edits.
        As a consequence,
        \begin{equation*}\label{eq:li_bound}
            \textstyle{\sum_{i \notin I}} \edl{L_i'}{Q}
            \leq \textstyle{\sum_{i \notin I}} \edal{\mX_T}(L_i',Q_i)
            \leq \textstyle{\sum_{i\notin I}} k_i
            < \textstyle{\sum_{i \notin I}} 1/8\cdot \edl{L_i}{Q}.
        \end{equation*}
        Further, from the disjointess of $L_i$ follows
        $\sum_{i\notin I} \ed(L_i',L_i) = \ed(P, T\fragmentco{t}{t'})\leq k$.
        Combining this last inequality together with \eqref{eq:li_bound}
        and the triangle inequality, we obtain
        \[
        \textstyle{\sum_{i\notin I}} \edl{L_i}{Q}
        \leq \textstyle{\sum_{i\notin I}} \big( \ed(L_i',L_i) + \edl{L_i'}{Q} \big)
        \leq k + \textstyle{\sum_{i\notin I}} 1/8 \cdot \edl{L_i}{Q}
        \]
        which when rearranged yields $\sum_{i\notin I} \edl{L_i}{Q} \leq 8/7 \cdot k \leq 4/7 \cdot d$,
        concluding the proof.
    \end{claimproof}

    Therefore, by union bound, we do not fail and
    we include an arbitrary position $t \in \OccE_k(P, T)$
    to $C$ with probability at least $1 - 4/7 - 1/4 \geq 5/28$.
    By repeating $\tilde{\mathcal{O}}(1)$ times independently the procedure,
    and by returning the union of the candidate sets of each run, we
    obtain with high probability a candidate set as claimed in the statement.

    We analyze the complexity of a single run of a procedure;
    the overall quantum and query time only differs from a logarithmic factor.
    First, notice that step \eqref{it:periodiccase_ed_small_c:i}
    because of \cref{lem:find_min_approx_per} and \cref{prp:quantumed_w_info}
    requires $\Oh(\sqrt{dm})$ queries and $\Oh(\sqrt{dm}+d^2)$ time.
    From \cref{lem:locked} follows that step \eqref{it:periodiccase_ed_small_c:ii}
    requires $\Oh(d)$ time.

    Lastly, notice that we can implement \eqref{it:periodiccase_ed_small_c:v} in $\Oh(d^2)$ time.
    In order to do this, we iterate over each edit $(e,e') \in \sE_{T, Q_T}(\mX_T)$,
    and compute the blocks $\fragmentco{jd}{(j+1)d}$ that contains some $x \in C'$
    and overlap with $\fragmentco{e-r-k}{e-\ell+k}$. For each such block we put a mark on it.
    To compute $\floor{C/k}$, we return all $j$ such that $\fragmentco{jd}{(j+1)d}$ contains
    at least $1/8\cdot \edl{L}{Q}$ marks.
    The overall number of marks is counted by $\sum_e b_e$ where $b_e$ is defined as in the proof of \cref{claim:periodiccase_ed_small_c:size}.
    From $\edl{L}{Q} \leq d$ and the upper bounds by $\Oh(d)$ on $8 / \edl{L}{Q} \cdot \sum_e b_e$ from the proof of \cref{claim:periodiccase_ed_small_c:size},
    we obtain that $\sum_e b_e = \Oh(d^2)$.
    We conclude that we  never put more than $\Oh(d^2)$ marks across all iterations.
\end{proof}

\begin{lemma}
    \label{lem:ed_compstructure}
    Let $P$ denote a pattern of length $m$,
    let $T$ denote a text of length $n \leq 3/2 \cdot m$, and let $k > 0$ denote an integer threshold.
    Suppose there is a positive integer $d \geq 2k$ and a primitive string $Q$ with $|Q| \leq m/8d$ and $\edl(P,Q) \leq d$.

    Then, there exists a quantum algorithm that outputs $\floor{C/d}$, where $C$ is a set of  \emph{candidate positions} such that $\OccE_k(P,T) \subseteq C \subseteq \OccE_{22d}(P,T)$ and $|\floor{C/d}| = \Oh(d)$.
    Alternatively, provided that $\hd(P,Q^*) = d$,
    the algorithm gives the possibility to output $q \in \Z_{>0}$ and $I \subseteq \fragmentco{0}{q}$ of size $|I| = \Oh(d)$ such that $C = \{ t \in \fragment{0}{n-m-k} : t \ \mathrm{mod} \ q \in I\}$
    and $\OccE_k(P,T) \subseteq C \subseteq \OccE_{22d}(P,T)$.
    The algorithm requires $\Ohtilde(\sqrt{dm})$ queries and $\Ohtilde(\sqrt{dm}+d^2)$ time.
\end{lemma}

\begin{proof}
    Consider the following procedure.
    \begin{enumerate}[(i)]
        \item Use \cref{lem:find_min_approx_per} on $P,Q$, and $d$ in the role of $S,Q$ and $k$, respectively.
        Let $x,y$ be the indices returned by \cref{lem:find_min_approx_per} for which $\edl{P}{Q} = \ed(P,Q^{\infty}\fragmentco{x}{y})$ holds.
        \label{lem:ed_compstructure:i}
        \item Use \cref{lem:find_min_approx_per} on $T' \coloneqq T\fragmentco{n-m-k}{m-k},Q$, and $\floor{3/2 \cdot d}$ as $S,Q$ and $k$, respectively.
        \label{lem:ed_compstructure:ii}
        \item If $\edl{T'}{Q} > \floor{3/2 \cdot d}$, then output $C = \emptyset$.
        \label{lem:ed_compstructure:iii}
        \item Otherwise, consider $x',y'$ returned by \cref{lem:find_min_approx_per} in
            \eqref{lem:ed_compstructure:ii} for which $\edl{T'}{Q} =
            \ed(T',Q^{\infty}\fragmentco{x'}{y'})$ holds.
        Set $Q_{\ell} \coloneqq \rot^{-y'}(Q)$ and $Q_{r} \coloneqq \rot^{x'}(Q)$
        By combining \cref{prp:candidatepos} with an exponential/binary search
        compute the smallest $\ell$ such that $\eds{T\fragmentco{\ell}{m-k}}{Q_{\ell}} \leq \floor{3/2 \cdot d}$ holds.
        Similarly, compute the largest $r$ for which $\ed(T\fragmentco{n-m-k}{r}, Q_r^*) \leq \floor{3/2 \cdot d}$ holds.
        \label{lem:ed_compstructure:iv}
        \item Set $\bar{T} \coloneqq T\fragmentco{\ell}{r}$.
            Whenever, a candidate set $C$ of size $\floor{C/d} = \Ohtilde(k)$ is asked,
            use \cref{lem:periodiccase_ed_small_c} on $P,\bar{T},k,d$ in the role of $P,T,k,d$, respectively,
            to find such set and return it.

            Otherwise, if $\hd(P,Q^*) = d$ and $C$ is asked in the other compressed representation,
            return $|Q|$ and $\big(\fragment{\ell_P- \ell_T-3d}{\ell_P- \ell_T+3d} \ \mathrm{mod} \ |Q|\big)$
            as $q$ and $I$, respectively.
        \label{lem:ed_compstructure:v}
    \end{enumerate}

    For the proof of correctness we rely on \cite{CKW20}.
    Step \eqref{lem:ed_compstructure:i}, \eqref{lem:ed_compstructure:ii},
    \eqref{lem:ed_compstructure:iii} and \eqref{lem:ed_compstructure:iv}
    of the procedure simulate one to one \cite[Algorithm 13]{CKW20}.
    In \cite[Lemma 6.8]{CKW20} is proven that if the procedure returns at \eqref{lem:ed_compstructure:iii}, then $\OccE_k(P,T) = \emptyset$.
    Otherwise, if we proceed to \eqref{lem:ed_compstructure:iv} the authors show that $\bar{T}$
    satisfies $\eds{\bar{T}}{Q} \leq 3d$.

    From this last property of $\bar{T}$ follows that the conditions for \cref{lem:periodiccase_ed_small_c} are satisfied.
    Whenever, $\hd(P,Q^*) = d$ holds,
    the conditions for \cref{lem:periodiccase_ed_progression} are satisfied as well.
    In this last case, observe that the definitions of $C$ given in the statement of \cref{lem:periodiccase_ed_progression} and the definition with respect to $q$ and $I$
    are equivalent.
\end{proof}

\subsubsection{Finding a Candidate Set in the Three Structural Cases}
\label{sec:three_cases}

We conclude this (sub)section by providing the proofs for \cref{lem:ed_breakcase}, \cref{lem:ed_repregion}, and \cref{lem:ed_approxpercase}.

\edbreakcase

\begin{proof}
    Consider the following procedure:
    \begin{enumerate}[(i)]
        \item select uniformly at random a break $B=P\fragmentco{\beta}{\beta + |B|}$ among the $2k$ disjoint breaks $B_1, \ldots, B_{2k}$;
        \label{breakcase:i}
        \item compute $\Occ(B, T)$ using \cref{thm:quantum_matching_all}; lastly
        \label{breakcase:ii}
        \item\label{breakcase:iii} Compute and return the set
            \[C = \{x \in \fragment{0}{n} \mid \text{there exists $x' \in \Occ(B, T)$ such that $|x' - \beta -x| \leq k$}\}.\]
    \end{enumerate}

    \begin{claim} \label{claim:breakcase}
        Let $\mA : P \onto T\fragmentco{t}{t'}$ be an alignment of cost at most $k$.
        Then, with constant probability, the following two claims hold:
        \begin{enumerate}[(a)]
            \item $\edal{\mA}(P\fragmentco{\beta}{\beta + |B|}, \mA(P\fragmentco{\beta}{\beta + |B|})) = 0$; and
            \label{it:breakcase:claim:a}
            \item $t \in C$.
            \label{it:breakcase:claim:b}
        \end{enumerate}
        Moreover, $|\floor{C/k}| = \Oh(k)$.
    \end{claim}

    \begin{claimproof}
        For~\eqref{it:breakcase:claim:a}, note that the probability that no edit occurs in break $B$ is at least $k/2k = 1/2$.

        For~\eqref{it:breakcase:claim:b}, consider $x' \in \Occ(B,T)$ such that $(\beta, x') \in \mA$.
        Since $\mA$ has cost at most $k$, then $|(x' - t) - \beta| \leq k$ and $t \in C$.

        For the last part of the claim, note that $|\Occ(B,T)|\le \lceil|T|/\per(B)\rceil \le 192k = \Oh(k)$ because $\per(B) > m/128k$ and $n \leq \threehalves m$.
        By rewriting $C = \bigcup_{x' \in \Occ(B,T)} C_{x'}$, where $C_{x'} = \fragment{0}{|T|} \cap \fragment{x'-\beta-k}{x'-\beta+k}$, we obtain $|\floor{C_{x'}/k}| = \Oh(1)$ for all $x'\in \Occ(B,T)$.
        Consequently, $|\floor{C/k}| \leq \sum_{x'} |\floor{C_{x'}/k}| = \Oh(k)$.
    \end{claimproof}

    From \cref{claim:breakcase} follows that $x\in C$ holds with constant probability for an arbitrary $x \in \OccE_k(P,T)$.
    Consequently, by repeating the procedure $Oh(\log n)$ times independently and returning the union of the candidate sets, we obtain with high probability a candidate set as described in the statement of the lemma.

    The quantum time required for a single execution of the procedure is dominated by the computation of $\Occ(B,T)$ using \cref{thm:quantum_matching_all}, which takes $\Ohtilde(\sqrt{km})$ quantum time. The total quantum time for the algorithm differs from that of a single execution by only a logarithmic factor.
\end{proof}

\edrepregion

\begin{proof}
    First, we prove a useful claim.

    \begin{claim} \label{claim:ed_repregions_standard_trick}
        Let $R \in \{R_1, \ldots, R_{r}\}$ be a repetitive region. Then, we can compute a candidate set $\OccE_{\kappa}(R, T) \subseteq C_R$ such that $|\floor{C_R/d}| = \Ohtilde(k)$, where $d = \edl{R}{Q}$ and $\kappa = \floor{4k/m \cdot |R|}$.
        This requires $\Ohtilde(\sqrt{km})$ queries and $\Ohtilde(\sqrt{km}+k^2)$ time.
    \end{claim}

    \begin{claimproof}
        We first show how to prove \cref{claim:ed_repregions_standard_trick} under two assumptions:
        \begin{enumerate}[(a)]
            \item we have an upper bound $u$ on $\kappa$, and
            \item we have a quantum algorithm that, given $T' \coloneqq T\fragment{t}{t'} \substr T$ of length $|T'| \leq 5/4 \cdot |R|+u$, outputs $\floor{C/d}$ where $C$ is such that $\OccE_{\kappa}(R,T') \subseteq C \subseteq \fragment{0}{n}$
            and $|\floor{C/d}| = \Ohtilde(d)$.
            The algorithm must do this within $\Oh(\sqrt{d|R|})$ queries and $\Oh(\sqrt{d|R|}+d^2)$ time.
            \label{asmp:ed_repregions_standard_trick:b}
        \end{enumerate}

        We combine the Standard Trick with the algorithm from \eqref{asmp:ed_repregions_standard_trick:b}.
        Divide $T$ into $\Oh(m/|R|)$ contiguous blocks of length $|R|/4$ (the last block might be shorter),
        and iterate over all blocks.
        When iterating over the $i$-th block of the form $T\fragmentco{x_i}{\min(x_i + |R|/4, n)}$,
        consider the segment $S_i = T\fragmentco{x_i}{\min(x_i + 5/4 \cdot |R|+u, n)}$ and use \cref{lem:ed_compstructure} on the algorithm from \eqref{asmp:ed_repregions_standard_trick:b} on $R$, $S_i$, $\kappa$, $d$ to retrieve a set $C_i$. Importantly, as $\kappa \leq u$, each $\kappa$-error occurrence is fully contained in at least one $S_i$.

        After processing all blocks, combine the candidate sets from each block to form
        $C_R = \bigcup_{i=1}^{b} C_i$ ()
        of size $|\floor{C_R/d}| = \Ohtilde(m/|R| \cdot d) = \Ohtilde(k)$.
        This requires $\Ohtilde(m/|R| \cdot \sqrt{d|R|}) = \Ohtilde(m/|R| \cdot \sqrt{k/m} \cdot |R|) = \Ohtilde(\sqrt{km})$ queries and
        $\Ohtilde(m/|R| \cdot (\sqrt{d|R|} + d^2)) = \Ohtilde(\sqrt{km}+k^2)$ time.

        To finish the proof, we distinguish two regimes of $\kappa$ where we can satisfy the two assumptions:
        \begin{itemize}
            \item If $\kappa \leq |R|/4$, we set $u = |R|/4$ and use as algorithm \cref{lem:ed_compstructure} on $R,T',\kappa,d$
            in the role of $P,T,k$ and $d$, respectively. As $d = \hd(R,Q^*)$, we can choose to retrieve a candidate set $C_i$ of bounded size.
            \item If $|R|/4 \leq \kappa \leq |R|$, we set $u = |R|/4$ and an algorithm
            that returns $C = \fragment{t}{t'}$.
            As $d \geq \kappa \geq |R|/4$, the set $C$ satisfies $|C| = \Oh(|R|) = \Oh(d)$.
            \claimqedhere
        \end{itemize}
    \end{claimproof}

    Next, consider the following procedure:
    \begin{enumerate}[(i)]
        \item Select a repetitive region $R=P\fragmentco{\rho}{\rho + |R|}$ among $R_1, \ldots, R_r$ with probability proportional to their length, that is,
        \[
            \pr{R = R_i} = \frac{|R_i|}{\sum_{i'=1}^{r} |R_{i'}|} \quad \text{for $i \in \fragment{1}{r}$.}
        \]
        \label{it:approxpercase:i}
        \item Use \cref{claim:ed_repregions_standard_trick} to compute $\floor{C_R/d}$ for a set $\OccE_{\kappa}(R, T)\subseteq C_R \subseteq \fragment{0}{n}$ such that $|\floor{C_R/d}|=\Ohtilde(k)$,
        where $\kappa = \floor{4k/m \cdot |R|}$ and $d = \edl{R}{Q}$.
        \label{it:approxpercase:ii}
        \item Compute and return the set
        \[
            C = \{x \in \fragment{0}{|T|} \mid \text{there exists $x' \in d \cdot \floor{C_R/d}$ such that $|x' - \rho - x| \leq 10k$}\}.
        \]
        \label{it:approxpercase:iv}
    \end{enumerate}

    \begin{claim} \label{approxpercase:claim}
        Let $\mA : P \onto T\fragmentco{t}{t'}$ be an alignment of cost at most $k$.
        Then, with constant probability, the following two claims hold:
        \begin{enumerate}[(a)]
            \item $\edal{\mA}(R, \mA(R)) \leq \lfloor 4k/m \cdot |R| \rfloor$; and
            \label{it:approxpercase:claim:a}
            \item $t \in C$.
            \label{it:approxpercase:claim:b}
        \end{enumerate}
        Moreover, $|\floor{C/k}| = \Ohtilde(k)$.
    \end{claim}

    \begin{claimproof}
        For~\eqref{it:approxpercase:claim:a}, define the index set $I \coloneqq \{i \in \fragment{1}{r} : \edal{\mA}(R_i, \mA(R_i)) \leq \lfloor 4k/m \cdot |R_i| \rfloor\}$.
        Observe that
        \[
            k \ge \edal{\mA}(P, T\fragmentco{t}{t'})
            \geq \sum_{i \in I} \edal{\mA}(R_i, \mA(R_i))
            \geq \sum_{i \notin I} 4k/m \cdot |R_i|
            \geq 4k/m \cdot \sum_{i \notin I} |R_i|,
        \]
        that is, $\sum_{i\notin I} |R_i| \le m/4$.
        At the same time, $\sum_{i=1}^r |R_i| \ge 3/8\cdot m$, so $\sum_{i\in I} |R_i| \ge (3/8-1/4)\cdot m \ge m/8$.
        It follows directly that
        $\pr{\edal{\mA}(R, \mA(R)) \leq \lfloor 4k/m \cdot |R| \rfloor} = \sum_{i \in I} |R_i| / \sum_{i=1}^{r} |R_i|
        \geq (m/8) / m \geq \Omega(1)$, concluding the proof for~\eqref{it:approxpercase:claim:a}.

        For~\eqref{it:approxpercase:claim:b} let $t_R$ be such that $(\rho, t_R) \in \mA$
        for which $t_R \in \OccE_{\kappa}(R,T) \subseteq C_R$ holds.
        From $|(t_R - t) - \rho| \leq k$ together with the triangle inequality follows
        $|t - (d \floor{t_R/d} - \rho)| \leq |d \floor{t_R/d} - t_R| + |(t_R - t) - \rho| \leq d + k \leq 10k$.
        Hence, for $x' = d\floor{t_R/d}$ we have $|x' - \rho - t| \leq 10k$ and $t \in C$ holds.

        For the last part of the claim,
        we rewrite
        $C = \bigcup_{x' \in d \cdot \floor{C_R/d}} C_{x'}$,
        where $C_{x'} = \{x \in \fragment{0}{|T|} \mid |x' - \rho - x| \leq 10k\}$, we obtain
        $|\floor{C_{x'}/k}| = \Oh(1)$ for all $x'\in d \cdot \floor{C_R/d}$.
        From $|\floor{C_R/d}| = \Ohtilde(k)$ follows that $|\floor{C/k}| \leq \sum_{x'} |\floor{C_{x'}/k}| = \Ohtilde(k)$.
    \end{claimproof}

    From \cref{approxpercase:claim} follows that $x\in C$ holds with constant probability for an arbitrary $x \in \OccE_k(P,T)$.
    Consequently, by repeating the procedure $\Oh(\log n)$ times independently and returning the union of the candidate sets, we obtain with high probability a candidate set as described in the statement of the lemma.

    The quantum time required for a single execution of the procedure is dominated by the computation of $C_R$ using \cref{claim:ed_repregions_standard_trick}, which takes $\Ohtilde(\sqrt{km})$ queries and $\Ohtilde(\sqrt{km}+k^2)$ time.
    The total quantum time for the algorithm differs from that of a single execution by only a logarithmic factor.
\end{proof}

\edapproxpercase

\begin{proof}
    We apply \cref{lem:ed_compstructure} to $P$, $T$, $k$, and $\max(\hd(P, Q^*), 2k)$, using the latter in place of $d$. If $\hd(P, Q^*) \geq 2k$, we set $d = \hd(P, Q^*)$ and use the algorithm from \cref{lem:ed_compstructure} to obtain $q$ and $I$.

    If, on the other hand, $\hd(P, Q^*) < 2k$, the algorithm from \cref{lem:ed_compstructure} returns a set $C'$ such that $\OccE_k(P, T) \subseteq C'$ and $|C' / 2k| = \Ohtilde(k)$. However, we note that dividing by $2k$ produces an incorrect result. Therefore, instead of returning $\floor{C'/2k}$, we return $\bigcup_{c \in C'} \{2c, 2c + 1\}$ as $\floor{C/k}$.
    Clearly, $|\floor{C/k}| = \Ohtilde(k)$. Additionally, since $\OccE_k(P, T) \subseteq C'$, for every $t \in \OccE_k(P, T)$, there exists some $c \in C'$ such that $t \in [c2k, (c+1)2k)$, which implies $\floor{t/k} \in C$.
\end{proof}

\bibliographystyle{alphaurl}
\bibliography{main}

\end{document}